\documentclass[11pt]{article}

\usepackage[margin=1in]{geometry}
\usepackage{amsmath,amssymb,amsfonts}
\usepackage{amsthm}
\usepackage{mathtools}
\usepackage{graphicx}
\usepackage{enumitem}
\usepackage{natbib}
\usepackage{hyperref}
\usepackage{mathrsfs}

\hypersetup{
    colorlinks=true,
    linkcolor=blue,
    urlcolor=blue,
    citecolor=blue
}

\title{Superloop Equations and Minimal Surfaces I:\\Confining minimal surface in \(4D,\mathcal N=1\) SYM 
}
\author{Alexander Migdal}
\date{July 17, 2026}

\newcommand{\Tr}{\operatorname{Tr}}

\newcommand{\Mathematica}{Wolfram Mathematica}
\newtheorem{theorem}{\bf Theorem}[section]

\newtheorem{lemma}{\bf Lemma}[section]
\newtheorem{remark}{\bf Remark}[section]
\newtheorem{proposition}{\bf Proposition}[section]
\newtheorem{definition}{\bf Definition}[section]

\begin{document}

\maketitle

\begin{abstract}
\begin{abstract}

We formulate the loop equations of pure
\(4d,\mathcal N=1\) super Yang--Mills (SYM) theory in a finite geometric
form.  The usual equal-point loop derivatives are singular because the
order of gauge-field insertions is lost when two points of the contour
are identified.  Our path-ordered operator calculus (POOC) keeps these
insertions in distinct ordered slots, forms the required graded
commutators and Jacobi combinations, and only then takes the coincidence
limit.  This removes the spurious kinematical singularities without
introducing a cutoff or an auxiliary scale.  The genuine ultraviolet
contact term remains a separate physical distribution.

We apply POOC to the exact Itoyama--Takashino superloop hierarchy and
construct a Lorentzian supersymmetric Hodge-dual (SHD) surface
functional.  Its Hodge-resolved area derivative is annihilated by the
local chiral and anti-chiral loop operators.  Exact additivity and
surface-geodesic sewing show that its exponential multiplies any
solution of the complete finite-\(N\) hierarchy without changing the
equations.  For planar contours the SHD functional reduces to the
geometric area.  In particular, for a long rectangular Wilson loop,
\[
    W[C_{T,L}]
    \sim
    \exp[-i\sigma_kLT],
    \qquad
    E_k(L)=\sigma_kL>0.
\]
The construction therefore gives an exact nonperturbative confining
area-law factor in \(\mathcal N=1\) SYM.

After Euclidean continuation, the same SHD factor is an exact fixed
point of the deterministic zero-noise SYM gradient flow.  Its
interpretation as a dynamically selected equilibrium still requires
stability in the long-flow-time, large-volume, and zero-noise limits.
This remaining dynamical question does not affect the exact POOC
zero-mode and finite-\(N\) dressing results.  The complete planar
Wilson-loop solution, including the undressed fluctuation factor and
the spectrum of excitations, will be developed in the next papers of
this series.

\end{abstract}
\end{abstract}

\section{Introduction}
\label{sec:Introduction}

Loop space is the natural language of a gauge theory. The Makeenko--Migdal equation \cite{MMEq79,MM1981NPB,Mig83} is the gauge-invariant Schwinger--Dyson equation written directly for Wilson loops. At finite \(N\) it is an infinite hierarchy: one loop may split at a self-contact, two different loops may join at a common point, and the finite-\(N\) color algebra reconnects the traces. In the planar limit the normalized multiloop correlators factorize and the hierarchy closes on the one-loop expectation value. This is why confinement has to be tested in the complete loop hierarchy. An area-law factor must survive every splitting, joining, and finite-\(N\) reconnection, not merely the planar one-loop equation.

There is, however, an old ambiguity in the formal loop calculus. Equal-point functional derivatives mix two different things: the physical ultraviolet contact and the kinematical ordering of operator insertions. The ordering must be fixed before the points are brought together. Treating this ambiguity as an ultraviolet divergence obscures the simple operator algebra of the loop equation.

The first purpose is to supersymmetrize the \emph{path-ordered operator calculus} (POOC). Its bosonic version was developed for the Makeenko--Migdal hierarchy in the Geometric QCD series \cite{migdal2025geometric,Migdal2026GeometricQCDII,Migdal2026GeometricQCDIII}. We place the tangent derivatives in distinct ordered slots, form their commutators, nested commutators, and Jacobi combinations at finite separation, and only then take the ordered coincidence limit. No auxiliary cutoff is introduced and no scale remains. This gives a direct operator definition of coincident loop differentiation and separates the loop algebra from the physical ultraviolet regularization.

We apply this finite calculus to the exact Schwinger--Dyson hierarchy for pure \(4d,\mathcal N=1\) super Yang--Mills theory derived by Itoyama and Takashino in 1996 \cite{ItoyamaTakashino1996,ItoyamaTakashino1997}. Their basic variable is the superspace Wilson loop, and their finite-\(N\) equations have the same splitting-and-joining structure as the ordinary Makeenko--Migdal hierarchy. A related superloop formulation of ordinary large-\(N\) QCD was developed at the same time in Ref.~\cite{Mig98Hidden}. There a local one-dimensional supersymmetry of a Grassmann-valued momentum loop combined the component loop equations into one object. That theory was QCD rather than a supersymmetric gauge theory, but the kinematical idea is the same and is used below.

The second purpose is to construct a Lorentzian supersymmetric Hodge-dual (SHD) finite-Stokes zero mode and prove that its phase is compatible with the complete finite-\(N\) hierarchy. This gives an exact four-dimensional example of gauge-string duality. The confining vacuum is described not by a fluctuating Polyakov worldsheet, but by a rigid master surface selected by the loop equation. Lorentzian signature is essential. Since \((*_4)^2=-1\), the two Hodge sectors are complex, while their physical combination must be a real Wilson-loop phase. For a long timelike rectangle of spatial width \(L\) and duration \(T\), the physical branch must obey
\begin{equation}
    W(T,L)
    \sim
    \exp[-iT E(L)],
    \qquad
    E(L)=\sigma L,
    \qquad
    \sigma>0.
    \label{eq:IntroductionStaticEnergyRequirement}
\end{equation}
Thus positivity refers to the Hamiltonian energy extracted from the Lorentzian Wilson loop, not to either complex Hodge component separately.  The corresponding string tension characterizes the selected nonperturbative vacuum branch.

\begin{remark}{\textbf{AI and \Mathematica{} as technical assistants}}

The equations of SUSY field theories are geometrically transparent and beautiful, but they are bulky when explicitly written in components. A human reader would be distracted when forced to follow the lines of Grassmann algebra with multiple tensor indices and implied conventions. The author could easily make mistakes in signs and coefficients of these equations. This is why the author employed AI assistants (GPT 5.6 Pro, and GPT 5.6 Sol Codex) to fix the coefficients and signs in the algebraic equations and also to verify all the main equations by creating \Mathematica{} notebooks where the Grassmann algebra was implemented and equations analytically verified. With this two-layer workflow, there is now a documented verification of the equations of this paper. While the author is responsible for the physical and geometric ideas and overall architecture, the analytic computations were outsourced to AI + \Mathematica{} and presented in Supplementary material. The readers are welcome to employ their AI to read and validate the supplementary material and \Mathematica{} files.
\end{remark}

\subsection*{Path-ordered superloop operators}

We use throughout the Lorentzian superspace notation and conventions
of Itoyama and Takashino. A closed superloop is a periodic map
\begin{equation}
    C:
    s\longmapsto
    z^M(s)
    =
    \left(
        x^a(s),
        \eta^\alpha(s),
        \bar\eta^{\dot\alpha}(s)
    \right),
    \qquad
    z^M(s+2\pi)=z^M(s).
    \label{eq:IntroductionSuperloop}
\end{equation}
The coordinate superloop is continuous and closed.  Its coordinate
and invariant tangents are allowed to be periodic functions of bounded
variation on the oriented circle.  Periodicity in this statement is
periodicity of the BV distribution; after choosing a marked cut it
does \emph{not} identify the two one-sided traces across that cut.  In
particular, \(E^A(2\pi-0)\) and \(E^A(0+0)\) may be different.  These
two traces are the ordered data used by POOC.  The BV loop class and
its integer Fourier representation are stated precisely in
Subsection~\ref{subsec:FlatN1SuperspaceAndSuperloops}.

The invariant one-form basis is denoted by \(e^A\), and \(D_A\) is
the dual flat superspace derivative.  For homogeneous superspace
indices \(A,B\), the constant flat-superspace torsion coefficients
\(T_{AB}{}^{C}\) are defined by the graded derivative algebra
\begin{equation}
    [D_A,D_B\}
    \equiv
    D_AD_B
    -
    (-1)^{|A||B|}
    D_BD_A
    =
    T_{AB}{}^{C}D_C .
    \label{eq:IntroFlatSuperspaceAlgebra}
\end{equation}
They obey
\[
    T_{AB}{}^{C}
    =
    -
    (-1)^{|A||B|}
    T_{BA}{}^{C}.
\]
In flat \(4d,\mathcal N=1\) superspace the only nonvanishing
components are
\[
    T_{\alpha\dot\beta}{}^{a}
    =
    T_{\dot\beta\alpha}{}^{a}
    =
    -2i\sigma^{a}_{\alpha\dot\beta}.
\]
Its pullback defines the invariant supertangent \(E^A(s)\):
\begin{equation}
    e^A\big|_C
    =
    ds\,E^A(s),
    \qquad
    \boxed{
    E^A(s)D_A
    =
    \frac{d}{ds}
    }.
    \label{eq:IntroductionInvariantTangentIdentity}
\end{equation}
The vector prepotential is denoted by \(V\); the symbol \(E^A\) is
reserved for the contour tangent.

With the Itoyama--Takashino connection convention,
\begin{equation}
    \nabla_A=D_A-A_A,
    \qquad
    [\nabla_A,\nabla_B\}
    =
    T_{AB}{}^C\nabla_C-F_{AB}.
    \label{eq:IntroductionCovariantDerivativeAlgebra}
\end{equation}
The ordinary Wilson transporter performs gauge parallel transport.
For the operator calculus we lift it by a flat backward translation,
so that the infinitesimal generator of the lifted transporter is
\(-E^A\nabla_A\). It is therefore natural to define the tangent
operator by
\begin{equation}
    \boxed{
    \delta_A(s)
    \equiv
    -\frac{\delta}{\delta E^A(s)}
    }.
    \label{eq:IntroductionTangentDerivative}
\end{equation}
This sign makes \(\delta_A(s)\) insert \(+\nabla_A\) at the marked
point.

For homogeneous tangent operators, the path-ordered graded
commutator is
\begin{equation}
\begin{aligned}
    [\delta_A,\delta_B\}_o(s)
    \equiv
    \lim_{\epsilon\downarrow0}
    \Bigl[
        &\delta_A(s-\epsilon)\delta_B(s+\epsilon)
        \\
        &-
        (-1)^{|A||B|}
        \delta_B(s-\epsilon)\delta_A(s+\epsilon)
    \Bigr].
\end{aligned}
    \label{eq:IntroductionOrderedCommutator}
\end{equation}
This ordered commutator inserts the complete covariant derivative
commutator, which contains both the known flat superspace torsion and
the gauge supercurvature. With the native I--T signs, the genuine
area derivative is therefore  defined as
\begin{equation}
    \boxed{
    \delta^\Sigma_{AB}(s)
    =
    T_{AB}{}^C\delta_C(s)
    -
    [\delta_A,\delta_B\}_o(s).
    }
    \label{eq:IntroductionAreaDerivative}
\end{equation}
Acting on the lifted transporter, this finite path-ordered operation
inserts \(F_{AB}(z(s))\) exactly.

For a homogeneous adjoint superfield \(X\), define its marked,
parallel-transported insertion by
\[
    \mathfrak I_s(X)
    =
    U(2\pi,s)\,
    X\bigl(z(s)\bigr)\,
    U(s,0).
\]
The ordered adjoint derivative is defined by placing the tangent
operation immediately before and immediately after the marked
insertion:
\begin{equation}
\begin{aligned}
    \mathfrak D_A^{\,o}(s)\,
    \mathfrak I_s(X)
    &\equiv
    \delta_A(s-0)\,
    \mathfrak I_s(X)
    -
    (-1)^{|A||X|}
    \mathfrak I_s(X)\,
    \delta_A(s+0)
    \\
    &=
    \mathfrak I_s\!\left(\mathcal D_AX\right),
    \qquad
    \mathcal D_AX
    =
    [\nabla_A,X\}
    =
    D_AX-[A_A,X\}.
    \label{eq:IntroOrderedAdjointDerivative}
\end{aligned}
\end{equation}
Here the two occurrences of \(\delta_A\) occupy distinct ordered
slots on the two sides of \(X\); the ordered coincidence limit is
taken only after their graded difference has been formed.
\begin{equation}
    \mathfrak D_A^{o}(s)\mathfrak I_s(X)
    =
    \mathfrak I_s(\mathcal D_A X),
    \qquad
    \mathcal D_A X=[\nabla_A,X\}.
    \label{eq:IntroductionOrderedAdjointDerivative}
\end{equation}
Thus POOC represents both the curvature insertion and its covariant
superspace derivative directly in operator form. Three tangent operations are kept in three
distinct ordered slots until their graded Jacobi combination has been
formed. Associativity at finite separation gives the exact graded
Jacobi identity; after torsion subtraction, this is the torsionful
superspace Bianchi identity. This identity is what gives the zero mode.

\subsection*{Parametric invariance}

The traditional notation \(W[C(s)]\) hides a simple fact: the
parameter \(s\) is not physical.  It is only a coordinate on the
oriented boundary circle.  Let
\begin{equation}
    \mathcal P_{\rm BV}
    =
    \left\{
        z:S^1\to\mathbb R^{4|4}
        \ \middle|\
        z\in C^0(S^1),\quad
        \dot z^M,\ E^A\in BV(S^1),\quad
        z(s+2\pi)=z(s)
    \right\}
    \label{eq:IntroductionParametrizedBVLoopSpace}
\end{equation}
be the computational space of parametrized superloops.  The physical
loop space is the quotient
\begin{equation}
    \boxed{
    \mathcal L_{\rm BV}
    =
    \mathcal P_{\rm BV}/\operatorname{Diff}^{+}(S^1).
    }
    \label{eq:IntroductionPhysicalBVLoopSpace}
\end{equation}
A point of this quotient retains the orientation, cyclic ordering,
multiple covering, winding number, and the distinct preimages of a
self-intersection.  It does not retain the numerical parameter used to
label these data.  A Wilson functional is therefore a function of the
orbit \([C]\), or equivalently an invariant functional on the
parametrized loop space,
\begin{equation}
    W[C\circ\varphi]=W[C],
    \qquad
    \varphi\in\operatorname{Diff}^{+}(S^1).
    \label{eq:IntroductionWilsonReparametrization}
\end{equation}

For the BV loop class, the natural reparametrization group is the
BV completion of the smooth orientation-preserving diffeomorphisms.
By a BV reparametrization we mean that the map itself is continuous,
strictly increasing, absolutely continuous, and bi-Lipschitz, while
its derivative may be a BV function:
\begin{equation}
\begin{aligned}
    \operatorname{Diff}^{+}_{\rm BV}(S^1)
    =\Bigl\{\varphi:\;&
    \varphi(u+2\pi)=\varphi(u)+2\pi,\quad
    \varphi\in W^{1,\infty},\quad
    \varphi'\in BV(S^1),
    \\
    &0<m\leq\varphi'(u)\leq M<\infty
    \quad\text{a.e.}\Bigr\}.
    \label{eq:IntroductionBVReparametrizationGroup}
\end{aligned}
\end{equation}
Thus ``BV'' refers to \(\varphi'\), not to a discontinuous map
\(\varphi\).  Smooth diffeomorphisms form a dense subgroup, and the
POOC identities extend to this completion by continuity.  More
singular monotone maps, with a singular measure in \(d\varphi\), are
not included because they would turn a BV tangent into a measure.

This larger group preserves the geometric meaning of a cusp.  If
\(s_0=\varphi(u_0)\) and the loop has nonzero one-sided tangent traces
\(v_\pm=C'(s_0\pm0)\), then
\begin{equation}
    \widetilde v_\pm
    =
    \varphi'(u_0\pm0)v_\pm.
    \label{eq:IntroductionCuspTraceTransformation}
\end{equation}
The two factors \(\varphi'(u_0\pm0)\) may be different, but they are
positive.  Hence the directions of the two one-sided tangents, and therefore
the Euclidean cusp angle or Lorentzian cusp rapidity, are unchanged.
A jump of \(\varphi'\) at a geometrically smooth point creates only a
jump of speed, not a cusp.  A geometric cusp is characterized by
\begin{equation}
    v_+\notin\mathbb R_{>0}v_-,
    \label{eq:IntroductionGeometricCuspCriterion}
\end{equation}
not merely by \(v_+\neq v_-\).

The advantage of tangent differentiation is now manifest.  Let
\begin{equation}
    s=\varphi(u),
    \qquad
    \varphi'(u)>0,
    \label{eq:IntroductionOrientationPreservingReparametrization}
\end{equation}
and write the same superloop in the new parameter as
\begin{equation}
    \widetilde z^M(u)=z^M(\varphi(u)).
    \label{eq:IntroductionReparametrizedSuperloop}
\end{equation}
Since \(e^A|_C\) is an invariant one-form,
\begin{equation}
    du\,\widetilde E^A(u)
    =
    ds\,E^A(s),
    \qquad
    \widetilde E^A(u)
    =
    \varphi'(u)E^A(\varphi(u)).
    \label{eq:IntroductionTangentReparametrization}
\end{equation}
The tangent has reparametrization weight one.  In the functional chain
rule, however, this Jacobian is canceled by the Jacobian of the
functional delta-function.  For every loop functional \(F\),
\begin{equation}
    \boxed{
    \frac{\delta F}
         {\delta\widetilde E^A(u)}
    =
    \left.
    \frac{\delta F}
         {\delta E^A(s)}
    \right|_{s=\varphi(u)} .
    }
    \label{eq:IntroductionTangentDerivativeWeightZero}
\end{equation}
Consequently, the POOC tangent operation has weight zero,
\begin{equation}
    \boxed{
    \widetilde\delta_A(u)
    =
    -\frac{\delta}{\delta\widetilde E^A(u)}
    =
    \delta_A(\varphi(u)).
    }
    \label{eq:IntroductionPOOCWeightZero}
\end{equation}
This is very different from the ordinary coordinate-path derivative.
Since \(z^M\) is a scalar,
\begin{equation}
    \frac{\delta F}
         {\delta\widetilde z^M(u)}
    =
    \varphi'(u)
    \left.
    \frac{\delta F}
         {\delta z^M(s)}
    \right|_{s=\varphi(u)} ,
    \label{eq:IntroductionCoordinateDerivativeWeightOne}
\end{equation}
so the coordinate derivative is a one-dimensional density of weight
one.

An orientation-preserving reparametrization preserves the cyclic order
and the two one-sided traces,
\[
    u-0\longmapsto \varphi(u)-0,
    \qquad
    u+0\longmapsto \varphi(u)+0.
\]
It follows that every local POOC operation is a reparametrization
scalar before any contour integration is performed:
\begin{equation}
\begin{aligned}
    [\widetilde\delta_A,\widetilde\delta_B\}_o(u)
    &=
    [\delta_A,\delta_B\}_o(\varphi(u)),
    \\
    \widetilde\delta^\Sigma_{AB}(u)
    &=
    \delta^\Sigma_{AB}(\varphi(u)),
    \\
    \widetilde{\mathfrak D}^{\,o}_A(u)
    &=
    \mathfrak D^{\,o}_A(\varphi(u)),
    \\
    \widetilde{\widehat{\mathbb L}}_\pm(u)
    &=
    \widehat{\mathbb L}_\pm(\varphi(u)).
\end{aligned}
    \label{eq:IntroductionPOOCParametricInvariance}
\end{equation}
Thus POOC is not merely covariant under a change of parameter.  Its
local operators are already defined on
\(\mathcal L_{\rm BV}\).

This observation also explains the old singularity of the coordinate
loop-space Laplacian.  A second coordinate derivative carries one
reparametrization weight at each of its two points.
For a local functional its diagonal singularity has the schematic form
\begin{equation}
    \frac{\delta^2F}
    {\delta z^M(s)\delta z_M(t)}
    =
    \delta(s-t)\,\mathcal K_F(s)
    +
    \partial_s\delta(s-t)\,\mathcal K_F^{(1)}(s)
    +\cdots .
    \label{eq:IntroductionCoordinateHessianDistribution}
\end{equation}
The coefficient \(\mathcal K_F(s)\) has weight one, not weight zero.
It may have a finite component in a chosen parametrization, but it has
no finite reparametrization-independent value at a point of loop space.  Its
meaning is distributional and appears only after pairing with an
object of complementary weight or after introducing an auxiliary
one-dimensional metric.  The formal equal-point loop Laplacian
therefore mixed a genuine ultraviolet contact with this purely
kinematical density-weight obstruction.  POOC removes the latter.  The
only distributional object left in the loop equation is the physical
ultraviolet contact kernel on its right-hand side.

A solution is therefore defined on the same space of loops modulo reparametrization.  The exact
hierarchy is an equation for
\begin{equation}
    W_n:
    \mathcal L_{\rm BV}^{\,n}\longrightarrow\mathbb C,
    \label{eq:IntroductionSolutionOnLoopQuotient}
\end{equation}
although any convenient parametrized representative may be used in a
calculation.  The integrated contact is parametrically invariant:
\(D_f(t)\) has weight one, \(dt\,D_f(t)\) is a
one-form, and the target-superspace delta-function depends only on the
positions.  At a self-contact, the numbers \(s,t\) merely label two
ordered preimages; after splitting, each daughter loop may be placed
on its own standard boundary circle.

Physical quark-loop observables have precisely the corresponding
quotient measure,
\begin{equation}
    \int
    \frac{\mathcal DC\,\mathcal DP}
         {\operatorname{Vol}(\operatorname{Diff}^{+}(S^1))}
    \,\cdots
    =
    \int_{\mathcal L_{\rm BV}}d\mu([C])\,\cdots ,
    \label{eq:IntroductionPathIntegralDiffQuotient}
\end{equation}
see Refs.~\cite{Migdal2026GeometricQCDII,Migdal2026GeometricQCDIII}.
It is therefore sufficient to construct one representative of every
loop orbit used in the solution.  The Super--Weierstrass variables do
this directly,
\begin{equation}
    (\lambda,\widetilde\lambda,\rho,\widetilde\rho)
    \longrightarrow
    (f',\widetilde f')
    \longrightarrow
    C'(\theta)
    \longrightarrow
    C(\theta)
    \longrightarrow
    [C].
    \label{eq:IntroductionTwistorGaugeOnLoopSpace}
\end{equation}
No inverse Douglas problem and no minimization over an unknown boundary
reparametrization are required.  BV tangent jumps are admissible
ordered loop data, not ultraviolet divergences of the finite-Stokes
sector.  The detailed weight calculation, the invariance of the
contact pairing, and the relation to the twistor gauge are given in
Supplementary Material, Subsec.~S1.1.

\subsection*{The exact Itoyama--Takashino hierarchy in POOC form}

Itoyama and Takashino obtain their closed equation by a nontrivial
Wess--Zumino-gauge derivation. The Wess--Zumino prepotential is
Grassmann-nilpotent, so its component expansion terminates. They
construct a field-dependent even vector field on the component space
\((v_a,\lambda_\alpha,\bar\lambda_{\dot\alpha},D)\) whose action on
the action gives the complete covariant chiral equation-of-motion
superfield. Its triangular field dependence makes its local graded
functional divergence vanish. When the same variation acts on the
Wilson line, it produces a nontrivial connection projector. The
explicit component fields in this projector can be rewritten exactly as
endpoint derivatives and mixed superspace area derivatives on the two
open transporter segments. Finally, gauge-unfixing converts the
Wess--Zumino contact expression into a manifestly supersymmetric
kernel; their difference is a total derivative in the integrated
contour point and vanishes by the restricted-supergauge Ward identity.
The marked point remains arbitrary throughout. The complete
component derivation is reconstructed in Supplementary Material,
Sec.~S4.

The resulting local chiral and anti-chiral operators take the
POOC forms
\begin{equation}
    \boxed{
    \widehat{\mathbb L}_+(s)
    =
    \frac18
    \epsilon_{\alpha\beta}
    \bar\sigma^{a\dot\alpha\beta}
    \mathfrak D^{o\,\alpha}(s)
    \delta^\Sigma_{a\dot\alpha}(s),
    }
    \label{eq:IntroductionChiralOperator}
\end{equation}
\begin{equation}
    \boxed{
    \widehat{\mathbb L}_-(s)
    =
    \frac18
    \epsilon_{\dot\alpha\dot\beta}
    \sigma^{a\alpha\dot\beta}
    \mathfrak D^{o\,\dot\alpha}(s)
    \delta^\Sigma_{a\alpha}(s).
    }
    \label{eq:IntroductionAntiChiralOperator}
\end{equation}
Since vector--spinor torsion vanishes,
\begin{equation}
    \delta^\Sigma_{a\dot\alpha}
    =
    -[\delta_a,\delta_{\dot\alpha}\}_o,
    \qquad
    \delta^\Sigma_{a\alpha}
    =
    -[\delta_a,\delta_\alpha\}_o,
    \label{eq:IntroductionMixedAreaDerivative}
\end{equation}
so each local operator acts as a triple path-ordered graded commutator with
the sign fixed by the I--T connection convention.

Although all explicit Lorentz indices in
\(\widehat{\mathbb L}_\pm\) are contracted, the resulting loop
equations are not one-component scalar equations. They are
Grassmann-even superfield equations at the arbitrary marked point
\(z(s)\). Their finite expansions in \(\eta(s)\) and
\(\bar\eta(s)\) package the component Schwinger--Dyson equations of
the entire gauge supermultiplet, mirroring the  superloop
formulation of large-\(N\) QCD \cite{Mig98Hidden}.
The flat spinor derivatives appearing in the contact kernel act at
the integrated superspace point \(z(t)\).  To distinguish evaluation
at that point from differentiation with respect to the contour
parameter \(t\), we define
\begin{equation}
    D_\alpha[t]\,\Phi\bigl(z(t)\bigr)
    \equiv
    \left.
        D_\alpha\Phi(z)
    \right|_{z=z(t)},
    \qquad
    \bar D_{\dot\alpha}[t]\,\Phi\bigl(z(t)\bigr)
    \equiv
    \left.
        \bar D_{\dot\alpha}\Phi(z)
    \right|_{z=z(t)} .
    \label{eq:IntegratedPointSpinorDerivatives}
\end{equation}
In chiral coordinates these operators are
\begin{equation}
    D_\alpha[t]
    =
    \left.
    \left(
        \frac{\partial}{\partial\eta^\alpha}
        +
        2i
        (\sigma^a\bar\eta)_\alpha
        \frac{\partial}{\partial y^a}
    \right)
    \right|_{z=z(t)},
    \qquad
    \bar D_{\dot\alpha}[t]
    =
    -
    \left.
        \frac{\partial}{\partial\bar\eta^{\dot\alpha}}
    \right|_{z=z(t)} .
    \label{eq:IntegratedPointSpinorDerivativesChiral}
\end{equation}
The brackets \([t]\) denote evaluation at \(z(t)\); they do not denote
a derivative with respect to \(t\).
The spinor derivatives act distributionally on the complete contact
pairing to their right, including the endpoint dependence of the
split-loop factors.
The chiral contact distribution is built from the one-form-valued line
operator
\begin{equation}
    D_f^\alpha
    =
    e^\alpha
    -
    \frac{i}{4}
    e^a\bar\sigma_a^{\dot\alpha\alpha}
    \bar D_{\dot\alpha},
    \label{eq:IntroductionChiralLineOperator}
\end{equation}
and its pullback to the integrated point \(z(t)\). 
The chiral differentiated contact kernel is therefore
\begin{equation}
    K_+(s,t;C)
    =
    D_f^\alpha(t)\,
    D_\alpha[t]\,
    \delta^{4|4}\!
    \left(
        z(t)-z(s)
    \right),
    \label{eq:ChiralContactKernel}
\end{equation}
and the anti-chiral kernel is
\begin{equation}
    K_-(s,t;C)
    =
    \bar D_f^{\dot\alpha}(t)\,
    \bar D_{\dot\alpha}[t]\,
    \delta^{4|4}\!
    \left(
        z(t)-z(s)
    \right).
    \label{eq:AntiChiralContactKernel}
\end{equation}

For \(SU(N)\), the exact one-loop equation is
\begin{equation}
    \boxed{
    \widehat{\mathbb L}_\pm(s)W[C]
    =
    \lambda
    \int_s^{s+2\pi}dt\,
    K_\pm(s,t;C)
    \left[
        W_2(C_{st},C_{ts})
        -
        \frac1{N^2}W[C]
    \right],
    }
    \label{eq:IntroductionFiniteNOneLoopHierarchy}
\end{equation}
where \(\lambda=g^2N\). The complete finite-\(N\) hierarchy also
contains joining of distinct loops and the corresponding
traceless-generator subtraction terms. In the planar limit, joining
and all subtraction terms are suppressed and
\(W_2(C_{st},C_{ts})\) factorizes into the product of the two split
one-loop averages.

The differentiated contact on the right-hand side is a genuine
short-distance physical distribution and must be distinguished sharply from the
ordered separation used to define POOC. We introduce a
Lorentz-covariant finite-width superspace bi-kernel and short gauge
transporters joining the separated contour points. The complete near-diagonal region is retained, and the hierarchy is solved at
fixed ultraviolet scale. The Wilson-loop functional remains an intermediate
gauge-invariant object; renormalization is imposed only after this solution has
been inserted into the worldline, worldsheet, proper-time, or other path
integrals defining physical observables.
Consequently, independent perimeter or cusp terms are not arbitrarily added to the finite-Stokes
solution space of the hierarchy.

\subsection*{Finite Stokes functionals and the dressing theorem}

An even loop functional \(S[C]\) is called a finite Stokes functional
when its first variation has the local form
\begin{equation}
    \delta S[C]
    =
    \int_C ds\,
    \xi^A(s)E^B(s)\Omega_{BA}[C;s],
    \label{eq:IntroductionFiniteStokesVariation}
\end{equation}
with a finite local super two-form \(\Omega_{AB}\). Its
bounded-variation point derivative is the discontinuity of the tangent
variation across a marked point. For a finite Stokes functional this
discontinuity vanishes. The torsion-subtracted area derivative then
obeys the graded chain rule, and the local I--T operators satisfy
\begin{equation}
    \widehat{\mathbb L}_\pm f(S)
    =
    f'(S)\widehat{\mathbb L}_\pm S
    \label{eq:IntroductionStokesChainRule}
\end{equation}
for every ordinary function of an even finite Stokes functional.
This gives the product rule: an exponential finite-Stokes zero mode
factors through every local left-hand side of the hierarchy.

The differentiated contact kernels make the right-hand side more
subtle.  At fixed ultraviolet scale we choose the connector to be a
geodesic wire on the selected stationary surface.  The same wire enters
the gauge-covariant contact and closes the two cut boundary arcs.  It is
an internal cut of one primary surface, so the two daughter
surfaces are obtained by restriction of the same Weierstrass map rather
than by solving two new Dirichlet problems.  For same-loop splitting,
define
\begin{equation}
    \Delta_S^{\rm split}(s,t)
    =
    S\!\left[C_{st}^{(\Lambda)}\right]
    +
    S\!\left[C_{ts}^{(\Lambda)}\right]
    -
    S[C],
    \label{eq:IntroductionSplittingDefect}
\end{equation}
and define the joining defect analogously.  Locality of the surface
action gives an exact partition identity,
\(
\Delta_S^{\rm split}\equiv0
\),
while the joined branch is the corresponding disjoint union with the
same wire traversed twice in opposite directions.  The dressing ratio
is therefore identically one as a functional of the two endpoint data,
not merely equal to one on the support of the contact.  Exact additivity
identifies the sum of the daughter actions with the same parent
finite-Stokes functional.  Proposition~\ref{prop:MainFiniteStokesChainRules}
then gives the three conditions required by the differentiated kernel,
\(
\Delta_S=0
\),
\(
\mathfrak D_A(s)\Delta_S=\mathfrak D_A(t)\Delta_S=0
\), and
\(
P_\alpha^{(s)}Q_{(t)}^\alpha\Delta_S=0
\), together with the anti-chiral conjugate.  No separate continuity or
linearized-sewing hypothesis is needed.

The result is exact at finite \(N\). If
\(W_{n,\Lambda}(C_1,\ldots,C_n)\) solves the complete
ultraviolet-defined I--T hierarchy and \(S[C]\) satisfies the
finite-Stokes, zero-mode, additivity, and sewing conditions, then
\begin{equation}
    \boxed{
    \widetilde W_{n,\Lambda}(C_1,\ldots,C_n)
    =
    W_{n,\Lambda}(C_1,\ldots,C_n)
    \exp\!\left[
        \rho\sum_{j=1}^{n}S[C_j]
    \right]
    }
    \label{eq:IntroductionFiniteNDressing}
\end{equation}
solves exactly the same hierarchy for every even constant \(\rho\), real or
complex. The theorem includes same-loop splitting, different-loop
joining, and all finite-\(N\), \(SU(N)\) subtraction terms. The planar result is a corollary, not an assumption of the proof.

\subsection*{The Lorentzian supersymmetric Hodge-dual zero mode}

The geometric zero modes completing this correspondence are supplied by a Lorentzian supersymmetric
Hodge-dual surface. The physical superloop is first mapped to its
even invariant boundary curve
\begin{equation}
    \mathcal C^a(s)
    =
    \mathcal C^a(s_0)
    +
    \int_{s_0}^{s}du\,E^a(u),
    \qquad
    \oint_Ce^a=0.
    \label{eq:IntroductionInvariantDevelopment}
\end{equation}
The second condition restricts us to superloops satisfying this invariant
closure condition.

For a complex Lorentzian two-form,
\begin{equation}
    (*_4)^2=-1.
\end{equation}
The chiral and anti-chiral Hodge sectors are selected by
\begin{equation}
    \boxed{
    P_\chi
    =
    \frac12(1+i\chi *_4),
    \qquad
    *_4P_\chi=-i\chi P_\chi,
    \qquad
    \chi=\pm1.
    }
    \label{eq:IntroductionLorentzianHodgeProjector}
\end{equation}
Let \(Y^a(\xi^1,\xi^2)\) be a complexified stationary Minkowski
embedding with boundary \(Y^a|_{\partial\mathcal D}=\mathcal C^a\).
For its tangent bivector \(B_{ab}\), define
\begin{equation}
    \Sigma_{ab}^{(\chi)}
    =
    2(P_\chi B)_{ab}
    =
    B_{ab}+i\chi(*_4B)_{ab}.
    \label{eq:IntroductionProjectedBivector}
\end{equation}
The stationary Lorentzian functional is
\begin{equation}
    S_\chi^{\rm HD}[Y]
    =
    -\int_{\mathcal D}d^2\xi\,
    \sqrt{
        -\frac12
        \Sigma_{ab}^{(\chi)}
        \Sigma^{(\chi)ab}
    }.
    \label{eq:IntroductionLorentzianHDFunctional}
\end{equation}
On a real timelike worldsheet its density is
\(\sqrt2\sqrt{-\det h}\), with the Lorentzian action sign shown
above. Hodge chirality is not an additional branch assumption; it follows directly from the projector.

The stationary parent surface directly supplies a finite Stokes area
 derivative $\Omega_{ab}^{(\chi)}$ belonging to one Lorentzian Hodge
 sector. No independently postulated chiral field strength,
 superinstanton, or superspace two-form completion is introduced. The
 complete supersymmetric dependence is carried by the invariant superspace
 one-form $\Pi^a|_C$ and hence by the invariant boundary curve. On the regular
 Super--Weierstrass branch constructed below, the stationary value of
 the parent functional defines $S_\chi[C]$ directly.  Because the bulk
 equations are satisfied, its first variation is entirely the boundary
 Stokes term. No separate
 loop-space curl or functional-integrability condition is required. On
 a physical Lorentzian slice the two branches are chosen as conjugates,
\begin{equation}
    S_-[C]=\bigl(S_+[C]\bigr)^\dagger.
    \label{eq:IntroductionSHDReality}
\end{equation}

The zero-mode identity is operator-theoretic. The
 torsion-subtracted POOC Jacobi identity is the loop-space Bianchi
 identity. After the Itoyama--Takashino chiral or anti-chiral
 projection, four-dimensional Hodge duality converts this identity
 directly into the homogeneous local loop equation. Therefore
\begin{equation}
    \boxed{
    \widehat{\mathbb L}_\pm S_+[C]
    =
    \widehat{\mathbb L}_\pm S_-[C]
    =0.
    }
    \label{eq:IntroductionCommonZeroModes}
\end{equation}

At a genuine loop contact, the free-trace BV
Super--Weierstrass data are restricted to the two sides of a
geodesic cut on the parent surface.  The four one-sided twistor germs are
redistributed between the daughter loops, and each daughter is placed
on its own standard boundary circle by moving the cut on the BV circle.  The
holomorphic functions are not recomputed: both daughters are conformal
pullbacks of restrictions of the same parent map.  The local parent
action therefore splits exactly into the two daughter actions.  The
same statement holds for the joined branch, where the double wire has
zero two-dimensional measure.  Consequently the splitting and joining
defects of \(S_+\) and \(S_-\) obey, branch by branch,
\(
\Delta_\chi=0
\),
\(
\mathfrak D_A(s)\Delta_\chi=\mathfrak D_A(t)\Delta_\chi=0
\), and
\(
P_\alpha^{(s)}Q_{(t)}^\alpha\Delta_\chi=0
\), together with the anti-chiral conjugate.  These are the geometric
conditions required by the finite-\(N\) dressing theorem.

For pure \(SU(N)\), let \(\Lambda_k^3\) denote the holomorphic strong
scale on the \(k\)-th supersymmetric vacuum branch, and choose a branch
of \(\Lambda_k^2=(\Lambda_k^3)^{2/3}\). The two conjugate Hodge
functionals unite into the real Lorentzian phase
\begin{equation}
    \boxed{
    \Phi_{\rm SHD}^{(k)}[C]
    =
    \frac{c}{2}
    \left(
        \Lambda_k^2S_+[C]
        +
        \bar\Lambda_k^2S_-[C]
    \right)
    =
    c\,\Re\!\left[
        \Lambda_k^2S_+[C]
    \right],
    }
    \label{eq:IntroductionPhysicalSHDPhase}
\end{equation}
where \(c\) is a real dimensionless coefficient. The exact dressed
solution is
\begin{equation}
    \boxed{
    W_{\Lambda,k}[C]
    =
    W_{\Lambda,0}[C]
    \exp\!\left[
        i\Phi_{\rm SHD}^{(k)}[C]
    \right].
    }
    \label{eq:IntroductionSHDDressedSolution}
\end{equation}
While the hierarchy perfectly constraints the admissible geometric zero mode, it does not
determine \(c\), the fractional-power branch, or the selected vacuum.
These remain defining properties of nonperturbative vacuum data.

For planar contours, no additional Plateau-asymptotic assumption is
required.  The exact Virasoro--Weierstrass reduction gives
\(S_+[C_{\rm pl}]=S_-[C_{\rm pl}]=2\sqrt2\,D[\mathcal C]\), and for the
ordinary static-source rectangle the boundary fermions vanish and
\(D[\mathcal C]=LT_E\).  Defining the positive vacuum coefficient
\(\kappa_k\) and
\(\sigma_k=2\sqrt2\,\kappa_k\), analytic continuation gives
\begin{equation}
    \Phi_{\rm SHD}^{(k)}[C_{T,L}]
    =-\sigma_kLT,
    \qquad
    \sigma_k>0,
    \label{eq:IntroductionRectanglePhase}
\end{equation}
and therefore
\begin{equation}
    W_{\Lambda,k}[C_{T,L}]
    =W_{\Lambda,0}[C_{T,L}]\exp[-i\sigma_kLT],
    \qquad
    \boxed{
    E_{{\rm SHD},k}(L)=\sigma_kL>0.
    }
    \label{eq:IntroductionPositiveStaticEnergy}
\end{equation}
Thus the two complex Lorentzian Hodge sectors combine into a real confining
phase whose planar static-energy contribution is exactly linear and positive.

\subsection*{Interpretation and organization}

This exact zero mode has two physical interpretations. First, after
analytic continuation to the Euclidean parabolic flow, the SHD factor is an
exact fixed point of the zero-noise SYM loop flow.  The induced
loop-flow generator is the exact local equation-of-motion combination of
\(\widehat{\mathbb L}_+\) and \(\widehat{\mathbb L}_-\); because both local
operators annihilate the SHD factor at every marked point, its stationarity
requires no additional factorization assumption.  A nontrivial stationary
Wilson-loop state can survive removal of the explicit stochastic source used
in stochastic quantization~\cite{ParisiWu1981,DamgaardHueffel1987} provided
this fixed point is dynamically stable and selected in the appropriate
long-flow-time, large-volume, and zero-noise limits.  This gives a
supersymmetric realization of \emph{spontaneous quantization}
\cite{Migdal_1986_Stochastic,migdal2025SQYMflow}.  Exact stationarity is
proved here; stability and dynamical selection remain conditional.  The
algebraic POOC zero-mode and finite-\(N\) dressing theorems are independent of
this remaining dynamical question.

Second, at large \(N\), this fixed loop-to-surface map is a gauge-invariant
version of Witten's master field \cite{Witten:1980ez}. It is not a fluctuating
colored gauge configuration. It is the map
\begin{equation}
    C
    \longmapsto
    \mathcal C_C
    \longmapsto
    Y_C
    \longmapsto
    \bigl(S_+[C],S_-[C]\bigr)
    \longmapsto
    \Phi_{\rm SHD}^{(k)}[C].
    \label{eq:IntroductionMasterSurfaceMap}
\end{equation}
This master surface describes the confining vacuum geometry carried by
the exact zero mode, while the undressed factor \(W_{\Lambda,0}[C]\)
encodes the remaining gauge dynamics.

The paper is organized as follows. Section~\ref{sec:LoopCalculus}
develops POOC in the native Lorentzian Itoyama--Takashino conventions:
the invariant supertangent, the lifted transporter, the
path-ordered commutator, the torsion-subtracted area derivative, the
ordered adjoint derivative, and the three-point Jacobi identity.
Section~\ref{sec:SchwingerDysonContactTerm} establishes the geometric form of
the exact I--T hierarchy, its local chiral and anti-chiral operators,
the finite-\(N\) splitting-and-joining equations, and the separate
ultraviolet definition of the contact term. The full Wess--Zumino
component derivation is  reconstructed in Supplementary Material,
Sec.~S4. Section~\ref{sec:StokesFunctionalsLeibniz} proves the finite-Stokes
chain and product rules, the exact geodesic-wire sewing
proposition, and the finite-$N$ dressing theorem. Section
\ref{sec:SuperHodgeDualSurface} constructs the Lorentzian SHD
functionals, confirms their common zero-mode equations, and extracts the
physical confining phase and positive static energy under the stated
geometric branch assumptions. Section
\ref{sec:SpontaneousStochasticization} articulates the conditional
gradient-flow interpretation, Section
\ref{sec:SHDAsWittenMasterField} outlines the planar master-surface
interpretation, and Section~\ref{sec:Conclusions} summarizes the
results and their physical implications.  Appendix~\ref{app:POOCExistingLoopCalculus}
places POOC in the earlier mathematical theory of generalized loops,
curvature derivatives, and bounded-variation paths; Supplementary Material,
Sec.~S11 gives the detailed comparison.
\section{Path-ordered superloop operator calculus}
\label{sec:LoopCalculus}

This section develops the path-ordered operator calculus (POOC) used
throughout the paper.  The construction is kinematical: it applies to
an arbitrary superspace connection and does not use the
Schwinger--Dyson equation.  We adopt the Lorentzian superspace
notation, metric, spinor conventions, connection sign, and transporter
orientation of Itoyama and Takashino
\cite{ItoyamaTakashino1997}.  The same conventions will be used in the
main text and in the Supplementary Material.

The basic principle is that coincident loop differentiation is defined
by path ordering rather than by an auxiliary regulator.  Tangent
operations are first placed in distinct ordered slots on the contour.
Their graded commutators, nested commutators, and Jacobi combinations
are formed at finite separation; only then is the ordered coincidence
limit taken.  This produces finite operator insertions and eliminates
the spurious kinematical singularities of formal equal-point
functional differentiation.  A genuine ultraviolet cutoff enters only
later, in the distributional contact term of the quantum loop
equation.

\subsection{Lorentzian superspace, invariant tangent, and BV loop class}
\label{subsec:FlatN1SuperspaceAndSuperloops}

We work in flat \(4d,\mathcal N=1\) Lorentzian superspace with
\begin{equation}
    z^M=(x^a,\eta^\alpha,\bar\eta^{\dot\alpha}),
    \qquad
    \eta_{ab}=\operatorname{diag}(-1,+1,+1,+1),
    \qquad
    \epsilon^{0123}=+1.
    \label{eq:SuperspaceCoordinates}
\end{equation}
All spinor, sigma-matrix, raising/lowering, and left-Grassmann-derivative
conventions are those of Itoyama and Takashino and are collected in
Supplementary Material, Sec.~S1.  The flat
derivatives obey
\begin{equation}
    \{D_\alpha,\bar D_{\dot\alpha}\}
    =-2i\sigma^a_{\alpha\dot\alpha}D_a,
    \qquad
    \{D_\alpha,D_\beta\}
    =\{\bar D_{\dot\alpha},\bar D_{\dot\beta}\}=0,
    \label{eq:FlatSuperspaceDerivativeAlgebra}
\end{equation}
or, collectively,
\begin{equation}
    [D_A,D_B\}=T_{AB}{}^C D_C,
    \qquad
    T_{\alpha\dot\alpha}{}^a
    =T_{\dot\alpha\alpha}{}^a
    =-2i\sigma^a_{\alpha\dot\alpha}.
    \label{eq:FlatSuperspaceTorsionAlgebra}
\end{equation}

The invariant one-form basis is
\begin{equation}
\begin{aligned}
    e^a&=dx^a-i\,d\eta\sigma^a\bar\eta
                  +i\,\eta\sigma^a d\bar\eta,\\
    e^\alpha&=d\eta^\alpha,
    \qquad
    e^{\dot\alpha}=-d\bar\eta^{\dot\alpha},
\end{aligned}
    \label{eq:InvariantOneFormBasis}
\end{equation}
with \(d=e^A D_A\).  A closed superloop is a continuous periodic map
\begin{equation}
    C:\ s\mapsto z^M(s),
    \qquad z^M(s+2\pi)=z^M(s),
    \label{eq:SuperloopDefinition}
\end{equation}
and its invariant tangent is defined by
\begin{equation}
    e^A\big|_C=ds\,E^A(s).
    \label{eq:InvariantTangentDefinition}
\end{equation}
Explicitly,
\begin{equation}
\begin{aligned}
    E^a&=\dot x^a-i\dot\eta\sigma^a\bar\eta
                    +i\eta\sigma^a\dot{\bar\eta},\\
    E^\alpha&=\dot\eta^\alpha,
    \qquad
    E^{\dot\alpha}=-\dot{\bar\eta}^{\dot\alpha},
\end{aligned}
    \label{eq:InvariantVectorTangent}
\end{equation}
and hence
\begin{equation}
    \boxed{E^A(s)D_A=\frac{d}{ds}}
    \label{eq:InvariantTangentDerivative}
\end{equation}
on every superfield pulled back to the contour.

\paragraph{BV loop class and circle-ordered symbols.}
The loop space used below consists of continuous closed maps whose
coordinate and invariant tangents have bounded variation,
\begin{equation}
    z^M\in C^0(S^1),
    \qquad
    \dot z^M,\ E^A\in BV(S^1).
    \label{eq:BVSuperloopClass}
\end{equation}
A periodic BV tangent is a function, equivalently a distribution, on
the oriented circle.  At every marked point it has two finite traces,
\begin{equation}
    E^A(s-0),
    \qquad
    E^A(s+0).
    \label{eq:BVOneSidedTangentTraces}
\end{equation}
After cutting the circle at \(s=0\), periodicity does not imply
\begin{equation}
    E^A(2\pi-0)=E^A(0+0).
    \label{eq:NoBVTangentTraceMatching}
\end{equation}
These are the two traces of a possible jump at the marked cut.  Closure
of the coordinate loop constrains only the zero mode,
\begin{equation}
    z^M(2\pi)-z^M(0)
    =\int_0^{2\pi}ds\,\dot z^M(s)=0,
    \label{eq:CoordinateLoopClosureZeroMode}
\end{equation}
not the local tangent traces.

Every periodic BV symbol has the ordinary integer Fourier
representation
\begin{equation}
    E^A(s)\sim\sum_{n\in\mathbb Z}E_n^A e^{ins}.
    \label{eq:BVIntegerFourierSeries}
\end{equation}
At a jump, the symmetric Fourier sums converge in the
Dirichlet--Jordan sense to
\begin{equation}
    \frac12\bigl(E^A(s-0)+E^A(s+0)\bigr),
    \label{eq:BVFourierJumpAverage}
\end{equation}
while the two traces remain distinct.  The slow \(1/|n|\) tail is part
of the ordered data, not a failure of periodicity.

This is the natural regularity class for POOC.  In a circle-ordered
exponential the BV coefficient \(E^A(s)\) is a one-dimensional symbol
multiplying the operator \(\nabla_A\).  Functional differentiation
inserts that operator at the ordered point, and the two traces represent
the slots immediately before and after the mark.  Their ordered
difference therefore realizes the operator commutator.  The detailed
Dirichlet--Jordan and distributional statements are recorded in
Supplementary Material, Secs.~S1 and ~S5.

\subsection{Constrained connection, super Wilson line, and operator lift}
\label{subsec:GaugeCovariantDerivativesAndWilsonOperator}

Let \(A=e^AA_A\) and
\begin{equation}
    \nabla_A=D_A-A_A,
    \qquad
    [\nabla_A,\nabla_B\}
    =T_{AB}{}^C\nabla_C-F_{AB}.
    \label{eq:GaugeCovariantDerivative}
\end{equation}
The conventional \(4d,\mathcal N=1\) constraints are
\begin{equation}
    F_{\alpha\beta}
    =F_{\dot\alpha\dot\beta}
    =F_{\alpha\dot\beta}=0.
    \label{eq:ConventionalSYMConstraints}
\end{equation}
The chiral-prepotential solution and the complete component
conventions are given in Supplementary Material,
Secs.~S1 and ~S4.

For an oriented segment, the super Wilson line is
\begin{equation}
    W_S[C_{z(s_2)z(s_1)}]
    =P\exp\int_{s_1}^{s_2}du\,E^A(u)A_A(z(u)).
    \label{eq:SuperWilsonLineParameterForm}
\end{equation}
Introduce the flat backward translation \(\mathsf T(s_1,s_2)\), which
returns a test superfield from \(z(s_2)\) to \(z(s_1)\), and define
\begin{equation}
    \boxed{
    \mathbb U(s_2,s_1)
    =W_S[C_{z(s_2)z(s_1)}]\,\mathsf T(s_1,s_2)
    =P\exp\!\left[-\int_{s_1}^{s_2}du\,E^A\nabla_A\right].
    }
    \label{eq:OperatorTransporter}
\end{equation}
For a closed continuous contour the accumulated translation is the
identity because \(z(2\pi)=z(0)\); no equality of the two tangent traces
at the cut is required.  Thus
\begin{equation}
    \frac1N\Tr\mathbb U(2\pi,0)
    =\frac1N\Tr W_S[C_{z_0z_0}].
    \label{eq:OperatorTransporterClosedTrace}
\end{equation}
The transporter composes as
\begin{equation}
    \mathbb U(s_3,s_1)
    =\mathbb U(s_3,s_2)\mathbb U(s_2,s_1).
    \label{eq:TransporterComposition}
\end{equation}

Because the exponent contains \(-E^A\nabla_A\), define
\begin{equation}
    \boxed{\delta_A(s)\equiv-\frac{\delta}{\delta E^A(s)}}.
    \label{eq:TangentFunctionalDerivative}
\end{equation}
Then
\begin{equation}
    \boxed{
    \delta_A(s)\mathbb U_C
    =\mathbb U(2\pi,s)\nabla_A(z(s))\mathbb U(s,0).
    }
    \label{eq:TangentDerivativeInsertion}
\end{equation}
The infinitesimal derivation, endpoint equations, and backward-
translation proof are given in Supplementary Material,
Sec.~S2.

\subsection{Path-ordered commutator and torsion-subtracted area derivative}
\label{subsec:OrderedCommutatorAndAreaDerivative}

Let \(\mathcal O_A(s)\) and \(\mathcal O_B(s)\) be homogeneous loop
insertions of parities \(|A|\) and \(|B|\).  Their path-ordered graded
commutator at a marked point is
\begin{equation}
\begin{aligned}
    [\mathcal O_A,\mathcal O_B\}_o(s)
    \equiv
    \lim_{\epsilon\downarrow0}
    \Big[
        &\mathcal O_A(s-\epsilon)
        \mathcal O_B(s+\epsilon)
        \\
        &-
        (-1)^{|A||B|}
        \mathcal O_B(s-\epsilon)
        \mathcal O_A(s+\epsilon)
    \Big].
\end{aligned}
    \label{eq:OrderedGradedCommutator}
\end{equation}
The ordered slots are part of the operator definition.  The parameter
\(\epsilon\) is not an ultraviolet cutoff and no scale survives the
coincidence limit.

For tangent operators,
\begin{equation}
\begin{aligned}
    [\delta_A,\delta_B\}_o(s)
    =
    \delta_A(s-0)\delta_B(s+0)
    -
    (-1)^{|A||B|}
    \delta_B(s-0)\delta_A(s+0).
\end{aligned}
    \label{eq:OrderedTangentDerivativeCommutator}
\end{equation}
Using \eqref{eq:TangentDerivativeInsertion} and the transporter
composition law,
\begin{align}
    [\delta_A,\delta_B\}_o(s)\mathbb U_C
    & =
    \mathbb U(2\pi,s)
    [\nabla_A,\nabla_B\}(z(s))
    \mathbb U(s,0)
    \nonumber\\
    & =
    \mathbb U(2\pi,s)
    \left(
        T_{AB}{}^C\nabla_C-F_{AB}
    \right)(z(s))
    \mathbb U(s,0).
    \label{eq:OrderedCommutatorGivesFullCommutator}
\end{align}
The first term is the known flat-superspace torsion translation.  The
curvature insertion is therefore the torsion-subtracted combination
\begin{equation}
    \boxed{
    \delta^\Sigma_{AB}(s)
    =
    T_{AB}{}^C\delta_C(s)
    -
    [\delta_A,\delta_B\}_o(s).
    }
    \label{eq:TorsionCorrectedAreaDerivative}
\end{equation}
The order of the two terms, relative to the convention often used with
\(D_A+A_A\), is fixed here by the I--T definitions
\(\nabla_A=D_A-A_A\) and \(F=dA+A A\).  It follows exactly that
\begin{equation}
    \delta^\Sigma_{AB}(s)\mathbb U_C
    =
    \mathbb U(2\pi,s)
    F_{AB}(z(s))
    \mathbb U(s,0).
    \label{eq:AreaDerivativeCurvatureInsertion}
\end{equation}

For a homogeneous adjoint superfield \(X\), define the marked
insertion
\begin{equation}
    \mathfrak I_s(X)
    =
    \mathbb U(2\pi,s)
    X(z(s))
    \mathbb U(s,0).
    \label{eq:MarkedInsertionMap}
\end{equation}
Then
\begin{equation}
    \delta^\Sigma_{AB}(s)\mathbb U_C
    =
    \mathfrak I_s(F_{AB}).
    \label{eq:AreaDerivativeAsMarkedInsertion}
\end{equation}
After closing the trace, this is the superspace Mandelstam formula
\begin{equation}
    \frac{\delta W[C]}
         {\delta\Sigma^{AB}(s)}
    =
    \frac1N
    \left\langle
        \Tr P
        \left[
            F_{AB}(z(s))
            \exp\oint_C A
        \right]
    \right\rangle.
    \label{eq:SuperspaceMandelstamFormula}
\end{equation}
Thus the formal I--T area derivative is represented directly by
the finite path-ordered operation
\eqref{eq:TorsionCorrectedAreaDerivative}.

The conventional constraints
\eqref{eq:ConventionalSYMConstraints}
annihilate the pure-spinor and mixed undotted--dotted spinor area
derivatives.  For one vector and one spinor index the flat torsion
vanishes, so
\begin{equation}
    \delta^\Sigma_{a\dot\alpha}
    =
    -[\delta_a,\delta_{\dot\alpha}\}_o,
    \qquad
    \delta^\Sigma_{a\alpha}
    =
    -[\delta_a,\delta_\alpha\}_o.
    \label{eq:MixedAreaDerivativePureCommutator}
\end{equation}
These signs will be inherited by the chiral and anti-chiral I--T
operators when the closed hierarchy is translated into POOC.

\subsection{Ordered adjoint derivative and the graded Jacobi identity}
\label{subsec:OrderedAdjointDerivativeAndJacobi}

The covariant derivative of a marked insertion is itself represented
by a path-ordered tangent operation.  For a homogeneous adjoint
superfield \(X\) of parity \(|X|\), define
\begin{align}
    \mathfrak D_A^{\,o}(s)\,
    \mathfrak I_s(X)
    & =
    \delta_A(s-0)\mathfrak I_s(X)
    \nonumber\\
    &\quad
    -
    (-1)^{|A||X|}
    \mathfrak I_s(X)\delta_A(s+0).
    \label{eq:OrderedAdjointTangentDerivative}
\end{align}
The first tangent insertion lies immediately before \(X\), and the
second immediately after it.  Using the endpoint identities,
\begin{equation}
    \boxed{
    \mathfrak D_A^{\,o}(s)\,
    \mathfrak I_s(X)
    =
    \mathfrak I_s(\mathcal D_A X).
    }
    \label{eq:OrderedAdjointDerivativeLemma}
\end{equation}
Therefore
\begin{equation}
    \mathcal D_A
    \quad\longleftrightarrow\quad
    \mathfrak D_A^{\,o}
    =
    \operatorname{ad}_{\delta_A}^{o}
    \label{eq:CovariantDerivativeLoopDictionary}
\end{equation}
at the marked-insertion level.

Applying this identity to the curvature insertion gives
\begin{equation}
    \mathfrak D_C^{\,o}
    \delta^\Sigma_{AB}(s)\mathbb U_C
    =
    \mathfrak I_s
    \left(
        \mathcal D_C F_{AB}
    \right).
    \label{eq:OrderedDerivativeOfAreaInsertion}
\end{equation}
Using
\eqref{eq:TorsionCorrectedAreaDerivative}, the same operator is
\begin{align}
    \mathfrak D_C^{\,o}
    \delta^\Sigma_{AB}
    ={}&
    T_{AB}{}^D
    [\delta_C,\delta_D\}_o
    \nonumber\\
    &-
    [\delta_C,
        [\delta_A,\delta_B\}_o
    \}_o.
    \label{eq:TorsionCorrectedTripleCommutator}
\end{align}

Nested brackets are defined with three distinct ordered slots.  At
finite separation, transporter multiplication is associative, so the
ordered tangent operators satisfy the exact graded Jacobi identity
\begin{align}
    0
    ={}&
    (-1)^{|A||C|}
    [\delta_A,[\delta_B,\delta_C\}_o\}_o
    \nonumber\\
    &+
    (-1)^{|B||A|}
    [\delta_B,[\delta_C,\delta_A\}_o\}_o
    \nonumber\\
    &+
    (-1)^{|C||B|}
    [\delta_C,[\delta_A,\delta_B\}_o\}_o.
    \label{eq:OrderedGradedJacobiIdentity}
\end{align}
Only after this sum has been formed are the three ordered points taken
to coincidence.  This is an operator identity, not a regulated
approximation.

Separating the flat torsion and using
\eqref{eq:TorsionCorrectedAreaDerivative}, the Jacobi identity becomes
the torsionful loop-space Bianchi identity
\begin{equation}
    \boxed{
    \mathfrak D^{\,o}_{[A}
    \delta^\Sigma_{BC)}
    -
    T_{[AB|}{}^D
    \delta^\Sigma_{D|C)}
    =
    0.
    }
    \label{eq:LoopSpaceTorsionfulBianchiIdentity}
\end{equation}
Here \([ABC)\) denotes the graded cyclic antisymmetrization appropriate
to a super three-form.  Acting on the transporter gives the ordinary
superspace identity
\begin{equation}
    \mathcal D_{[A}F_{BC)}
    -
    T_{[AB|}{}^D
    F_{D|C)}
    =
    0.
    \label{eq:SuperspaceTorsionfulBianchiFromLoopJacobi}
\end{equation}
This Jacobi--Bianchi correspondence is the central kinematical result
of POOC.  Section~\ref{sec:SchwingerDysonContactTerm} will apply it to
the chiral and anti-chiral differential combinations in the exact
Itoyama--Takashino hierarchy.

\subsection{Bounded-variation point derivative}
\label{subsec:PointDerivativeAndClassicalSMMOperator}

The ordered adjoint derivative above acts on matrix-valued marked
insertions.  For an ordinary scalar loop functional \(F[C]\), define
instead the BV point derivative
\begin{equation}
    \mathfrak D_A(s)F[C]
    =\frac{\delta F}{\delta E^A(s-0)}
     -\frac{\delta F}{\delta E^A(s+0)}.
    \label{eq:PointDerivativeDefinition}
\end{equation}
This measures the jump of the conjugate tangent variation; it does not
assert equality of the tangent traces themselves.

For a covariant loop deformation
\begin{equation}
    \xi^A=\delta z^M e_M{}^A,
    \label{eq:CovariantLoopVariation}
\end{equation}
the invariant tangent varies as
\begin{equation}
    \boxed{
    \delta E^A=\partial_s\xi^A+E^B\xi^C T_{CB}{}^A.
    }
    \label{eq:VariationOfInvariantTangent}
\end{equation}
The derivative term generates the two endpoint traces; the torsion
term is an ordinary finite interval contribution.  Consequently every
even finite Stokes functional obeys
\begin{equation}
    \mathfrak D_A(s)S[C]=0.
    \label{eq:FiniteStokesPointDerivativePreview}
\end{equation}
The shrinking-interval BV proof is given in Supplementary Material,
Sec.~S5.

\section{Geometric form of the exact Itoyama--Takashino hierarchy}
\label{sec:SchwingerDysonContactTerm}

Itoyama and Takashino derived an exact finite-\(N\)
Schwinger--Dyson hierarchy for the \(4d,\mathcal N=1\) superspace
Wilson loop \cite{ItoyamaTakashino1996,ItoyamaTakashino1997}.  Their
calculation is performed in the Lorentzian superspace conventions
adopted in Section~\ref{sec:LoopCalculus}.  It passes through
Wess--Zumino gauge, where the functional variation can be evaluated in
components, and ends with a manifestly supersymmetric and
supergauge-invariant equation written entirely in super Wilson-loop
variables.

The complete Wess--Zumino-gauge calculation, the modified functional
variation, its graded functional divergence, the induced connection
projector, the gauge-unfixing Ward identity, and the finite-\(N\) color
algebra are reconstructed in Supplementary Material,
Sec.~S4.  The purpose of the present section is to state this geometry and
rewrite the final closed hierarchy in the path-ordered operator calculus
developed above.

All contour coefficients in this hierarchy are understood as periodic
BV symbols in the sense of
\eqref{eq:BVSuperloopClass}--\eqref{eq:BVFourierJumpAverage}.  A
contact condition identifies superspace positions, not tangent germs.
Thus, at a genuine self-contact \(z(s)=z(t)\), the four one-sided
values
\begin{equation}
    E^A(s-0),\quad E^A(s+0),\quad
    E^A(t-0),\quad E^A(t+0)
    \label{eq:FourBVContactGerms}
\end{equation}
remain independent ordered data.  The splitting operation redistributes
these four germs between the two daughter loops; it does not impose a
periodic, anti-periodic, or smooth matching condition on them.

\subsection{From Wess--Zumino gauge to a closed superloop equation}
\label{subsec:ITGeometricOrigin}

Itoyama and Takashino evaluate the Schwinger--Dyson variation in
Wess--Zumino gauge and then restore the supergauge volume.  The full
component calculation, including the modified divergence-free
functional vector field, its action on the constrained connection,
and the finite-\(N\) color algebra, is reconstructed in Supplementary
Material, Sec.~S4.  Its local insertion is the
covariant chiral equation-of-motion superfield
\begin{equation}
    \mathcal E_+(z)
    =\frac18\epsilon_{\alpha\beta}
      \bar\sigma^{a\dot\alpha\beta}
      \mathcal D^\alpha F_{a\dot\alpha}(z),
    \qquad
    \mathcal E_+\propto\mathcal D^\alpha W_\alpha.
    \label{eq:ITChiralEOMSuperfieldMain}
\end{equation}
The Lorentzian conjugate insertion is
\(\mathcal E_-\propto
\mathcal D_{\dot\alpha}\bar W^{\dot\alpha}\).

Gauge unfixing converts the residual difference between the WZ-gauge
contact expression and the manifest superspace kernel into a total
derivative in the integrated contour point,
\begin{equation}
    \mathscr I_d(s;C)
    =\oint_C dz^M(t)\,\partial_M^{(t)}
      \left[
        \delta^{4|4}\!\left(z(t)-z(s)\right)
        \mathscr F(t,s;C)
      \right].
    \label{eq:ITGaugeUnfixingTotalDerivativeMain}
\end{equation}
The marked point is not integrated.  The contour integral vanishes by
the restricted-supergauge Ward identity.  In the BV formulation this
requires no trace matching: after the circle is cut, the endpoint
difference of the absolutely continuous part is canceled by the atomic
jump in the distributional derivative of the periodic BV primitive.
The resulting equation is therefore closed in super Wilson-loop
variables while retaining the two ordered germs at every marked point.

\subsection{Path-ordered form of the local chiral and anti-chiral operators}
\label{subsec:ITClosedPOOCHierarchy}

The local differential operator in the manifest Itoyama--Takashino
equation is
\begin{equation}
    \frac18
    \epsilon_{\alpha\beta}
    \bar\sigma^{a\dot\alpha\beta}
    D^{\alpha}
    \frac{\delta}{\delta\Sigma^{a\dot\alpha}(z')}.
    \label{eq:ITFormalChiralLoopOperator}
\end{equation}
The formal area derivative is represented by
\(\delta^\Sigma_{a\dot\alpha}\), while the covariant differentiation
of the marked curvature insertion is represented by
\(\mathfrak D^{o\,\alpha}\).  Hence its POOC form is
\begin{equation}
    \boxed{
    \widehat{\mathbb L}_+(s)
    =
    \frac18
    \epsilon_{\alpha\beta}
    \bar\sigma^{a\dot\alpha\beta}
    \mathfrak D^{o\,\alpha}(s)
    \delta^\Sigma_{a\dot\alpha}(s).
    }
    \label{eq:ProjectedChiralTripleCommutator}
\end{equation}
Using the mixed-area identity
\eqref{eq:MixedAreaDerivativePureCommutator}, this becomes
\begin{equation}
    \boxed{
    \widehat{\mathbb L}_+(s)
    =
    -\frac18
    \epsilon_{\alpha\beta}
    \bar\sigma^{a\dot\alpha\beta}
    \left[
        \delta^{\alpha},
        [\delta_a,\delta_{\dot\alpha}\}_o
    \right\}_o(s).
    }
    \label{eq:ProjectedChiralTripleCommutatorExpanded}
\end{equation}
The minus sign is fixed by the native I--T conventions
\(\nabla_A=D_A-A_A\), \(F=dA+AA\), and
\begin{equation}
    \delta^\Sigma_{a\dot\alpha}
    =
    -[\delta_a,\delta_{\dot\alpha}\}_o.
\end{equation}
Acting on the normalized Wilson-loop average,
\begin{equation}
    \widehat{\mathbb L}_+(s)W[C]
    =
    \frac1N
    \left\langle
        \Tr\!\left[
            \mathcal E_+\bigl(z(s)\bigr)
            W_S[C_{z(s)z(s)}]
        \right]
    \right\rangle.
    \label{eq:ITChiralEOMInsertionMain}
\end{equation}

The Lorentzian conjugate equation is represented by
\begin{equation}
    \boxed{
    \widehat{\mathbb L}_-(s)
    =
    \frac18
    \epsilon_{\dot\alpha\dot\beta}
    \sigma^{a\alpha\dot\beta}
    \mathfrak D^{o\,\dot\alpha}(s)
    \delta^\Sigma_{a\alpha}(s),
    }
    \label{eq:ProjectedAntiChiralTripleCommutator}
\end{equation}
that is,
\begin{equation}
    \boxed{
    \widehat{\mathbb L}_-(s)
    =
    -\frac18
    \epsilon_{\dot\alpha\dot\beta}
    \sigma^{a\alpha\dot\beta}
    \left[
        \delta^{\dot\alpha},
        [\delta_a,\delta_\alpha\}_o
    \right\}_o(s).
    }
    \label{eq:ProjectedAntiChiralTripleCommutatorExpanded}
\end{equation}
It inserts
\begin{equation}
    \mathcal E_-(z)
    =
    \frac18
    \epsilon_{\dot\alpha\dot\beta}
    \sigma^{a\alpha\dot\beta}
    \mathcal D^{\dot\alpha}F_{a\alpha}(z),
    \label{eq:ITAntiChiralEOMSuperfieldMain}
\end{equation}
which is proportional to
\(\mathcal D_{\dot\alpha}\bar W^{\dot\alpha}\).
In Lorentzian signature the two equations are conjugate.

Although \(\widehat{\mathbb L}_+\) and
\(\widehat{\mathbb L}_-\) have no free Lorentz index, they are not
single-component scalar loop equations.  They are Grassmann-even
superfield equations at the marked point \(z(s)\).  Their finite
expansions in \(\eta(s)\) and \(\bar\eta(s)\) package the component
Schwinger--Dyson equations of the complete gauge supermultiplet.

\subsection{Manifest contact kernels and the exact finite-\texorpdfstring{\(N\)}{N} hierarchy}
\label{subsec:ITManifestHierarchy}

The pulled-back chiral and anti-chiral line operators are
\begin{equation}
    D_f^\alpha(t)
    =E^\alpha(t)
     -\frac{i}{4}E^a(t)
       \bar\sigma_a^{\dot\alpha\alpha}\bar D_{\dot\alpha}^{(t)},
    \label{eq:ITChiralLineDerivativeMain}
\end{equation}
\begin{equation}
    \bar{\mathcal D}_f^{\dot\alpha}(t)
    =E^{\dot\alpha}(t)
     +\frac{i}{4}E^a(t)
       \sigma_a^{\alpha\dot\alpha}D_\alpha^{(t)}.
    \label{eq:ITAntiChiralLineDerivativeMain}
\end{equation}
The manifest contact distributions are
\begin{equation}
    \boxed{
    K_+(s,t;C)
    =D_f^\alpha(t)D_\alpha^{(t)}
      \delta^{4|4}\!\left(z(t)-z(s)\right),
    }
    \label{eq:ProjectedContactKernel}
\end{equation}
\begin{equation}
    \boxed{
    K_-(s,t;C)
    =\bar{\mathcal D}_f^{\dot\alpha}(t)
      \bar D_{\dot\alpha}^{(t)}
      \delta^{4|4}\!\left(z(t)-z(s)\right).
    }
    \label{eq:ProjectedAntiChiralContactKernel}
\end{equation}
The derivatives initially act on the superspace delta-function and,
after graded integration by parts, on the endpoint dependence of the
split or joined loop factors.

Let
\begin{equation}
    W_n(\mathbf C)
    =\left\langle
      \prod_{\ell=1}^n\frac1N\Tr W_S[C_\ell]
      \right\rangle,
    \qquad \mathbf C=(C_1,\ldots,C_n),
    \label{eq:NormalizedMultiloopCorrelatorMain}
\end{equation}
and \(\lambda=g^2N\).  If the marked point is on \(C_j\), the exact
finite-\(N\), \(SU(N)\) hierarchy is
\begin{align}
    \widehat{\mathbb L}_{\pm,j}(s)W_n(\mathbf C)
    ={}&
    \lambda\int_{C_j}dt\,K_\pm^{jj}(s,t)
    \left[
      W_{n+1}\!\left(\mathbf C^{j\to(st,ts)}\right)
      -\frac1{N^2}W_n(\mathbf C)
    \right]
    \nonumber\\
    &+\frac{\lambda}{N^2}
    \sum_{k\ne j}\int_{C_k}dt\,K_\pm^{jk}(s,t)
    \left[
      W_{n-1}\!\left(\mathbf C^{jk\to j\# k}\right)
      -W_n(\mathbf C)
    \right].
    \label{eq:FiniteNMultiloopSMM}
\end{align}
The first line is same-loop splitting and the second is different-loop
joining.  The contour orientations, connector prescription, and color
completeness calculation are detailed in Supplementary Material,
Sec.~S4.

For a single loop this reduces to
\begin{equation}
    \boxed{
    \widehat{\mathbb L}_\pm(s)W[C]
    =\lambda\int_s^{s+2\pi}dt\,K_\pm(s,t;C)
      \left[W_2(C_{st},C_{ts})-\frac1{N^2}W[C]\right].
    }
    \label{eq:FiniteNSingleLoopSMM}
\end{equation}
At large \(N\), factorization closes the hierarchy,
\begin{equation}
    \boxed{
    \widehat{\mathbb L}_\pm(s)W[C]
    =\lambda\int_s^{s+2\pi}dt\,K_\pm(s,t;C)
      W(C_{st})W(C_{ts})+O(N^{-2}).
    }
    \label{eq:PlanarFullSuperspaceSMM}
\end{equation}

\subsection{Geodesic-wire ultraviolet definition of the contact}
\label{subsec:GaugeCovariantRegularization}

POOC removes only the kinematical ambiguity of ordered coincidence.
The contact on the right-hand side of the quantum hierarchy is a true
ultraviolet distribution and is defined at a finite physical scale.
For the SHD branch we use the freedom in the shape of the short gauge
connector to choose a geodesic wire on the selected Euclideanized
surface.  Denote this wire by
\(
\gamma_{ts}^{C}:C(t)\to C(s)
\),
its inverse by
\(
\gamma_{st}^{C}=(\gamma_{ts}^{C})^{-1}
\),
and its length by
\begin{equation}
    L_C(s,t)=\operatorname{Length}_{\Sigma_C}(\gamma_{ts}^{C}).
    \label{eq:SurfaceGeodesicWireLengthMain}
\end{equation}
Every surface path is also a target-space path, and therefore
\begin{equation}
    L_C(s,t)\geq \bigl|C(s)-C(t)\bigr|.
    \label{eq:SurfaceGeodesicChordBoundMain}
\end{equation}
Thus the large-\(\Lambda\) profile suppresses every pair away from the
target-space contact set.

The universal four-dimensional normalization of the exponential profile
is
\begin{equation}
    \boxed{
    \mathcal N(\Lambda)
    =\left[
      \int_{\mathbb R^4}d^4x\,e^{-\Lambda|x|}
      \right]^{-1}
    =\frac{\Lambda^4}{12\pi^2}.
    }
    \label{eq:SurfaceGeodesicNormalizationMain}
\end{equation}
Indeed,
\(
\int_{\mathbb R^4}d^4x\,e^{-\Lambda|x|}
=2\pi^2\int_0^\infty r^3e^{-\Lambda r}dr
=12\pi^2/\Lambda^4
\).
A convenient finite-width representative of the superspace contact is
therefore
\begin{equation}
    \boxed{
    \Delta_{\Lambda,\Sigma}^{4|4}(s,t)
    =\frac{\Lambda^4}{12\pi^2}
      e^{-\Lambda L_C(s,t)}
      \delta^{0|4}\!\left(\Theta(t)-\Theta(s)\right).
    }
    \label{eq:RegularizedSuperspaceDelta}
\end{equation}
Here \(\Theta=(\eta,\bar\eta)\), and the Berezin delta is kept exact.
On the localization scale \(L_C\sim\Lambda^{-1}\), the blown-up
surface is asymptotically flat: its dimensionless second fundamental
form is \(O(\Lambda^{-1})\), and the surface geodesic agrees with the
straight target-space chord up to corrections which change the
normalized profile only by \(O(\Lambda^{-2})\).  Consequently,
\begin{equation}
    \Delta_{\Lambda,\Sigma}^{4|4}(s,t)
    \xrightarrow[\Lambda\to\infty]{}
    \delta^{4|4}\!\left(z(t)-z(s)\right)
    \label{eq:SurfaceGeodesicDeltaLimitMain}
\end{equation}
in the loop-contact pairing.  Every stationary wire of strictly
positive limiting length is exponentially suppressed.  The detailed
Laplace-method proof is given in Supplementary Material, Sec.~S4.12.

The color index is transported along the same wire.  The resulting
adjoint bi-kernel is
\begin{equation}
    \boldsymbol\Delta_{\Lambda,\Sigma}(s,t)
    =\operatorname{Ad}\!\left[\mathcal U_{\gamma^C}(s,t)\right]
      \Delta_{\Lambda,\Sigma}^{4|4}(s,t),
    \label{eq:GaugeCovariantAdjointKernelMain}
\end{equation}
and the two ultraviolet contact kernels are
\begin{equation}
    \boxed{
    K_{+,\Lambda}
    =D_f^\alpha D_\alpha\boldsymbol\Delta_{\Lambda,\Sigma},
    \qquad
    K_{-,\Lambda}
    =\bar{\mathcal D}_f^{\dot\alpha}\bar D_{\dot\alpha}
      \boldsymbol\Delta_{\Lambda,\Sigma}.
    }
    \label{eq:RegularizedProjectedChiralKernel}
\end{equation}
The complementary wire closes the other cut segment with the opposite
orientation.  The profile, the transporter, and the wire may all depend
on the endpoint data; the derivatives in
\eqref{eq:RegularizedProjectedChiralKernel} act on the complete
contact pairing.

The hierarchy is first solved at fixed \(\Lambda\), and only then
inserted into the physical worldline, proper-time, or other path
integrals:
\begin{equation}
    \boxed{
    \text{define at fixed }\Lambda
    \;\longrightarrow\;
    \text{solve for }W_{n,\Lambda}
    \;\longrightarrow\;
    \text{renormalize physical observables}.
    }
    \label{eq:UVOrderOfOperationsMain}
\end{equation}
The surface geodesic wire has a second role: it is an internal cut
of the primary stationary surface.  This makes the differentiated
sewing identity exact at fixed cutoff, as proved in the next section
and in Supplementary Material, Sec.~S7.

\section{Finite Stokes zero modes and the exact dressing theorem}
\label{sec:StokesFunctionalsLeibniz}
\label{sec:ZeroModeDressing}

This section separates two statements.  The first is the single-loop
finite-Stokes algebra, which controls every local left-hand side of the
hierarchy.  The second is the exact sewing of the differentiated
ultraviolet contact.  The sewing is implemented by drawing the contact
wire as a geodesic cut on the selected primary surface.  Cutting
one local surface action along this wire gives two daughter actions
whose sum is exactly the parent action.  The complete BV, graded
Leibniz, geodesic-wire, and differentiated-contact calculations are
given in Supplementary Material, Secs.~S5--S8 and S10.5.

\subsection{Finite Stokes class and POOC identities}

\begin{definition}[Finite Stokes functional]
\label{def:FiniteStokesFunctional}
An even reparametrization-invariant scalar functional \(S[C]\) is a
finite Stokes functional if its first variation is
\begin{equation}
    \boxed{
    \delta S[C]
    =\int ds\,\xi^A(s)E^B(s)\Omega_{BA}[C;s],
    \qquad
    \Omega_{AB}=\delta^\Sigma_{AB}S,
    }
    \label{eq:FiniteStokesVariationMain}
\end{equation}
with a finite local graded-antisymmetric two-form \(\Omega_{AB}\).
\end{definition}
A local perimeter functional is not of this form and is therefore not
an independent finite-Stokes zero mode.

\begin{proposition}[BV and POOC chain rules]
\label{prop:MainFiniteStokesChainRules}
For every even finite Stokes functional and every ordinary function
\(f\),
\begin{equation}
    \mathfrak D_A S=0,
    \label{eq:PointDerivativeFiniteStokesZeroMain}
\end{equation}
\begin{equation}
    \delta^\Sigma_{AB}f(S)
    =f'(S)\delta^\Sigma_{AB}S,
    \label{eq:AreaDerivativeLeibnizMain}
\end{equation}
\begin{equation}
    \widehat{\mathbb L}_\pm f(S)
    =f'(S)\widehat{\mathbb L}_\pm S.
    \label{eq:SMMLeibnizMain}
\end{equation}
\end{proposition}
The first identity is intraloop: the two tangent derivatives act on the
two sides of the same marked cut of the same closed loop.  It follows
from the shrinking-interval BV argument.  The \(f''\)-terms cancel in
the ordered graded commutator, and the remaining mixed terms vanish
because the scalar endpoint jump is zero.  For a finite-Stokes common
zero mode, \(\mathcal Z_\rho=e^{\rho S}\),
\begin{equation}
    \boxed{
    \widehat{\mathbb L}_\pm(\mathcal Z_\rho W)
    =\mathcal Z_\rho\widehat{\mathbb L}_\pm W.
    }
    \label{eq:SMMProductRuleZeroModeMain}
\end{equation}
The detailed proof is in Supplementary Material, Sec.~S6.

\subsection{Geodesic-wire sewing of differentiated contacts}
\label{subsec:MainSurfaceGeodesicSewing}
\label{subsec:MainWhiteBridgeSewing}

At a same-loop contact \(z(s)=z(t)\), the four BV germs
\begin{equation}
    E_s^-,\ E_s^+,\ E_t^-,\ E_t^+
    \label{eq:SewingFourTangentGerms}
\end{equation}
remain independent and are redistributed as
\begin{equation}
    \boxed{
    C_{st}^{\gamma}:(E_s^+,E_t^-),
    \qquad
    C_{ts}^{\gamma}:(E_t^+,E_s^-).
    }
    \label{eq:SewingRedistributionOfGerms}
\end{equation}
The two contours contain the same geodesic wire with opposite
orientations.  It closes the cut arcs but imposes no equality on the
four original germs.

\begin{figure}[t]
    \centering
    \includegraphics[width=0.72\textwidth]{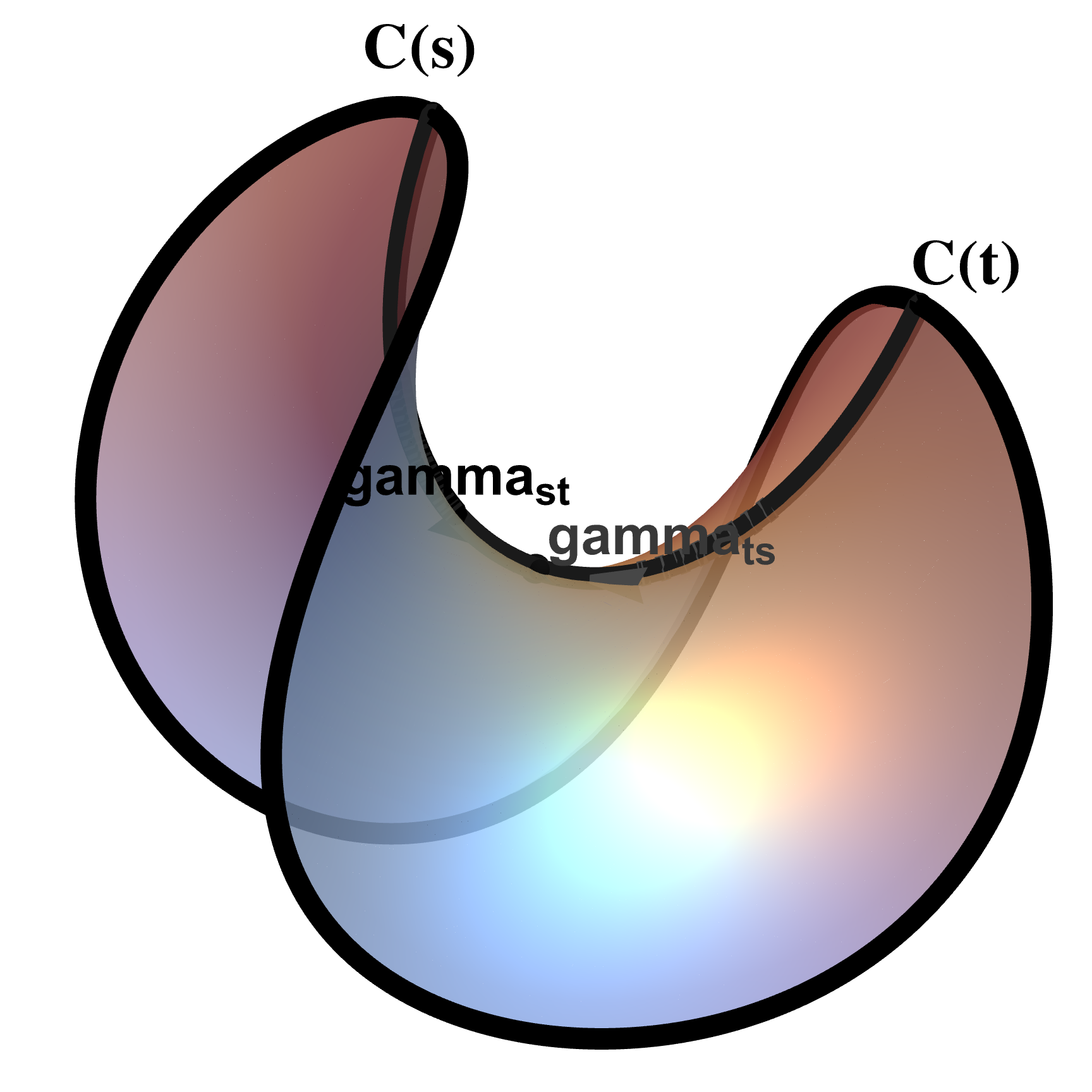}
    \caption{Geodesic-wire realization of the ultraviolet sewing
    wire.  A simply connected nonplanar minimal disk is divided by one
    geodesic joining the boundary points \(C(s)\) and \(C(t)\).  The
    same geodesic line occurs with the two opposite orientations
    \(\gamma_{st}\) and \(\gamma_{ts}=\gamma_{st}^{-1}\) in the
    daughter boundaries.  Both daughters are restrictions of the same
    parent Weierstrass map; their surface actions therefore add exactly,
    while their oriented Stokes currents on the common wire cancel
    pointwise.  The Enneper patch used in the illustration is only a
    convenient classical model of this general construction.}
    \label{fig:SurfaceGeodesicUVWire}
\end{figure}

Let \(S[C]\) be the stationary value of a local surface action on a
selected branch \(Y_C:\mathcal D_C\to\mathcal M\).  Choose a simple
geodesic cut on the parent surface \(\widehat\gamma_{st}\subset\mathcal D_C\)
whose image is the gauge wire \(\gamma_{st}\).  It partitions the
parameter disk,
\begin{equation}
    \mathcal D_C
    =\mathcal D_{st}^{\gamma}
     \cup_{\widehat\gamma_{st}}
     \mathcal D_{ts}^{\gamma},
    \label{eq:MainSurfaceDomainPartition}
\end{equation}
and defines the daughter surfaces by restriction,
\begin{equation}
    Y_{st}=Y_C|_{\mathcal D_{st}^{\gamma}},
    \qquad
    Y_{ts}=Y_C|_{\mathcal D_{ts}^{\gamma}}.
    \label{eq:MainDaughterSurfaceRestrictions}
\end{equation}
No new Dirichlet problem is solved.  The same analytic fields are used
on both domains.

At fixed ultraviolet scale define the splitting and joining defects by
\begin{equation}
    \Delta_{S,\Lambda}^{\rm split}
    =S[C_{st}^{\gamma}]
     +S[C_{ts}^{\gamma}]-S[C],
    \label{eq:MainSplittingDefect}
\end{equation}
\begin{equation}
    \Delta_{S,\Lambda}^{\rm join}
    =S[C_i\#_{s,t}^{\gamma}C_j]-S[C_i]-S[C_j].
    \label{eq:MainJoiningDefect}
\end{equation}

\begin{proposition}[Exact geodesic-wire sewing]
\label{prop:MainSurfaceGeodesicWireSewing}
Assume that the same-loop daughter surfaces are the restrictions
\eqref{eq:MainDaughterSurfaceRestrictions} of one primary stationary
surface, and that the different-loop joined branch is the disjoint
union of the two primary surfaces with the same connector traversed
once in each direction.  For either
\(
\Delta_S=\Delta_{S,\Lambda}^{\rm split}
\)
or
\(
\Delta_S=\Delta_{S,\Lambda}^{\rm join}
\), the following three conditions hold:
\begin{equation}
    \boxed{
    \Delta_S=0.
    }
    \label{eq:ExactSurfaceGeodesicDefectsMain}
\end{equation}
\begin{equation}
    \boxed{
    \mathfrak D_A(s)\Delta_S
    =\mathfrak D_A(t)\Delta_S
    =0.
    }
    \label{eq:MainSurfaceGeodesicPointDerivativeConditions}
\end{equation}
\begin{equation}
    \boxed{
    P_\alpha^{(s)}Q_{(t)}^\alpha\Delta_S=0.
    }
    \label{eq:MainProjectedSewingCondition}
\end{equation}
The anti-chiral kernel obeys the conjugate projected identity.  Hence,
for every even contact factor \(\Phi\) and every even constant \(\rho\),
\begin{equation}
    \boxed{
    K_{\pm,\Lambda}\bigl(e^{\rho\Delta_S}\Phi\bigr)
    =K_{\pm,\Lambda}\Phi.
    }
    \label{eq:ContactKernelSewingIdentityMain}
\end{equation}
\end{proposition}

\begin{proof}
The parent action is a local integral.  Equation
\eqref{eq:MainSurfaceDomainPartition} therefore gives the exact domain
identity
\begin{equation}
    S[C]
    =S[C_{st}^{\gamma}]+S[C_{ts}^{\gamma}],
    \label{eq:SameLoopAdditivityCondition}
\end{equation}
with no bridge-area correction: the common wire is one-dimensional and
is counted only as an internal cut.  For two distinct loops the selected
joined branch is the disjoint union of the two surfaces; the double wire
has zero two-dimensional measure and gives
\begin{equation}
    S[C_i\#_{s,t}^{\gamma}C_j]
    =S[C_i]+S[C_j].
    \label{eq:DifferentLoopAdditivityCondition}
\end{equation}
This proves \(\Delta_S=0\).

Scalar additivity is not by itself the differentiated-contact statement.
The Euler--Lagrange boundary variation of the stationary parent action has
the finite-Stokes form
\begin{equation}
    \delta S[C]
    =\oint_C ds\,\xi^A(s)E^B(s)\Omega_{BA}[C;s].
    \label{eq:MainSewingParentStokesVariation}
\end{equation}
The two daughter restrictions carry the same area derivative on the
common wire,
\begin{equation}
    \left.\Omega^{st}_{AB}\right|_{\gamma_{ts}}
    =\left.\Omega^{ts}_{AB}\right|_{\gamma_{st}}
    =\left.\Omega^{C}_{AB}\right|_{\gamma},
    \label{eq:MainWhiteBridgeOmegaMatching}
\end{equation}
and their wire orientations are opposite.  Their wire Stokes currents
therefore cancel pointwise.  Consequently the sum of the two daughter
actions is not merely equal in scalar value to the parent action; it is
the same finite-Stokes functional of the endpoint-dependent boundary
data.  Proposition~\ref{prop:MainFiniteStokesChainRules} gives
\begin{equation}
\begin{aligned}
    \mathfrak D_A(s)
    \bigl(S[C_{st}^{\gamma}]+S[C_{ts}^{\gamma}]\bigr)
    &=\mathfrak D_A(s)S[C]=0,
    \\
    \mathfrak D_A(t)
    \bigl(S[C_{st}^{\gamma}]+S[C_{ts}^{\gamma}]\bigr)
    &=\mathfrak D_A(t)S[C]=0.
\end{aligned}
    \label{eq:MainDaughterSumPointDerivativeConditions}
\end{equation}
Subtracting the parent contribution proves
\(
\mathfrak D_A(s)\Delta_S=\mathfrak D_A(t)\Delta_S=0
\).  The same argument applies to the disjoint-union joining identity.
In particular, the spinor components and the line-operator combination
selected by the chiral kernel satisfy
\(
P_\alpha^{(s)}\Delta_S=Q_{(t)}^\alpha\Delta_S=0
\).  Because these point-derivative equalities hold as functional
identities on the geodesic-regulated branch, applying the complementary
endpoint operator gives
\(
P_\alpha^{(s)}Q_{(t)}^\alpha\Delta_S=0
\), and similarly in the anti-chiral sector.

For completeness, if \(B_\rho=e^{\rho\Delta_S}\), the exact graded
Leibniz expansion is
\begin{align}
 &P_\alpha Q^\alpha(B_\rho\Phi)
  -B_\rho P_\alpha Q^\alpha\Phi
 \nonumber\\
 &\quad=B_\rho\Big[
   \rho(P_\alpha Q^\alpha\Delta_S)\Phi
   +\rho^2(P_\alpha\Delta_S)(Q^\alpha\Delta_S)\Phi
 \nonumber\\
 &\hspace{35mm}
   -\rho(Q^\alpha\Delta_S)P_\alpha\Phi
   +\rho(P_\alpha\Delta_S)Q^\alpha\Phi
   \Big].
    \label{eq:MainExactDifferentiatedContactCorrection}
\end{align}
The first condition sets the dressing ratio to one, the point-derivative
conditions remove the quadratic and mixed terms, and the projected
condition removes the remaining linear term.  Derivatives acting on the
geodesic length, its normalization, or the gauge transporter occur equally
in the dressed and undressed contact pairings and do not alter the
identity.
\end{proof}

\subsection{Exact finite-\texorpdfstring{\(N\)}{N} theorem}

For \(\mathbf C=(C_1,\ldots,C_n)\), define
\begin{equation}
    \mathcal Z_{n,\rho}(\mathbf C)
    =\exp\!\left[\rho\sum_{j=1}^nS[C_j]\right].
    \label{eq:MultiloopDressingFactor}
\end{equation}

\begin{theorem}[Finite-\texorpdfstring{\(N\)}{N} zero-mode dressing]
\label{thm:ProjectedFiniteNDressing}
Let \(W_{n,\Lambda}(\mathbf C)\) solve the complete finite-\(N\)
Itoyama--Takashino hierarchy at fixed ultraviolet scale.  Let \(S[C]\)
be an even finite Stokes functional satisfying
\begin{equation}
    \boxed{
    \widehat{\mathbb L}_+S
    =\widehat{\mathbb L}_-S=0
    }
    \label{eq:ProjectedZeroModeConditions}
\end{equation}
and suppose that its selected surface branch obeys Proposition
\ref{prop:MainSurfaceGeodesicWireSewing} for the same-loop split and
different-loop join.  Then
\begin{equation}
    \boxed{
    \widetilde W_{n,\Lambda}(\mathbf C)
    =\mathcal Z_{n,\rho}(\mathbf C)W_{n,\Lambda}(\mathbf C)
    }
    \label{eq:DressedRegularizedMultiloopCorrelator}
\end{equation}
solves the same chiral and anti-chiral hierarchy for every even
constant \(\rho\).
\end{theorem}

\begin{proof}
Equation~\eqref{eq:SMMProductRuleZeroModeMain} supplies the same factor
on every local left-hand side.  Proposition
\ref{prop:MainSurfaceGeodesicWireSewing} says that the dressing ratio is
identically one for every regulated split and join and that the
differentiated kernels pass through it exactly.  The \(SU(N)\)
reconnection terms preserve the original contour list.  Hence every
term in every equation acquires the same scalar factor, which cancels.
\end{proof}

In the planar limit,
\begin{equation}
    \boxed{
    W_{\rho,\Lambda}[C]
    =W_{0,\Lambda}[C]e^{\rho S[C]}.
    }
    \label{eq:ProjectedPlanarDressedSolution}
\end{equation}
For several independent common zero modes one may use
\begin{equation}
    \boxed{
    \exp\!\left[\sum_r\rho_rS_r[C]\right].
    }
    \label{eq:MultipleZeroModeDressing}
\end{equation}
The theorem is exact at fixed ultraviolet scale and imposes no reality
condition.  The Lorentzian coefficients are chosen later so that a
long timelike rectangle obeys
\begin{equation}
    e^{\rho S[C_{T,L}]}
    =e^{-iTE(L)},
    \qquad E(L)=\sigma L>0.
    \label{eq:LorentzianStaticEnergyConditionDressing}
\end{equation}

\section{Lorentzian supersymmetric Hodge-dual surface and the Super--Weierstrass representation}
\label{sec:SuperHodgeDualSurface}

We now construct the geometric finite-Stokes zero modes to which
Theorem~\ref{thm:ProjectedFiniteNDressing} will be applied.  The
supersymmetrization problem is fixed by the loop equation itself.  We
do not need to postulate a fundamental Green--Schwarz world sheet, a
world-sheet superspace, or an independent chiral
``completion'' of a bosonic current.  What is required is an even
geometric functional of the physical superloop whose
four-dimensional boundary area derivative belongs to one Lorentzian
Hodge sector,
\begin{equation}
    *_4\Omega^{(\chi)}
    =-i\chi\,\Omega^{(\chi)},
    \qquad
    \chi=\pm1,
    \label{eq:SHDSupersymmetrizationCriterion}
\end{equation}
and whose stationary conformal branch admits a Dirichlet and
Weierstrass representation.  Equation
\eqref{eq:SHDSupersymmetrizationCriterion}, together with the ordered
Jacobi identity established above, is the geometric input needed by
the supersymmetric loop equation.

The minimal supersymmetric extension is obtained by allowing the
auxiliary surface coordinates to take values in the even part of the
Grassmann algebra generated by the boundary superloop and by replacing
the ordinary boundary differential by the invariant superspace
one-form.  The bulk parameter domain remains an ordinary
2-dimensional disk.  Thus, in this section a \emph{surface
superfield} means an ordinary world-sheet field with values in the even Grassmann
algebra.  Target-space supersymmetry enters through
the boundary data; no independent odd world-sheet coordinates are
required.

The construction is related to the generalized super-Gauss-map,
doubly supersymmetric, and conformally parametrized supersurface
formalisms
\cite{ViswanathanParthasarathy1992SuperGauss,
BandosEtAl1995DoublySupersymmetric,
BertrandGrundlandHariton2015SuperSurfaces}.  We use the same
twistor factorization of null conformal currents, but the
present surface is the auxiliary
$\mathbb R^3\otimes\mathbb R^{1,3}$ Hodge surface required by loop
space rather than a fundamental superstring.

\subsection{Invariant boundary curve of the superloop}
\label{subsec:InvariantBoundaryCurveAndBosonicHDInput}

Let the physical boundary be the closed superloop
\begin{equation}
    C:\quad
    s\longmapsto
    z^M(s)
    =
    \left(
        x^a(s),
        \eta^\alpha(s),
        \bar\eta^{\dot\alpha}(s)
    \right),
    \qquad
    z^M(s+2\pi)=z^M(s).
    \label{eq:PhysicalSuperloopSHD}
\end{equation}
In the conventions of Section~\ref{sec:LoopCalculus}, the invariant
vector one-form is
\begin{equation}
    \Pi^a
    \equiv e^a
    =
    dx^a
    -i\,d\eta\,\sigma^a\bar\eta
    +i\,\eta\sigma^a d\bar\eta .
    \label{eq:InvariantOneFormSHD}
\end{equation}
Its contour pullback is
\begin{equation}
    \Pi^a\big|_C
    =ds\,E^a(s),
    \qquad
    E^a
    =
    \dot x^a
    -i\dot\eta\sigma^a\bar\eta
    +i\eta\sigma^a\dot{\bar\eta}.
    \label{eq:InvariantVectorTangentSHD}
\end{equation}
We retain the BV loop class of
\eqref{eq:BVSuperloopClass}: the physical superloop is continuous and
closed, whereas \(E^a(s)\) is a periodic BV symbol with independent
one-sided traces at a marked cut.

The even boundary curve seen by the auxiliary surface is obtained by
integrating the invariant tangent:
\begin{equation}
    \mathcal C^a(s)
    =
    \mathcal C^a(s_0)
    +\int_{s_0}^{s}du\,E^a(u),
    \qquad
    d\mathcal C^a=\Pi^a\big|_C.
    \label{eq:InvariantDevelopedCurve}
\end{equation}
The additive constant is immaterial.  Since $\Pi^a$ is not an exact
one-form on target superspace, periodicity of $z^M(s)$ does not by
itself imply closure of this boundary curve.  We therefore work with superloops satisfying the invariant closure condition,
\begin{equation}
    \oint_C\Pi^a
    =\int_0^{2\pi}ds\,E^a(s)=0.
    \label{eq:InvariantClosureCondition}
\end{equation}

Let $\mathscr G$ be the Grassmann algebra generated by the fermionic
boundary data.  Then
$\mathcal C^a,E^a\in\mathscr G_{\bar0}$.  A variation of the physical
superloop induces the corresponding boundary-curve variation
\begin{equation}
    \zeta^a(s)=\delta\mathcal C^a(s),
    \qquad
    \partial_s\zeta^a=\delta E^a,
    \label{eq:DevelopedCurveVariation}
\end{equation}
with
\begin{equation}
    \zeta^a(2\pi)=\zeta^a(0),
    \qquad
    \int_0^{2\pi}ds\,\partial_s\zeta^a=0.
    \label{eq:DevelopedVariationClosure}
\end{equation}
Its derivative \(\partial_s\zeta^a\) may nevertheless have different
one-sided traces at the marked cut.  The supersymmetric boundary rule
is therefore
\begin{equation}
    dC^a\ \longrightarrow\ \Pi^a\big|_C.
    \label{eq:BoundarySupersymmetrizationRule}
\end{equation}

\subsection{Lorentzian Hodge frame and the off-shell area form}
\label{subsec:LorentzianHodgeFrame}
\label{subsec:HodgeDualExtremalSurface}

For a complex spacetime two-form $B_{ab}$, define
\begin{equation}
    (*_4B)_{ab}
    =\frac12\epsilon_{ab}{}^{cd}B_{cd},
    \qquad
    \eta_{ab}=\operatorname{diag}(-1,+1,+1,+1),
    \qquad
    \epsilon_{0123}=+1.
    \label{eq:LorentzianHodgeStarDefinition}
\end{equation}
Then $(*_4)^2=-1$ on two-forms.  The two complex Hodge projectors are
\begin{equation}
    P_\chi
    =\frac12\left(1+i\chi *_4\right),
    \qquad
    P_\chi^2=P_\chi,
    \qquad
    *_4P_\chi=-i\chi P_\chi.
    \label{eq:HodgeProjector}
\end{equation}

Choose the explicit Lorentzian continuation of the $'t$~Hooft frame
\begin{equation}
    \mathfrak e_{0j}^{\chi,i}
    =i\,\delta^i{}_j,
    \qquad
    \mathfrak e_{jk}^{\chi,i}
    =\chi\,\epsilon^i{}_{jk},
    \qquad
    \mathfrak e_{ab}^{\chi,i}
    =-\mathfrak e_{ba}^{\chi,i},
    \label{eq:ExplicitLorentzianTHooftBasis}
\end{equation}
where $i,j,k=1,2,3$.  It satisfies
\begin{equation}
    *_4\mathfrak e^{\chi,i}
    =-i\chi\,\mathfrak e^{\chi,i}
    \label{eq:LorentzianTHooftHodgeEigenvalue}
\end{equation}
and the crossed Fierz identity
\begin{equation}
\begin{aligned}
    \sum_{i=1}^{3}
    \mathfrak e_{a\rho}^{\chi,i}
    \mathfrak e_{b\sigma}^{\chi,i}
    ={}&
    \eta_{ab}\eta_{\rho\sigma}
    -\eta_{a\sigma}\eta_{b\rho}
    -i\chi\epsilon_{ab\rho\sigma}.
\end{aligned}
    \label{eq:LorentzianHodgeFrameProjectionIdentity}
\end{equation}
For a general antisymmetric tensor $n^{ab}$, the right-hand side of
\eqref{eq:LorentzianHodgeFrameProjectionIdentity} gives
$n_{\rho\sigma}-2i\chi(*_4n)_{\rho\sigma}$; it is not by itself a
Hodge projector.  If, however,
$*_4n=-i\chi n$, then
\begin{equation}
    \sum_i
    \mathfrak e_{a\rho}^{\chi,i}
    \mathfrak e_{b\sigma}^{\chi,i}n^{ab}
    =-n_{\rho\sigma}.
    \label{eq:CrossedIdentityOnHodgeSector}
\end{equation}
This elementary distinction is essential when the conformal solution
is continued off shell.

For each $\chi=\pm1$, introduce twelve independent even fields
\begin{equation}
    X_a^{\chi,i}(\xi)
    \in\mathscr G_{\bar0},
    \qquad
    a=0,1,2,3,
    \qquad
    i=1,2,3,
    \qquad
    \xi=(\xi^1,\xi^2)\in\mathcal D.
    \label{eq:HDEmbedding}
\end{equation}
Their raw extended two-form is
\begin{equation}
    \Xi_{ab}[X_\chi]
    =
    \epsilon^{mn}
    \sum_{i=1}^{3}
    \partial_mX_a^{\chi,i}
    \partial_nX_b^{\chi,i}.
    \label{eq:HDExtendedTangentTwoForm}
\end{equation}
The Hodge area form entering the off-shell functional is
\begin{equation}
    \Sigma_{ab}^{(\chi)}[X]
    =
    (P_\chi\Xi[X_\chi])_{ab},
    \qquad
    *_4\Sigma^{(\chi)}=-i\chi\Sigma^{(\chi)}.
    \label{eq:HDProjectedAreaForm}
\end{equation}
On the original constrained Hodge branch
$P_\chi\Xi=\Xi$, this definition is identical to the bosonic
square-root functional.  Writing the projector explicitly is the
minimal off-shell continuation that preserves the one-Hodge-sector
boundary flux away from that branch.

Define
\begin{equation}
    \boxed{
    S_\chi^{\rm HD}[X]
    =
    -\nu_{\rm HD}
    \int_{\mathcal D}d^2\xi\,
    \sqrt{
        -\frac12
        \Sigma_{ab}^{(\chi)}[X]
        \Sigma^{(\chi)ab}[X]
    } .
    }
    \label{eq:HDAreaFunctional}
\end{equation}
If the argument of the square root is
$A=A_{\rm body}+A_{\rm nil}$ with $A_{\rm body}\ne0$, the chosen
Lorentzian branch is defined by the finite nilpotent expansion of
$\sqrt{A_{\rm body}+A_{\rm nil}}$.

Supersymmetry is imposed through the boundary locking
\begin{equation}
    \boxed{
    dX_a^{\chi,i}\big|_{\partial\mathcal D}
    =
    \mathfrak e_{ab}^{\chi,i}\,
    \Pi^b\big|_C .
    }
    \label{eq:HDBoundaryLockingDifferential}
\end{equation}
Equivalently, after fixing an irrelevant integration constant,
\begin{equation}
    X_a^{\chi,i}\big|_{\partial\mathcal D}
    =
    \mathfrak e_{ab}^{\chi,i}\mathcal C^b.
    \label{eq:HDBoundaryLocking}
\end{equation}
No equation
$dX_a^{\chi,i}=\mathfrak e_{ab}^{\chi,i}\Pi^b$ is imposed in the
interior.  Such an extension would be inconsistent because $d^2X=0$
whereas $d\Pi^a$ is the nonzero flat-superspace torsion two-form.

Set
\begin{equation}
    \mathcal L_\chi
    =
    \sqrt{-\Sigma^{(\chi)}_{ab}\Sigma^{(\chi)ab}/2},
    \qquad
    n_{(\chi)}^{ab}
    =\frac{\Sigma^{(\chi)ab}}{\mathcal L_\chi}.
    \label{eq:NormalizedHDBoundaryFlux}
\end{equation}
Since $P_\chi$ is self-adjoint on two-forms and
$n_{(\chi)}$ already lies in its image, independent variation of all
twelve fields gives
\begin{equation}
    \epsilon^{mn}\partial_m
    \left(
        n_{(\chi)}^{ab}
        \partial_nX_b^{\chi,i}
    \right)=0,
    \qquad i=1,2,3.
    \label{eq:HDEulerLagrangeEquation}
\end{equation}
The surface is therefore determined by a genuine bulk variational
problem rather than prescribed by the boundary superloop.

\subsection{Boundary variation and Hodge duality of the area derivative}
\label{subsec:BoundaryVariationAndHodgeAreaDerivative}

On a stationary surface the bulk term in the first variation vanishes,
and the remaining boundary flux is
\begin{equation}
    \delta S_\chi^{\rm HD}
    =
    \nu_{\rm HD}
    \oint_{\partial\mathcal D}ds\,
    \delta X_a^{\chi,i}
    n_{(\chi)}^{ab}
    \partial_sX_b^{\chi,i},
    \label{eq:HDBoundaryVariation}
\end{equation}
up to the common orientation convention.  At the boundary,
\begin{equation}
    \delta X_a^{\chi,i}
    =\mathfrak e_{a\rho}^{\chi,i}\zeta^\rho,
    \qquad
    \partial_sX_b^{\chi,i}
    =\mathfrak e_{b\sigma}^{\chi,i}E^\sigma.
    \label{eq:HDBoundaryVariationLocking}
\end{equation}
Because $n_{(\chi)}$ is a Hodge eigenform,
\eqref{eq:CrossedIdentityOnHodgeSector} gives
\begin{equation}
    \delta S_\chi^{\rm HD}[C]
    =
    \oint_Cds\,
    \zeta^\rho(s)E^\sigma(s)
    \Omega_{\sigma\rho}^{(\chi)}[C;s],
    \label{eq:HDBosonicStokesVariation}
\end{equation}
where, after absorbing the orientation sign into $\nu_{\rm HD}$,
\begin{equation}
    \Omega_{ab}^{(\chi)}
    =\nu_{\rm HD}\,n_{ab}^{(\chi)}.
    \label{eq:HDBoundaryAreaDerivative}
\end{equation}
Consequently,
\begin{equation}
    \boxed{
    *_4\Omega^{(\chi)}
    =-i\chi\,\Omega^{(\chi)}.
    }
    \label{eq:HodgeDualAreaDerivative}
\end{equation}
The Hodge star acts only on the four Lorentz-vector indices and is
blind to the even Grassmann coefficients.  Hence the complete
supersymmetric dependence may enter through $\Pi^a$ without changing
the four-dimensional Hodge argument.

Equation~\eqref{eq:HodgeDualAreaDerivative} has been derived before
any conformal, Dirichlet, or Weierstrass reduction.  This order is
essential: the loop-space area derivative must be taken from the
unreduced Hodge functional, not from its final scalar Dirichlet value.

\subsection{An off-shell lift of the conformal branch}
\label{subsec:SHDVirasoroReduction}

To exhibit a solvable family of stationary surfaces, introduce an
even four-vector field $Y^a(\xi)$ and an even world-sheet density
$\Lambda_{\chi,i}^{a}(\xi)$.  The parent functional
\begin{equation}
\begin{aligned}
    \mathscr S_\chi[X,Y,\Lambda]
    ={}&
    S_\chi^{\rm HD}[X]
    +\int_{\mathcal D}d^2\xi\,
    \Lambda_{\chi,i}^{a}
    \left(
        X_a^{\chi,i}
        -\mathfrak e_{ab}^{\chi,i}Y^b
    \right)
\end{aligned}
    \label{eq:SHDParentFunctional}
\end{equation}
provides an off-shell lift of the conformal branch.  The multiplier
term contains no derivatives and therefore contributes no boundary
flux.  Its equations impose
\begin{equation}
    X_a^{\chi,i}
    =\mathfrak e_{ab}^{\chi,i}Y^b.
    \label{eq:OnShellHodgeParametrization}
\end{equation}
The reduced field obeys the Dirichlet boundary condition
\begin{equation}
    Y^a\big|_{\partial\mathcal D}=\mathcal C^a.
    \label{eq:HDBoundaryConditionY}
\end{equation}
This is a construction of one explicitly integrable stationary family;
it is not an assertion that every solution of
\eqref{eq:HDEulerLagrangeEquation} has this form.

Let $\mathcal E_{\chi,i}^{a}[X]$ denote the Euler--Lagrange derivative
of $S_\chi^{\rm HD}[X]$.  The remaining parent equations are
\begin{equation}
    \Lambda_{\chi,i}^{a}
    =-\mathcal E_{\chi,i}^{a}[X],
    \qquad
    \sum_i
    \mathfrak e_{ab}^{\chi,i}
    \Lambda_{\chi,i}^{a}=0.
    \label{eq:SHDParentEquations}
\end{equation}
The second equation is exactly the Euler--Lagrange equation obtained
by varying the reduced functional
$S_\chi^{\rm HD}[\mathfrak eY]$ with respect to $Y$.  Hence every
stationary $Y$ of the reduced problem lifts to a stationary point of
the parent problem by choosing
$\Lambda=-\mathcal E[\mathfrak eY]$.  The twelve-field boundary
momentum, and therefore the Hodge area derivative, is retained.

Define the simple two-form of $Y$ by
\begin{equation}
    F_{ab}[Y]
    =\epsilon^{mn}
    \partial_mY_a\partial_nY_b.
    \label{eq:SHDSimpleYBivector}
\end{equation}
The crossed identity gives the exact Lorentzian algebra
\begin{equation}
    \Xi_{ab}[\mathfrak eY]
    =F_{ab}-2i\chi(*_4F)_{ab},
    \qquad
    \Sigma_{ab}^{(\chi)}[\mathfrak eY]
    =-(P_\chi F)_{ab}.
    \label{eq:HDProjectedBranchAlgebra}
\end{equation}
Since $F$ is simple,
\begin{equation}
    F_{ab}(*_4F)^{ab}=0,
    \qquad
    (*_4F)_{ab}(*_4F)^{ab}=-F_{ab}F^{ab}.
    \label{eq:SimpleBivectorHodgeIdentities}
\end{equation}
It follows that
\begin{equation}
    \Sigma_{ab}^{(\chi)}
    \Sigma^{(\chi)ab}
    =\frac12F_{ab}F^{ab}.
    \label{eq:HDDiscriminantReduction}
\end{equation}
Thus the projected square root is precisely the ordinary
2-dimensional area discriminant of $Y$, up to a fixed normalization.

Let $z$ and $\widetilde z$ be independent coordinates on the
complexified world sheet and set
\begin{equation}
    p^a=\partial_zY^a,
    \qquad
    q^a=\partial_{\widetilde z}Y^a.
    \label{eq:SHDConformalTangents}
\end{equation}
Equation~\eqref{eq:HDDiscriminantReduction} gives
\begin{equation}
    S_\chi^{\rm HD}[\mathfrak eY]
    =
    T_{\rm NG}
    \int dz\,d\widetilde z\,
    \sqrt{(p\cdot q)^2-p^2q^2},
    \label{eq:SHDReducedNambuGotoFunctional}
\end{equation}
where $T_{\rm NG}$ includes the constant normalization, coordinate
Jacobian, and Lorentzian square-root phase.  The result is independent
of $\chi$.

On every regular patch, 2-dimensional conformal flatness permits the
Virasoro gauge
\begin{equation}
    \boxed{
    p^2=0,
    \qquad
    q^2=0.
    }
    \label{eq:SHDVirasoroConstraints}
\end{equation}
For an even Grassmann-valued metric, conformal coordinates for the
body extend recursively through the finite nilpotent ideal.  On this
gauge slice,
\begin{equation}
    \boxed{
    S_\chi^{\rm HD}[Y]\big|_{\rm Vir}
    =
    T_{\rm D}
    \int_{\mathcal D}dz\,d\widetilde z\,
    \partial_zY^a\partial_{\widetilde z}Y_a,
    }
    \label{eq:SHDDirichletFunctional}
\end{equation}
with the orientation and square-root branch absorbed into $T_{\rm D}$.
The Euler--Lagrange equation is
\begin{equation}
    \partial_z\partial_{\widetilde z}Y^a=0.
    \label{eq:SHDDirichletEquation}
\end{equation}
Equations~\eqref{eq:SHDParentEquations} show that this Dirichlet
solution is a stationary solution of the parent variational problem,
not merely of an unrelated reduced functional.

\subsection{The Super--Weierstrass family and supertwistor variables}
\label{subsec:SHDDirichletReduction}
\label{subsec:SHDSuperWeierstrass}

The local solution of \eqref{eq:SHDDirichletEquation} is
\begin{equation}
    \boxed{
    Y^a(z,\widetilde z)
    =f^a(z)+\widetilde f^{\,a}(\widetilde z),
    }
    \label{eq:SHDWeierstrassSolution}
\end{equation}
where $f^a$ and $\widetilde f^{\,a}$ are even
Grassmann-valued analytic functions.  The Virasoro constraints become
\begin{equation}
    f'{}^2(z)=0,
    \qquad
    \widetilde f'{}^2(\widetilde z)=0.
    \label{eq:SHDNullWeierstrassCurrents}
\end{equation}
Thus the two conformal currents are complex lightlike vectors in
four-dimensional Minkowski space with even supernumber coefficients.

In bispinor notation, a regular null $2\times2$ matrix has rank one.
Over the local even Grassmann algebra this gives
\begin{equation}
    \boxed{
    f'_{\alpha\dot\alpha}(z)
    =\lambda_\alpha(z)
    \widetilde\lambda_{\dot\alpha}(z),
    \qquad
    \widetilde f'_{\alpha\dot\alpha}(\widetilde z)
    =\rho_\alpha(\widetilde z)
    \widetilde\rho_{\dot\alpha}(\widetilde z).
    }
    \label{eq:SHDSuperTwistorFactorization}
\end{equation}
The four spinors are commuting even superfields.  The factorization
exists locally whenever the body of the null current has rank one.
The only local redundancy is
\begin{equation}
\begin{aligned}
    \lambda&\longmapsto g_+(z)\lambda,
    &\qquad
    \widetilde\lambda&\longmapsto
    g_+^{-1}(z)\widetilde\lambda,
    \\
    \rho&\longmapsto g_-(\widetilde z)\rho,
    &
    \widetilde\rho&\longmapsto
    g_-^{-1}(\widetilde z)\widetilde\rho,
\end{aligned}
    \label{eq:SHDTwistorRescalings}
\end{equation}
where $g_\pm$ are invertible even functions.

\paragraph{Free BV twistor traces.}
The analytic currents in
\eqref{eq:SHDSuperTwistorFactorization} are used through their Hardy
boundary values, understood as BV symbols on the oriented circle.
After choosing the marked cut at \(\theta=0\), the one-sided values
\begin{equation}
\begin{gathered}
    \lambda(0^+),\quad \lambda(2\pi^-),\qquad
    \widetilde\lambda(0^+),\quad
    \widetilde\lambda(2\pi^-),
    \\
    \rho(0^+),\quad \rho(2\pi^-),\qquad
    \widetilde\rho(0^+),\quad
    \widetilde\rho(2\pi^-)
\end{gathered}
    \label{eq:FreeTwistorBoundaryTraces}
\end{equation}
are independent free traces, modulo the local projective rescalings
\eqref{eq:SHDTwistorRescalings}.  In particular, no condition of the
form
\begin{equation}
    \lambda(2\pi^-)=\pm\lambda(0^+)
    \label{eq:NoTwistorParityCondition}
\end{equation}
is imposed, and similarly for the other twistor factors.  Their
integer Hardy--Fourier series are Fourier representations of periodic
BV symbols; at a jump the series converges to the average of the two
traces and therefore does not identify them.  Holomorphy constrains
the Hardy content of the null currents, while closure imposes one
global bilinear zero-mode condition.  Neither condition creates a
periodic or anti-periodic twistor sector.

Along the boundary, introduce the chiral and anti-chiral incidence
coordinates
\begin{equation}
    x_+^{\alpha\dot\alpha}
    =x^{\alpha\dot\alpha}
    +i\eta^\alpha\bar\eta^{\dot\alpha},
    \qquad
    x_-^{\alpha\dot\alpha}
    =x^{\alpha\dot\alpha}
    -i\eta^\alpha\bar\eta^{\dot\alpha}.
    \label{eq:SHDChiralIncidenceCoordinates}
\end{equation}
For the first null current, define
\begin{equation}
\begin{aligned}
    \mu^{\dot\alpha}
    &=x_+^{\beta\dot\alpha}\lambda_\beta,
    &\qquad
    \vartheta&=\eta^\beta\lambda_\beta,
    \\
    \widetilde\mu^{\alpha}
    &=-x_-^{\alpha\dot\beta}
      \widetilde\lambda_{\dot\beta},
    &
    \widetilde\vartheta
    &=\bar\eta^{\dot\beta}
      \widetilde\lambda_{\dot\beta}.
\end{aligned}
    \label{eq:SHDSupertwistorIncidence}
\end{equation}
Together with the projective rescaling
\eqref{eq:SHDTwistorRescalings}, these are the $\mathcal N=1$
superambitwistor incidence data of the null tangent.  The second
current is treated identically with
$(\lambda,\widetilde\lambda)$ replaced by
$(\rho,\widetilde\rho)$.  The odd coordinates are therefore not an
independent superfield strength appended to a bosonic surface; they
are the incidence data of the same physical superloop whose invariant
one-form appears in the boundary condition.

For a disk boundary, take
\begin{equation}
    z=e^{i\theta},
    \qquad
    \widetilde z=e^{-i\theta}.
    \label{eq:SHDBoundaryConformalCoordinates}
\end{equation}
Then the invariant boundary tangent is generated directly by the two
null currents:
\begin{equation}
    \boxed{
    \frac{d\mathcal C_{\alpha\dot\alpha}}{d\theta}
    =
    iz\,\lambda_\alpha(z)
    \widetilde\lambda_{\dot\alpha}(z)
    -i\widetilde z\,
    \rho_\alpha(\widetilde z)
    \widetilde\rho_{\dot\alpha}(\widetilde z).
    }
    \label{eq:SHDTwistorGeneratedBoundaryTangent}
\end{equation}
Integration gives $\mathcal C(\theta)$.  The physical bosonic
coordinate is recovered from
\begin{equation}
    dx^a
    =d\mathcal C^a
    +i\,d\eta\sigma^a\bar\eta
    -i\,\eta\sigma^a d\bar\eta.
    \label{eq:SHDRecoverPhysicalBosonicLoop}
\end{equation}
Closure imposes
\begin{equation}
    \int_0^{2\pi}d\theta\,
    \frac{d\mathcal C^a}{d\theta}=0,
    \label{eq:SHDTwistorClosureCondition}
\end{equation}
together with the required Lorentzian reality conditions.  For an
integer Fourier representation
\begin{equation}
\begin{aligned}
    \lambda(\theta)&=\sum_r\lambda_r e^{ir\theta},
    &\qquad
    \widetilde\lambda(\theta)
      &=\sum_r\widetilde\lambda_r e^{ir\theta},
    \\
    \rho(\theta)&=\sum_r\rho_r e^{ir\theta},
    &
    \widetilde\rho(\theta)
      &=\sum_r\widetilde\rho_r e^{ir\theta},
\end{aligned}
    \label{eq:TwistorBVFourierExpansions}
\end{equation}
this is the bilinear zero-mode constraint
\begin{equation}
    \boxed{
    \sum_r
    \lambda_{\alpha,r}
    \widetilde\lambda_{\dot\alpha,-r-1}
    =
    \sum_r
    \rho_{\alpha,r}
    \widetilde\rho_{\dot\alpha,1-r}.
    }
    \label{eq:TwistorBilinearClosureConstraint}
\end{equation}
The mode ranges in
\eqref{eq:TwistorBVFourierExpansions} are restricted by the chosen
Hardy charts and projective gauge, but the closure relation is always
the vanishing of the zero mode of the bilinear current, with the sums
understood in the BV/Hardy convolution sense.

There is no boundary reparametrization $\tau(\phi)$ to determine.  We
are not given an unparametrized geometric curve and asked to recover
its conformal parameter by minimizing a Douglas functional.  The
direction of construction is
\begin{equation}
    \boxed{
    \text{supertwistor data}
    \longrightarrow
    f',\widetilde f'
    \longrightarrow
    \mathcal C'(\theta)
    \longrightarrow
    (\mathcal C,\eta,\bar\eta)
    \longrightarrow C.
    }
    \label{eq:SHDForwardTwistorConstruction}
\end{equation}
We parametrize a family of conformal stationary surfaces together with
their boundary superloops.  The inverse Douglas problem is bypassed
rather than solved.

\subsection{Fourier--Hilbert representation of the planar superarea}
\label{subsec:SHDPlanarSuperarea}

On the Virasoro--Weierstrass branch the reduction to the Dirichlet
functional is exact.  Its stationary value is therefore evaluated
without expanding the square root of the parent functional.  For the
invariant boundary curve \(\mathcal C^a(\theta)\), with
\(E^a=\dot{\mathcal C}^{\,a}\) and
\(\int_0^{2\pi}E^a d\theta=0\), one obtains
\begin{equation}
    \boxed{
    \begin{aligned}
    D[\mathcal C]
    &=\frac12\int_0^{2\pi}d\theta\,
      \mathcal C^a\mathsf H_sE_a
      =\frac12\int_0^{2\pi}d\theta\,
      E^a|\partial_\theta|^{-1}E_a\\
    &=2\pi\sum_{n\ge1}n\,
      \mathcal C_n^a\mathcal C_{-n,a}.
    \end{aligned}
    }
    \label{eq:ExactBoundaryDirichletFunctional}
\end{equation}
Here \(\mathsf H_s\) is the circular Hilbert transform acting
coefficientwise on the even Grassmann algebra, and the Fourier series
is understood in the periodic-BV sense.  The harmonic extension,
Hardy projections, Fourier Virasoro equations, and Green-identity
derivation are given in Supplementary Material,
Sec.~S10.

If the bosonic body is planar, write
\begin{equation}
    E^a=\dot x^a+\Upsilon^a,
    \qquad
    \Upsilon^a
    =-i\dot\eta\sigma^a\bar\eta
      +i\eta\sigma^a\dot{\bar\eta}.
    \label{eq:ExactFermionicDevelopment}
\end{equation}
Then the exact quadratic boundary functional terminates at fourth
Grassmann order,
\begin{equation}
    \boxed{
    D[\mathcal C]
    =A_{\rm pl}[x]
     +\Delta_\eta^{(2)}
     +\Delta_\eta^{(4)},
    }
    \label{eq:ExactPlanarBodySuperareaDecomposition}
\end{equation}
with
\begin{equation}
\begin{aligned}
    \Delta_\eta^{(2)}
    &=\int_0^{2\pi}d\theta\,
      x^a\mathsf H_s\Upsilon_a,\\
    \Delta_\eta^{(4)}
    &=\frac12\int_0^{2\pi}d\theta\,
      \Upsilon^a|\partial_\theta|^{-1}\Upsilon_a.
\end{aligned}
    \label{eq:ExactPlanarBodyFermionicCorrections}
\end{equation}
No higher terms occur because \(D\) is exactly quadratic and
\(\Upsilon\) is already fermion-bilinear.

With the normalization inherited from the bosonic Hodge surface,
\begin{equation}
    \boxed{
    S_+[C_{\rm pl}]=S_-[C_{\rm pl}]
    =2\sqrt2\,D[\mathcal C].
    }
    \label{eq:ExactPlanarSHDSuperareaValue}
\end{equation}
This equality concerns stationary scalar values only; the area
derivative is still obtained from the unreduced twelve-field parent
functional.  For a Euclidean rectangular body,
\begin{equation}
    S_+[C_{T_E,L}]=S_-[C_{T_E,L}]
    =2\sqrt2\left(LT_E+\Delta_\eta^{(2)}
                         +\Delta_\eta^{(4)}\right),
    \label{eq:ExactSuperRectangleArea}
\end{equation}
and for vanishing or constant boundary fermions
\begin{equation}
    \boxed{\sigma=2\sqrt2\,\kappa.}
    \label{eq:ExactPlanarStringTensionNormalization}
\end{equation}
Nonconstant boundary fermions multiply the bosonic rectangle by the
finite nilpotent factor
\begin{equation}
    \exp\!\left[-2\sqrt2\,\kappa
    \left(\Delta_\eta^{(2)}+\Delta_\eta^{(4)}\right)\right].
    \label{eq:ExactNilpotentRectangleFactor}
\end{equation}

\subsection{The unreduced functional and the hidden Hodge variables}
\label{subsec:NoReducedYSurface}
\label{subsec:SHDHiddenVariables}

The Dirichlet representation
\eqref{eq:SHDDirichletFunctional} is used to solve the stationary
surface, but it is not the definition of the loop functional for the
purpose of taking area derivatives.  The definition remains the
stationary value of the parent functional
\eqref{eq:SHDParentFunctional}, with twelve independently varied
$X_a^{\chi,i}$ fields, the projected square root
\eqref{eq:HDAreaFunctional}, and the Hodge-locked boundary condition
\eqref{eq:HDBoundaryLocking}.  The multiplier term vanishes on shell
and has no boundary derivative, but it provides a legitimate off-shell
lift of the Weierstrass branch.

The two Hodge chiralities have the same reduced Dirichlet equations and
the same stationary scalar value.  Nevertheless, their boundary
conjugate momenta differ because the map from the physical boundary
variation to the twelve hidden variables contains
$\mathfrak e^{\chi,i}$.  Their area derivatives therefore belong to
opposite Hodge sectors.  Reducing first to the scalar Dirichlet value
and differentiating only that value would discard precisely the
boundary Hessian data required by the loop equation.  The correct
order is
\begin{equation}
    \boxed{
    \begin{gathered}
    \text{vary the full Hodge-projected twelve-field parent
    functional,}\\
    \text{then evaluate the answer on the
    Virasoro--Weierstrass branch.}
    \end{gathered}
    }
    \label{eq:SHDCorrectOrderOfOperations}
\end{equation}

\subsection{Finite Stokes functionals and Lorentzian reality}
\label{subsec:SHDStokesIntegrability}

On the regular Super--Weierstrass branch constructed in
Subsections~\ref{subsec:SHDVirasoroReduction} and
\ref{subsec:SHDSuperWeierstrass}, let $X_\chi[C]$, $Y_\chi[C]$, and
$\Lambda_\chi[C]$ denote the corresponding stationary parent
solution with boundary data \eqref{eq:HDBoundaryLocking}.  Its
stationary value defines
\begin{equation}
    S_\chi[C]
    =\mathscr S_\chi
    [X_\chi[C],Y_\chi[C],\Lambda_\chi[C]].
    \label{eq:SHDStationaryLoopFunctional}
\end{equation}
On every regular branch, the on-shell boundary variation and
\eqref{eq:HDBosonicStokesVariation} give
\begin{equation}
    \delta S_\chi[C]
    =
    \oint_Cds\,
    \delta\mathcal C^a E^b
    \Omega_{ba}^{(\chi)}.
    \label{eq:SHDFullStokesVariation}
\end{equation}
The current is finite wherever the body of the stationary surface is
nondegenerate.  Thus $S_\chi$ is an even finite Stokes functional in
the invariant boundary-curve variables.  Because $S_\chi[C]$ is
defined as the stationary value of the parent functional,
\eqref{eq:SHDFullStokesVariation} is its exact first variation by the on-shell boundary variation; no independent loop-space curl or
``functional-integrability'' condition is imposed.  Its superspace
dependence is already contained in $d\mathcal C=\Pi|_C$.

On a physical Lorentzian real slice, choose conjugate Hodge frames,
stationary branches, and additive constants so that
\begin{equation}
    S_-[C]
    =\bigl(S_+[C]\bigr)^\dagger.
    \label{eq:SHDRealityCondition}
\end{equation}
The two complex functionals are the conjugate Hodge components from
which the real Lorentzian Wilson-loop phase is formed.

\subsection{Hodge--Jacobi identity and the common zero modes}
\label{subsec:SHDProjectedZeroModes}

The path-ordered graded Jacobi identity of
Subsection~\ref{subsec:OrderedAdjointDerivativeAndJacobi} is the
loop-space Bianchi identity.  The construction of Sections
\ref{sec:LoopCalculus}--\ref{sec:StokesFunctionalsLeibniz} shows that,
after the chiral or anti-chiral Itoyama--Takashino projection is
applied, a finite Stokes functional whose four-dimensional boundary
area derivative lies in one Lorentzian Hodge sector is annihilated
by the homogeneous local operators.  Applying this result to
\eqref{eq:HodgeDualAreaDerivative} gives
\begin{equation}
    \boxed{
    \widehat{\mathbb L}_\pm S_+[C]
    =
    \widehat{\mathbb L}_\pm S_-[C]
    =0.
    }
    \label{eq:BothBranchesProjectedZeroModes}
\end{equation}
No separately postulated chiral field strength, superinstanton, or
independent superspace two-form is required.  The supersymmetric
dependence comes from the invariant boundary one-form $\Pi^a$, while
the vanishing follows from four-dimensional Hodge duality of the area
derivative and the ordered Jacobi identity.

By the chain rule of Section~\ref{sec:StokesFunctionalsLeibniz}, every
ordinary function of $S_+$ and $S_-$, in particular their exponential,
is again a common homogeneous zero mode.

\subsection{Exact additive branch, moving the cut on the BV circle, and geodesic-wire sewing}
\label{subsec:SHDAdditivityAndSewing}

The SHD sewing construction is internal to the explicit
Super--Weierstrass solution.  We do not invoke a separate existence or
bridge theorem for the Plateau problem.  The regulator uses a geodesic
wire drawn on the selected primary surface, and the daughter surfaces
are its two restrictions.

We first record the reparametrization statement used to place each
daughter boundary on its own standard circle.

\begin{lemma}[Relocation of a BV boundary period]
\label{lem:MainBVPeriodRelocation}
Let \(\Gamma\simeq S^1\) be a fixed oriented closed boundary and let
\(g\) be a BV trace on \(\Gamma\).  Choose any ordered interval
\(I_{t,s}=[t,s]\) as one traversal of \(\Gamma\), and let
\(g(\theta\mid s,t)\) denote the corresponding interval
representation.  For every orientation-preserving,
reparametrization-invariant functional \(\mathcal A[g]\),
\begin{equation}
    \boxed{
    \mathcal A_{[t,s]}[g(\,\cdot\mid s,t)]
    =\mathcal A_{S^1}[g].
    }
    \label{eq:MainBVPeriodFunctionalInvariance}
\end{equation}
The distributional closure condition is likewise independent of the
chosen period,
\begin{equation}
    \boxed{
    \int_t^s dg(\theta\mid s,t)=0.
    }
    \label{eq:MainBVPeriodClosure}
\end{equation}
The one-sided traces at \(t+0\) and \(s-0\) remain independent BV data;
the jump at the marked cut is included in the distributional derivative.
\end{lemma}

For the invariant boundary curve \(\mathcal C=f+\widetilde f\), the daughter
closure condition is
\begin{equation}
    \int_t^s d\theta\,
    \partial_\theta\mathcal C(\theta\mid s,t)=0.
    \label{eq:MainDaughterClosure}
\end{equation}
Thus the numerical endpoints used as a period carry no geometric
information.

Let
\begin{equation}
    Y_C=f(z)+\widetilde f(\widetilde z)
    \label{eq:MainCommonWhiteBridgeMap}
\end{equation}
be the selected parent Super--Weierstrass surface on the disk
\(\mathcal D_C\).  For an ordered pair \(s,t\), choose the geodesic
wire \(\gamma_{st}\) used in the contact kernel and a simple lift
\(\widehat\gamma_{st}\subset\mathcal D_C\).  The lift divides the disk
into two domains,
\begin{equation}
    \mathcal D_C
    =\mathcal D_{st}^{\gamma}
     \cup_{\widehat\gamma_{st}}
     \mathcal D_{ts}^{\gamma}.
    \label{eq:SHDParentDomainPartition}
\end{equation}
The daughter maps are the restrictions of the same analytic solution.
After mapping each daughter domain conformally to the unit disk by
\(h_{st},h_{ts}:\mathbb D\to\mathcal D_{st,ts}^{\gamma}\),
\begin{equation}
\begin{aligned}
    f_{st}&=f\circ h_{st},
    &\qquad
    \widetilde f_{st}&=\widetilde f\circ\widetilde h_{st},
    \\
    f_{ts}&=f\circ h_{ts},
    &
    \widetilde f_{ts}&=\widetilde f\circ\widetilde h_{ts}.
\end{aligned}
    \label{eq:MainDaughterPullbackMaps}
\end{equation}
No Hilbert transform is recomputed and no holomorphic function is
perturbed.  The nonlocal Dirichlet reconstruction would arise only if
one independently prescribed two new daughter boundaries.  Here one
already has the parent analytic map and merely restricts its domain.

Because the parent action is a local integral, the partition
\eqref{eq:SHDParentDomainPartition} gives, branch by branch,
\begin{equation}
    \boxed{
    S_\pm[C]
    =S_\pm[C_{st}^{\gamma}]
     +S_\pm[C_{ts}^{\gamma}].
    }
    \label{eq:SHDSplittingAdditivity}
\end{equation}
For two distinct loops, the selected joined branch is the disjoint
union of their two parent surfaces with the connector traversed in the
two opposite directions.  The wire has zero two-dimensional measure,
so
\begin{equation}
    \boxed{
    S_\pm[C_i\#_{s,t}^{\gamma}C_j]
    =S_\pm[C_i]+S_\pm[C_j].
    }
    \label{eq:SHDMergingAdditivity}
\end{equation}
These are exact identities at fixed \(\Lambda\).

The four BV twistor germs
\begin{equation}
    \mathfrak t_s^\pm
    =(\lambda,\widetilde\lambda,\rho,\widetilde\rho)(s\pm0),
    \qquad
    \mathfrak t_t^\pm
    =(\lambda,\widetilde\lambda,\rho,\widetilde\rho)(t\pm0)
    \label{eq:FourTwistorContactGerms}
\end{equation}
are redistributed according to
\begin{equation}
    \boxed{
    C_{st}^{\gamma}:(\mathfrak t_s^+,\mathfrak t_t^-),
    \qquad
    C_{ts}^{\gamma}:(\mathfrak t_t^+,\mathfrak t_s^-).
    }
    \label{eq:TwistorGermRedistribution}
\end{equation}
There is no equality between the beginning and ending twistors of a
daughter loop.  Each daughter remains in the same free-trace
Super--Weierstrass class and satisfies its own bilinear closure
condition.

The local Stokes argument is equally direct.  On the conformal branch,
set
\begin{equation}
    p^a=f'{}^a(z),
    \qquad
    q^a=\widetilde f'{}^a(\widetilde z),
    \qquad
    F_{ab}=p_aq_b-p_bq_a.
    \label{eq:SHDBridgeSimpleBivector}
\end{equation}
After varying the full parent functional first, its area derivative on
the Weierstrass solution is
\begin{equation}
    \boxed{
    \Omega_{ab}^{(\chi)}
    =-\nu_{\rm HD}
      \frac{(P_\chi F)_{ab}}
      {\sqrt{-\frac12(P_\chi F)_{cd}(P_\chi F)^{cd}}}.
    }
    \label{eq:SHDWeierstrassAreaDerivative}
\end{equation}
The two daughter values on the geodesic wire are restrictions of this
one tensor.  Therefore
\begin{equation}
    \boxed{
    \left.\Omega_{ab}^{(\chi),st}\right|_{\gamma_{ts}}
    =
    \left.\Omega_{ab}^{(\chi),ts}\right|_{\gamma_{st}}
    =
    \left.\Omega_{ab}^{(\chi),C}\right|_{\gamma}.
    }
    \label{eq:SHDBridgeOmegaMatching}
\end{equation}
The boundary orientations are opposite, and the two Stokes currents
cancel pointwise.  The same is true for a common variation of the
parent map,
\begin{equation}
    \boxed{
    \left.\delta\Omega_{ab}^{(\chi),st}\right|_{\gamma_{ts}}
    =
    \left.\delta\Omega_{ab}^{(\chi),ts}\right|_{\gamma_{st}}.
    }
    \label{eq:SHDBridgeLinearizedMatching}
\end{equation}
These equalities are consequences of restriction, not independent
continuity assumptions.

In superspace,
\begin{equation}
    p_{\alpha\dot\alpha}
    =\lambda_\alpha\widetilde\lambda_{\dot\alpha},
    \qquad
    q_{\alpha\dot\alpha}
    =\rho_\alpha\widetilde\rho_{\dot\alpha}.
    \label{eq:MainBridgeTwistorCurrents}
\end{equation}
The individual twistors may differ between daughter charts by
projective transition functions, but the bilinear one-forms are
restrictions of the same parent currents.  Equations
\eqref{eq:SHDBridgeOmegaMatching} and
\eqref{eq:SHDBridgeLinearizedMatching} therefore hold coefficient by
coefficient in the finite Grassmann expansion.

For either Hodge branch, define
\begin{equation}
\begin{aligned}
    \Delta_\chi^{\rm split}
    &=S_\chi[C_{st}^{\gamma}]
      +S_\chi[C_{ts}^{\gamma}]-S_\chi[C],
    \\
    \Delta_\chi^{\rm join}
    &=S_\chi[C_i\#_{s,t}^{\gamma}C_j]
      -S_\chi[C_i]-S_\chi[C_j].
\end{aligned}
    \label{eq:SHDGeodesicDefects}
\end{equation}
Exact restriction additivity and the finite-Stokes property give, for
\(\Delta_\chi=\Delta_\chi^{\rm split}\) or
\(\Delta_\chi^{\rm join}\),
\begin{equation}
    \boxed{
    \Delta_\chi=0.
    }
    \label{eq:SHDDefectValueCondition}
\end{equation}
\begin{equation}
    \boxed{
    \mathfrak D_A(s)\Delta_\chi
    =\mathfrak D_A(t)\Delta_\chi
    =0.
    }
    \label{eq:SHDDefectPointDerivativeConditions}
\end{equation}
\begin{equation}
    \boxed{
    P_\alpha^{(s)}Q_{(t)}^\alpha\Delta_\chi=0.
    }
    \label{eq:SHDDefectProjectedCondition}
\end{equation}
The anti-chiral kernel obeys the conjugate projected identity.  This
verifies Proposition~\ref{prop:MainSurfaceGeodesicWireSewing} for each
Hodge branch.  The detailed geodesic contact profile, BV period
relocation, domain partition, and differentiated-kernel proof are given
in Supplementary Material, Secs.~S7, S8, and S10.5.

\subsection{Lorentzian confining phase and static energy}
\label{subsec:SHDDressingApplication}

Equations~\eqref{eq:BothBranchesProjectedZeroModes},
\eqref{eq:SHDDefectValueCondition}--\eqref{eq:SHDDefectProjectedCondition}
verify the hypotheses of Theorem~\ref{thm:ProjectedFiniteNDressing} on
the selected branch.

For pure $SU(N)$, $4d,\mathcal N=1$ SYM, define
\begin{equation}
    \Lambda_{\mathcal N=1}^{3N}
    =
    \mu^{3N}
    \exp\!\left[
        -\frac{8\pi^2}{g^2(\mu)}
        +i\theta_{\rm YM}
    \right],
    \label{eq:HolomorphicStrongScale}
\end{equation}
\begin{equation}
    \Lambda_k^3
    =
    \left|\Lambda_{\mathcal N=1}\right|^3
    \exp\!\left[
        \frac{i(\theta_{\rm YM}+2\pi k)}{N}
    \right],
    \qquad
    k=0,\ldots,N-1,
    \label{eq:SYMVacuumBranches}
\end{equation}
and
\begin{equation}
    \Lambda_k^2
    =\left(\Lambda_k^3\right)^{2/3},
    \label{eq:DimensionTwoVacuumParameter}
\end{equation}
where the fractional-power branch is part of the vacuum data.  The
real Lorentzian phase is
\begin{equation}
    \Phi_{\rm SHD}^{(k)}[C]
    =
    \frac{c}{2}
    \left(
        \Lambda_k^2S_+[C]
        +
        \bar\Lambda_k^2S_-[C]
    \right)
    =
    c\,\Re\!\left[
        \Lambda_k^2S_+[C]
    \right],
    \label{eq:PhysicalSHDAreaExponent}
\end{equation}
where $c$ is a real dimensionless vacuum coefficient.  The dressed
solution is
\begin{equation}
    W_{\Lambda,k}[C]
    =
    W_{\Lambda,0}[C]
    \exp\!\left[
        i\Phi_{\rm SHD}^{(k)}[C]
    \right].
    \label{eq:SHDDressedWilsonLoop}
\end{equation}

The exact Fourier--Hilbert result of
Section~\ref{subsec:SHDPlanarSuperarea} fixes the relation between the
coefficient of the SHD functional and the physical string tension.  On
the Euclidean continuation, define the positive coefficient
$\kappa_k$ of the common planar stationary value by
\begin{equation}
    \Phi_{E,\rm SHD}^{(k)}[C_{\rm pl}]
    \equiv
    \frac{c}{2}
    \left(
        \Lambda_k^2S_+[C_{\rm pl}]
        +\bar\Lambda_k^2S_-[C_{\rm pl}]
    \right)
    =\kappa_k S_{\rm pl}[C_{\rm pl}],
    \qquad
    \kappa_k>0,
    \label{eq:EffectivePlanarKappaVacuum}
\end{equation}
where
\begin{equation}
    S_{\rm pl}[C]
    =2\sqrt2\,D[\mathcal C].
    \label{eq:CommonPlanarSHDValueStatic}
\end{equation}
On a Euclidean real branch with $S_+=S_-=S_{\rm pl}$, one has
$\kappa_k=c\,\Re(\Lambda_k^2)$ after selecting the physical
fractional-power branch.

For a rectangular bosonic body of Euclidean temporal length $T_E$ and
spatial width $L$, Eq.~\eqref{eq:ExactSuperRectangleArea} gives
\begin{equation}
    D[\mathcal C]
    =LT_E+\Delta^{(2)}_\eta+\Delta^{(4)}_\eta.
    \label{eq:EuclideanRectangleExactSHDArea}
\end{equation}
For the ordinary static-source loop the boundary fermions are set to
zero.  Consequently,
\begin{equation}
    W_{E,k}[C_{T_E,L}]
    \sim
    W_{E,0}[C_{T_E,L}]
    \exp[-\sigma_kLT_E],
    \qquad
    \boxed{
    \sigma_k=2\sqrt2\,\kappa_k.
    }
    \label{eq:ExactSYMPlanarStringTension}
\end{equation}
Analytic continuation $T_E\to iT$ gives
\begin{equation}
    W_{\Lambda,k}[C_{T,L}]
    \sim
    W_{\Lambda,0}[C_{T,L}]
    \exp[-i\sigma_kLT],
    \qquad
    E_k(L)=\sigma_kL>0.
    \label{eq:SHDPositiveStaticEnergy}
\end{equation}
For nonconstant boundary fermions the Wilson factor is multiplied by
the finite nilpotent factor in
Eq.~\eqref{eq:ExactNilpotentRectangleFactor}; the body of the static
energy remains $\sigma_kL$.

The loop hierarchy fixes the admissible Hodge-zero-mode structure and
the additive contact continuation.  It does not determine the vacuum
coefficient, the fractional-power branch, the supersymmetric vacuum
label, or the numerical value of the string tension.

\section{Loop-space fixed points of the SYM gradient flow}
\label{sec:SpontaneousStochasticization}

The exact results of Sections~\ref{sec:LoopCalculus}--
\ref{sec:SuperHodgeDualSurface} are statements about the Lorentzian
Itoyama--Takashino hierarchy and do not depend on a gradient-flow
interpretation.  The present section gives an additional dynamical
interpretation.  Since a downward Yang--Mills gradient flow and
Parisi--Wu stochastic quantization are naturally formulated for the
Euclidean action, one analytic continuation is required in this
section.  No second spinor convention is introduced: dotted and
undotted spinors remain independent complex variables during the
continuation, and all POOC identities are continued algebraically.

For a physical timelike contour, write the Lorentzian SHD factor as
\begin{equation}
    \mathcal Z_{M,k}[C]
    =
    \exp\!\left[
        i\Phi_{\rm SHD}^{(k)}[C]
    \right].
    \label{eq:LorentzianSHDFactorFlowSection}
\end{equation}
Under
\begin{equation}
    x^0=-ix_E^4,
    \qquad
    S_E=-iS_M,
    \label{eq:GradientFlowWickRotation}
\end{equation}
define the analytically continued SHD functional by
\begin{equation}
    \Phi_{E,\rm SHD}^{(k)}[C_E]
    =
    -i\,
    \Phi_{\rm SHD}^{(k)}[C_M]
    \big|_{x^0=-ix_E^4}.
    \label{eq:EuclideanSHDFunctionalFlowSection}
\end{equation}
Then
\begin{equation}
    \mathcal Z_{M,k}[C_M]
    \longrightarrow
    \mathcal Z_{E,k}[C_E]
    =
    \exp\!\left[
        -\Phi_{E,\rm SHD}^{(k)}[C_E]
    \right].
    \label{eq:EuclideanSHDFactorFlowSection}
\end{equation}
For the rectangular contour of
\eqref{eq:SHDPositiveStaticEnergy},
\begin{equation}
    \Phi_{E,\rm SHD}^{(k)}[C_{T_E,L}]
    =
    \sigma_kLT_E,
    \qquad
    \sigma_k>0,
    \label{eq:EuclideanRectangleDamping}
\end{equation}
so the same branch which gives the positive Lorentzian static energy
produces a decaying Euclidean weight.

\subsection{Euclidean supersymmetric gradient flow}
\label{subsec:EuclideanSYMGradientFlow}

The ordinary Yang--Mills gradient flow, also called the Wilson flow,
is a standard tool in continuum and lattice gauge theory
\cite{Luscher2010WilsonFlow,Luscher2013Chiral,
DelDebbioPatellaRago2013}.  After Euclidean continuation, its
stationary points include self-dual and anti-self-dual instantons.
The mathematical theory of the bosonic Yang--Mills heat flow includes
existence, convergence, and bubbling results
\cite{Rade1992YangMillsHeat,Struwe1994YangMillsFlow,
Donaldson1985ASD}.

The supersymmetric flow is more delicate because the superspace
connection is constrained and the flow must preserve supergauge
covariance.  Generalized gradient flows for
\(4d,\mathcal N=1\) SYM and their Wess--Zumino-gauge component forms
have been constructed in
Refs.~\cite{KikuchiOnogi2014,KadohUkita2018}.  We need only this form
here.

Let
\begin{equation}
    \phi^I
    =
    \left(
        v_a,\lambda_\alpha,
        \bar\lambda_{\dot\alpha},D
    \right)
    \label{eq:GradientFlowWZFields}
\end{equation}
denote the analytically continued Wess--Zumino component fields.  A
generalized downward flow may be written schematically as
\begin{equation}
    \partial_\tau\phi^I
    =
    -\mathcal G^{IJ}(\phi)
    \frac{\delta S_E}{\delta\phi^J}
    +\mathcal R_\omega^I(\phi),
    \label{eq:SYMGradientFlow}
\end{equation}
where \(\tau\) is the fictitious flow time,
\(\mathcal G^{IJ}\) is the graded field-space metric, and
\(\mathcal R_\omega^I\) is a gauge-damping vector field.  The precise
component form depends on the chosen supersymmetric flow convention.
The damping term changes the representative on a supergauge orbit but
has no effect on a closed Wilson loop.

Adding the corresponding graded Gaussian noise gives
\begin{equation}
    \partial_\tau\phi^I
    =
    -\mathcal G^{IJ}
    \frac{\delta S_E}{\delta\phi^J}
    +\mathcal R_\omega^I
    +\sqrt{\varepsilon}\,\xi^I,
    \label{eq:StochasticSYMGradientFlow}
\end{equation}
where \(\varepsilon\) measures the external noise strength.  With the
appropriate fluctuation--dissipation normalization, the stationary
measure is the Euclidean gauge-theory measure of Parisi--Wu stochastic
quantization
\cite{ParisiWu1981,DamgaardHueffel1987}.  The fermionic components of
\eqref{eq:StochasticSYMGradientFlow} are understood in the standard
graded stochastic sense; no positivity statement is assigned to the
fermionic block of \(\mathcal G^{IJ}\).

\subsection{Exact induced evolution of closed super Wilson loops}
\label{subsec:InducedLoopSpaceGradientFlow}

Let
\begin{equation}
    W_{S,\tau}[C]
    =
    \frac1N
    \Tr P\exp
    \oint_C e^A A_A(z,\tau)
    \label{eq:FlowedWilsonLoop}
\end{equation}
be the super Wilson loop evaluated on the flowed connection.  The
exact flow-time derivative is the marked insertion
\begin{equation}
\begin{aligned}
    \partial_\tau W_{S,\tau}[C]
    ={}&
    \frac1N
    \oint_C ds\,
    \Tr\!\left[
        U_\tau(s+2\pi,s)
        E^A(s)
        \partial_\tau A_A\bigl(z(s),\tau\bigr)
    \right].
\end{aligned}
    \label{eq:FlowDerivativeWilsonLoop}
\end{equation}
A gauge-damping contribution of the form
\(\mathcal D_A\omega\) becomes an endpoint commutator in
\eqref{eq:FlowDerivativeWilsonLoop}; its closed color trace vanishes.
Thus the deterministic connection-space flow induces a homogeneous,
gauge-invariant first-order evolution operator on loop functionals:
\begin{equation}
    \partial_\tau W_{S,\tau}[C]
    =
    -\mathcal H_{\rm loop}^{(0)}
    W_{S,\tau}[C].
    \label{eq:LoopSpaceGradientFlow}
\end{equation}

The component Euler derivatives
\(\delta S_E/\delta\phi^I\) are precisely the finite components of the
covariant chiral and anti-chiral equation-of-motion superfields
\begin{equation}
    \mathcal E_+
    \propto
    \mathcal D^\alpha W_\alpha,
    \qquad
    \mathcal E_-
    \propto
    \bar{\mathcal D}_{\dot\alpha}\bar W^{\dot\alpha}.
    \label{eq:SYMGradientEOMSuperfields}
\end{equation}
This component-to-superfield assembly is the same one used in the
Itoyama--Takashino Schwinger--Dyson derivation.  Acting with the graded
field-space metric, applying the exact chain rule
\(\phi^I\mapsto A_A[\phi]\), and inserting the result into
\eqref{eq:FlowDerivativeWilsonLoop} therefore gives the exact
loop-space decomposition
\begin{equation}
    \boxed{
    \mathcal H_{\rm loop}^{(0)}
    =
    \oint_C ds\,
    \left[
        \mathcal T_+(s)\circ
        \widehat{\mathbb L}_{+,E}(s)
        +
        \mathcal T_-(s)\circ
        \widehat{\mathbb L}_{-,E}(s)
    \right].
    }
    \label{eq:ProjectedLoopFlowGenerator}
\end{equation}
Here \(\mathcal T_\pm(s)\) contain the tangent contraction, the local
chain rule from the Wess--Zumino fields to the constrained connection,
and the chosen graded field-space metric.  Equation
\eqref{eq:ProjectedLoopFlowGenerator} is the loop-space image of the
gradient-flow equation itself.  It is not an additional factorization
hypothesis and is unrelated to large-\(N\) factorization.

At nonzero noise, It\^o differentiation of products of Wilson lines
produces a second-variation term.  In loop variables this term has the
same splitting-and-joining geometry as the differentiated contact
kernels of the I--T hierarchy.  At stochastic equilibrium the
Schwinger--Dyson identities of the Euclidean measure are therefore the
analytic continuation of the finite-cutoff hierarchy of
Section~\ref{sec:SchwingerDysonContactTerm}.  The exact matching of a
particular supersymmetric noise covariance to the normalization of
\(K_{\pm,\Lambda}\) depends on the chosen gradient-flow convention and
is not required for the fixed-point identity below.

\subsection{Exact stationarity of the SHD factor}
\label{subsec:SHDGradientFlowStationarity}

The local operator identity proved in
Section~\ref{subsec:SHDProjectedZeroModes} is
\begin{equation}
    \widehat{\mathbb L}_\pm(s) S_+[C]
    =
    \widehat{\mathbb L}_\pm(s) S_-[C]
    =0
    \qquad
    \text{for every marked point }s.
    \label{eq:SHDZeroModeGradientFlowSection}
\end{equation}
Consequently, the real Lorentzian phase
\eqref{eq:PhysicalSHDAreaExponent} satisfies
\begin{equation}
    \widehat{\mathbb L}_\pm(s)
    \Phi_{\rm SHD}^{(k)}[C]
    =0,
    \label{eq:SHDAreaExponentZeroMode}
\end{equation}
and the finite-Stokes chain rule gives
\begin{equation}
    \widehat{\mathbb L}_\pm(s)
    \exp\!\left[
        i\Phi_{\rm SHD}^{(k)}[C]
    \right]
    =0.
    \label{eq:SHDExponentialZeroModeGradient}
\end{equation}
The same identity holds on the regular analytically continued branch:
\begin{equation}
    \widehat{\mathbb L}_{\pm,E}(s)
    \mathcal Z_{E,k}[C_E]
    =0.
    \label{eq:EuclideanSHDZeroModeGradient}
\end{equation}
Substitution into the exact decomposition
\eqref{eq:ProjectedLoopFlowGenerator} yields
\begin{equation}
    \boxed{
    \mathcal H_{\rm loop}^{(0)}
    \mathcal Z_{E,k}[C_E]
    =0,
    \qquad
    \partial_\tau\mathcal Z_{E,k}[C_E]=0.
    }
    \label{eq:SHDStationaryGradientFlow}
\end{equation}
Thus the analytically continued SHD factor is an exact fixed point of
the zero-noise SYM loop flow.  No additional
factorization assumption is required.  Its Lorentzian boundary value
is the confining phase \(\exp[i\Phi_{\rm SHD}^{(k)}]\), whose
rectangular-loop behavior gives the positive static energy
\eqref{eq:SHDPositiveStaticEnergy}.

Exact stationarity does not by itself imply dynamical selection.  The
remaining open question is the stability of the SHD fixed point.  Let
\(\delta W_\tau\) be a physical gauge-invariant perturbation transverse
to the family of SHD vacuum zero modes.  Denote by
\(\mathcal M_k\) the corresponding linearized loop-flow generator,
\begin{equation}
    \partial_\tau\delta W_\tau
    =
    -\mathcal M_k\,\delta W_\tau.
    \label{eq:LinearizedSHDLoopFlow}
\end{equation}
The dynamical stability assumption is that, after quotienting
supergauge directions and the neutral directions tangent to the
vacuum family,
\begin{equation}
    \Re\operatorname{spec}'\mathcal M_k>0,
    \qquad
    \lim_{\tau\to\infty}
    \lVert\delta W_\tau\rVert=0.
    \label{eq:SHDStabilityAssumption}
\end{equation}
Equivalently, the classical flow is attractive in the
physical directions transverse to the SHD fixed-point manifold.  This
stability statement is not proved in the present paper.

\subsection{Spontaneous stochasticization}
\label{subsec:SpontaneousStochasticizationMechanism}

At finite noise, stochastic quantization evolves a probability
functional on field space toward the Euclidean quantum measure.  A
naive zero-noise limit might be expected to collapse onto individual
classical trajectories of the classical gradient flow.  A
nontrivial statistical state can nevertheless survive if the
infinite-time, infinite-volume, or large-\(N\) limit is taken before
the external noise is removed.

We call this possibility \emph{spontaneous stochasticization}, or
spontaneous quantization.  The proposition that the external noise
driving the stochastic quantization of quantum field theories
(encompassing both gauge and gravity theories) might emerge
spontaneously from the chaotic nonlinear dynamics of internal degrees
of freedom was originally proposed by the author in 1986
\cite{Migdal_1986_Stochastic}.  Concrete constructions now exist in
classical fluid dynamics and in the Yang--Mills gradient flow.  The
present theory supplies an exact supersymmetric loop-space fixed
point; its interpretation as a dynamically selected stochastic state
requires the stability condition
\eqref{eq:SHDStabilityAssumption}.

The physical analogy is with spontaneous magnetization: an
infinitesimal external source selects a phase, and the order parameter
remains after the source is removed in the thermodynamic limit.  A
closer dynamical analogy is decaying turbulence, where thermal or
external microscopic noise can select a nontrivial statistical
solution whose macroscopic fluctuations remain finite even as the
selecting noise tends to zero.  The corresponding Wilson-loop
construction for Navier--Stokes flow was studied in
Refs.~\cite{migdal2023exact,migdal2024quantum,ReviewPaperAM}.

The analogous proposal for gauge theory was formulated in
Refs.~\cite{migdal2025SQYMflow,migdal2025geometric}.  In the present
setting, two statements must be separated.  First, the Euclidean SHD
factor is an exact fixed point of the zero-noise SYM loop flow by
\eqref{eq:SHDStationaryGradientFlow}.  Second, if this fixed point is
stable and selected by an arbitrarily weak external noise in the
appropriate infinite-time and large-volume limits, then its
Lorentzian continuation
\begin{equation}
    W_{\Lambda,k}[C]
    =
    W_{\Lambda,0}[C]
    \exp\!\left[
        i\Phi_{\rm SHD}^{(k)}[C]
    \right]
    \label{eq:SpontaneouslyQuantizedWilsonLoop}
\end{equation}
survives the zero-noise limit with a nonzero confining string tension.
Under this stability assumption, the SHD fixed point realizes
spontaneous stochasticization.

The present paper does not prove the stability condition
\eqref{eq:SHDStabilityAssumption}.  Such a proof requires the spectrum
of the linearized loop-flow generator, control of its supergauge and
vacuum zero modes, and a precise specification of the order of the
zero-noise, infinite-time, and large-volume limits.  In the case of
fluid turbulence, a related stability analysis was recently performed
in Ref.~\cite{migdal2026EulerStability} under self-consistent
compactness assumptions for the momentum-loop dynamics.

\subsection{Connection-space fixed points and loop-space zero modes}
\label{subsec:ConnectionAndLoopFixedPoints}

The loop-space fixed point must not be confused with an individual
fixed connection.  After Euclidean continuation, a connection-space
fixed point satisfies
\begin{equation}
    \frac{\delta S_E}{\delta\phi^I}=0
    \label{eq:ConnectionSpaceFixedPoint}
\end{equation}
modulo gauge transformations.  Self-dual and anti-self-dual
instantons, and their supersymmetric extensions, are examples
\cite{ADHM,Osborn1979,DoreyKhozeMattis2002,
WessBagger,GatesGrisaruRocekSiegel}.

For every point \(m\) of an on-shell superinstanton moduli space and
every topological sector \(Q\), the corresponding classical Wilson
functional obeys the homogeneous local equations
\begin{equation}
    \widehat{\mathbb L}_{\pm,E}
    W_{Q,m}[C]
    =0.
    \label{eq:ModuliWilsonFunctionalZeroMode}
\end{equation}
The relevant linearity is the linearity of the homogeneous loop
operators acting on Wilson functionals, not a linearity of gauge
potentials.  A formal moduli-space average
\begin{equation}
    W[C]
    =
    \sum_Q
    \int_{\mathcal M_Q}
    d\mu_Q(m)\,
    \rho_Q(m)\,
    W_{Q,m}[C]
    \label{eq:ModuliSpaceAverage}
\end{equation}
therefore remains a homogeneous loop zero mode, provided the measure
has no boundary anomaly under the local loop operations.  The weights
\(\rho_Q\) are vacuum data and are not fixed by the homogeneous
equations.

The SHD construction does not reconstruct such a connection-space
measure.  It produces the gauge-invariant zero mode directly from the
invariant boundary curve, the stationary Hodge-dual surface, and
its Hodge-resolved Stokes area derivative.  The statistical
superinstanton picture is therefore an interpretation of the
loop-space solution, not an additional step in its proof.

\subsection{Scope of the gradient-flow interpretation}
\label{subsec:GradientFlowScope}

The logical content of the paper may be separated as follows:
\begin{equation}
\begin{gathered}
    \text{I--T + POOC + SHD}
    \Longrightarrow
    \text{exact finite-}\!N\text{ zero mode},
    \\
    \text{Euclidean continuation}
    \Longrightarrow
    \text{exact SHD loop-flow fixed point},
    \\
    \text{stability + zero-noise order of limits}
    \Longrightarrow
    \text{spontaneous stochasticization}.
\end{gathered}
    \label{eq:GradientFlowLogicalSeparation}
\end{equation}
The first two implications are established in this paper.  The third
is the remaining dynamical hypothesis.

Accordingly, no separate loop-flow factorization assumption is
required for stationarity.  What is not established here is the
asymptotic stability or attractiveness of the SHD fixed point, a
unique invariant probability measure on the full space of constrained
connections, or the limiting procedure by which an infinitesimal
external noise selects one member of the SHD vacuum family.  The exact
result is the existence of the gauge-invariant confining fixed point;
its dynamical selection as the zero-noise equilibrium is conditional
only on the stability and limit assumptions stated above.

\section{Discussion: the Lorentzian SHD surface as a loop-space master field}
\label{sec:SHDAsWittenMasterField}
\label{subsec:SHDAsWittenMasterField}

The exact Itoyama--Takashino hierarchy is the loop-space form of the
Schwinger--Dyson equations of pure \(4d,\mathcal N=1\) SYM.  Like an
ordinary equation of motion, it constrains the admissible states but
does not by itself select a vacuum.  Section~\ref{sec:StokesFunctionalsLeibniz}
shows how this vacuum freedom is realized geometrically: multiplication by
the exponential of a finite-Stokes zero mode that is additive under the
geodesic-wire split and join maps every solution of the complete finite-\(N\),
ultraviolet-defined hierarchy to another solution.

For the Lorentzian SHD pair constructed in
Section~\ref{sec:SuperHodgeDualSurface}, the corresponding vacuum
factor is
\begin{equation}
    \boxed{
    \mathcal Z_{{\rm SHD},k}[C]
    =
    \exp\!\left[
        i\Phi_{\rm SHD}^{(k)}[C]
    \right],
    }
    \label{eq:SHDMasterAreaLawFactor}
\end{equation}
where \(\Phi_{\rm SHD}^{(k)}\) is the real Lorentzian phase
\eqref{eq:PhysicalSHDAreaExponent}.  The overall coefficient, the
fractional-power branch in \(\Lambda_k^2\), and the supersymmetric
vacuum label \(k\) are not fixed by the homogeneous loop equations.
They are nonperturbative vacuum data.  The physical branch is selected
by the static-energy condition
\eqref{eq:SHDPositiveStaticEnergy},
\begin{equation}
    W[C_{T,L}]
    \sim
    \exp[-i\sigma_kLT],
    \qquad
    E_k(L)=\sigma_kL,
    \qquad
    \sigma_k>0.
    \label{eq:MasterSurfaceStaticEnergy}
\end{equation}
Thus the two complex Lorentzian Hodge sectors combine into a real
phase whose long-rectangle limit yields a real positive Hamiltonian
energy.

\subsection{The loop-to-surface map}
\label{subsec:MasterSurfaceMapDiscussion}

There is a useful large-\(N\) interpretation of this geometric vacuum
factor.  It is a gauge-invariant loop-space version of Witten's master field
\cite{Witten:1980ez}.  Large-\(N\) factorization makes the planar theory
classical at the level of gauge-invariant observables, but this does not mean
that the master variable must be an ordinary colored gauge potential.  Here it
is a fixed assignment of geometric data to every admissible superloop.

There is no single surface independent of the observable contour.  The
construction is the map
\begin{equation}
    \boxed{
    C
    \longmapsto
    \mathcal C_C
    \longmapsto
    Y_C
    \longmapsto
    \bigl(S_+[C],S_-[C]\bigr)
    \longmapsto
    \Phi_{\rm SHD}^{(k)}[C].
    }
    \label{eq:SHDMasterSurfaceMap}
\end{equation}
The first arrow reconstructs the invariant boundary curve from
Subsection~\ref{subsec:InvariantBoundaryCurveAndBosonicHDInput}:
\begin{equation}
    d\mathcal C^a
    =
    e^a\big|_C,
    \qquad
    \oint_C e^a=0.
    \label{eq:MasterSurfaceInvariantBoundary}
\end{equation}
The second arrow selects a stationary embedding with the Dirichlet
condition
\begin{equation}
    Y_C^a\big|_{\partial\mathcal D}
    =
    \mathcal C_C^a,
    \label{eq:MasterSurfaceDirichletBoundary}
\end{equation}
and the two Hodge-locked auxiliary representations are
\begin{equation}
    X_{C,a}^{\chi,i}
    =
    \mathfrak e_{ab}^{\chi,i}Y_C^b,
    \qquad
    X_{C,a}^{\chi,i}\big|_{\partial\mathcal D}
    =
    \mathfrak e_{ab}^{\chi,i}\mathcal C_C^b,
    \qquad
    \chi=\pm1.
    \label{eq:MasterSurfaceProjectedBoundary}
\end{equation}
These are precisely the boundary relations
\eqref{eq:HDBoundaryConditionY} and
\eqref{eq:HDBoundaryLocking}.

The index \(i=1,2,3\) labels a basis of the three-dimensional complex Lorentzian Hodge
sector.  A complex orthogonal change of
this basis changes the auxiliary coordinates \(X_{C,a}^{\chi,i}\), but
leaves the projector \(P_\chi\), the stationary functional, and its
Stokes variation invariant.  This is only the freedom to choose a basis in the Hodge representation, not
an additional physical gauge field.  On a
physical Lorentzian real slice, the two Hodge frames and the two Stokes
functionals are chosen as conjugate pairs according to
\eqref{eq:SHDRealityCondition}.

\subsection{Gauge dynamics encoded by path-ordered operators}
\label{subsec:MasterSurfacePOOCDiscussion}

The master-surface construction parallels the loop-space
representation of the gauge theory.  The covariant derivative
\begin{equation}
    \nabla_A
    =
    D_A-A_A
    \label{eq:CovariantDerivativeMasterFieldDiscussion}
\end{equation}
is not itself a gauge-invariant observable, but Section~\ref{sec:LoopCalculus}
represents its action directly on a Wilson transporter.  The POOC
dictionary is
\begin{equation}
\begin{aligned}
    \delta_A
    &\longleftrightarrow
    \nabla_A,
    \\
    \delta^\Sigma_{AB}
    =
    T_{AB}{}^C\delta_C
    -[\delta_A,\delta_B\}_o
    &\longleftrightarrow
    F_{AB},
    \\
    \mathfrak D_A^{\,o}\delta^\Sigma_{BC}
    &\longleftrightarrow
    \mathcal D_AF_{BC}.
\end{aligned}
    \label{eq:MasterSurfacePOOCDictionary}
\end{equation}
Path ordering is the direct operator definition of coincident
loop differentiation; it is not an ultraviolet smoothing.  The only
physical ultraviolet scale enters the contact bi-kernel and short
connectors of
Subsection~\ref{subsec:GaugeCovariantRegularization}.

The exact closed hierarchy of Section~\ref{sec:SchwingerDysonContactTerm}
then supplies the dynamical equations for Wilson-loop functionals.
The path-ordered graded Jacobi identity gives the torsionful loop-space
Bianchi identity. The parent surface supplies finite Stokes area
derivatives in one Lorentzian Hodge sector, without an independently
postulated superspace two-form completion. After the chiral or anti-chiral
Itoyama--Takashino projection, four-dimensional Hodge duality converts
the Jacobi--Bianchi identity directly into the local equations.
Consequently,
\begin{equation}
    \boxed{
    \widehat{\mathbb L}_\pm S_+[C]
    =
    \widehat{\mathbb L}_\pm S_-[C]
    =0,
    }
    \label{eq:MasterSurfaceProjectedZeroModes}
\end{equation}
as proved in
Subsection~\ref{subsec:SHDProjectedZeroModes}.

This relation is stronger than a formal equation for the reduced
scalar area. The complete supersymmetric dependence is carried by the
invariant superspace one-form $\Pi^a|_C$ and hence by the boundary
curve. Subsection~\ref{subsec:SHDStokesIntegrability} defines each
$S_\chi[C]$ as the stationary value of the parent functional and derives
its Hodge-resolved Stokes current from the on-shell boundary variation. The
complete supersymmetric surface data are therefore
\begin{equation}
    \boxed{
    \mathfrak M_C
    =
    \left(
        Y_C,
        \Omega^{(+)}[C],
        \Omega^{(-)}[C]
    \right),
    }
    \label{eq:SupersymmetricMasterSurfacePair}
\end{equation}
rather than the bosonic embedding alone.

\subsection{Planar master geometry and the finite-
\texorpdfstring{\(N\)}{N} theorem}
\label{subsec:MasterSurfacePlanarFiniteN}

At large \(N\), normalized traces factorize and the exact multiloop
hierarchy \eqref{eq:FiniteNMultiloopSMM} closes on the one-loop
functional through \eqref{eq:PlanarFullSuperspaceSMM}.  The SHD factor
is then a nonfluctuating geometric zero mode of the planar equation.
In this limited but precise sense, the loop-to-surface assignment
\eqref{eq:SHDMasterSurfaceMap} is a master geometry for the confining
vacuum.

The dressing theorem is, however, stronger than the master-field
interpretation.  Theorem~\ref{thm:ProjectedFiniteNDressing} holds for
the complete finite-\(N\) hierarchy at fixed ultraviolet scale.  It
includes same-loop splitting, different-loop joining, and the
\(SU(N)\) traceless-generator terms.  The free-BV twistor factorization and the geodesic-wire theorem
establish the corresponding exact defect identities
in Subsection~\ref{subsec:SHDAdditivityAndSewing}.  Therefore
\begin{equation}
    \boxed{
    W_{\Lambda,k}[C]
    =
    W_{\Lambda,0}[C]
    \exp\!\left[
        i\Phi_{\rm SHD}^{(k)}[C]
    \right]
    }
    \label{eq:MasterSurfaceDressedWilsonLoop}
\end{equation}
solves the same ultraviolet-defined hierarchy as
\(W_{\Lambda,0}[C]\), exactly as stated in
\eqref{eq:SHDDressedWilsonLoop}.

The hierarchy does not determine the vacuum coefficient, the
fractional-power branch, or the supersymmetric vacuum label.  The
stationary branch used here is the explicit parent
Virasoro--Weierstrass solution, while its additive touching-loop continuation is fixed by restriction of the same
Super--Weierstrass surface along the geodesic UV wire.  POOC, the
finite-Stokes Leibniz identities, and the sewing theorem then
show that the resulting phase is compatible with every local loop
operation in the hierarchy.

\subsection{What the master surface does and does not encode}
\label{subsec:MasterSurfaceScope}

The stationary SHD surface is not a replacement for the ordinary
Wilson holonomy and is not a sum over arbitrary Polyakov worldsheets.
For each boundary it selects one stationary Hodge-dual branch, with the
exact additive touching-loop continuation defined by the geodesic cut on the parent surface.  The
undressed factor \(W_{\Lambda,0}[C]\) contains the remaining gauge
dynamics, while the SHD phase supplies the nonperturbative confining
vacuum factor.  The master surface therefore represents the geometric
vacuum sector carried by the exact zero mode; it does not by itself
encode the complete spectrum of planar excitations.

The stationary surface is rigid, but the quantum theory still has rich
dynamics.  The surface fixes the confining background geometry and the positive
static string tension.  Glueball and gluinoball excitations require
additional degrees of freedom propagating on, or coupled to, this
background.  In the nonsupersymmetric Geometric QCD construction
\cite{Migdal2026GeometricQCDII,Migdal2026GeometricQCDIII}, analogous
spectral information is carried by internal fermionic variables.  The
supersymmetric extension of those variables, their supertwistor
representation, and the associated path integral belong to subsequent
parts of the construction.

Section~\ref{sec:SpontaneousStochasticization} gives a complementary
dynamical interpretation.  The analytically continued SHD factor is an exact
fixed point of the homogeneous zero-noise loop flow because the induced
generator \eqref{eq:ProjectedLoopFlowGenerator} is the exact local
equation-of-motion combination of the chiral and anti-chiral I--T/POOC
operators, each of which annihilates the SHD factor at every marked point.
No separate loop-flow factorization assumption is involved.  This
fixed-point result is not needed for the master-surface theorem: the exact
zero-mode and finite-\(N\) dressing statements follow already from
Sections~\ref{sec:LoopCalculus}--\ref{sec:SuperHodgeDualSurface}.  What
remains conditional is the asymptotic stability and zero-noise selection of
the fixed point, as summarized in Subsection~\ref{subsec:GradientFlowScope}.

The relation to Witten's proposal may therefore be summarized as
\begin{equation}
    \boxed{
    \text{Witten master field}
    \quad\longrightarrow\quad
    \text{Lorentzian SHD master-surface functional in loop space}.
    }
    \label{eq:WittenMasterFieldToSHDMasterSurface}
\end{equation}
This is an analogy at the level of large-\(N\) gauge-invariant
observables, not an identification with a single colored spacetime
configuration.  The construction is the loop-to-surface map
\eqref{eq:SHDMasterSurfaceMap}.  Its auxiliary coordinates are not
observables, but its invariant boundary curve, Hodge-resolved Stokes
area derivative, real Lorentzian phase, and positive rectangular-loop
energy are gauge-invariant data.  Their compatibility with both planar
factorization and the complete finite-\(N\), gauge-covariantly
ultraviolet-defined Itoyama--Takashino hierarchy is the precise content
of the master-surface interpretation.

\section{Conclusions}
\label{sec:Conclusions}

We have done three things in this paper.  We have written the exact
Itoyama--Takashino hierarchy in a finite path-ordered operator calculus,
proved the finite-\(N\) multiplicative dressing theorem, and constructed
the Lorentzian supersymmetric Hodge-dual Stokes functionals which provide
the confining zero modes of pure \(4d,\mathcal N=1\) super Yang--Mills
theory.

Section~\ref{sec:LoopCalculus} replaced formal equal-point differentiation
by finite ordered loop operations in the native Lorentzian conventions of
Itoyama and Takashino. The
invariant basis pulls back to the contour as
\begin{equation}
    e^A\big|_C
    =ds\,E^A(s),
    \qquad
    E^A(s)D_A
    =\frac{d}{ds}.
\end{equation}
Because the lifted transporter contains \(-E^A\nabla_A\), the tangent
operator is
\begin{equation}
    \delta_A(s)
    =-
    \frac{\delta}{\delta E^A(s)},
\end{equation}
so that it inserts the gauge-covariant superspace derivative
\(\nabla_A=D_A-A_A\).  We replace singular equal-point differentiation by finite path-ordered
products in distinct contour slots.  Their ordered commutator gives the
torsion-subtracted area derivative,
\begin{equation}
    \boxed{
    \delta^\Sigma_{AB}
    =
    T_{AB}{}^C\delta_C
    -[\delta_A,\delta_B\}_o,
    }
\end{equation}
which inserts the gauge supercurvature \(F_{AB}\) exactly.  A subsequent
ordered adjoint operation inserts its covariant derivative \(\mathcal D_CF_{AB}\),
and the three-point graded Jacobi identity gives the torsionful superspace
Bianchi identity.  The ordered separation is part of the operator definition;
it is not an ultraviolet regulator and leaves no auxiliary scale.

Section~\ref{sec:SchwingerDysonContactTerm} rewrote the exact
Itoyama--Takashino hierarchy in this finite operator language. Their
Wess--Zumino-gauge derivation, summarized geometrically in the main text and
reconstructed in detail in Supplementary Material, Sec.~S4,
uses an even field-space vector field whose contraction with the action
gradient yields the complete covariant equation-of-motion superfield. Its
triangular field dependence ensures its local graded functional divergence vanishes,
while its action on the Wilson line generates a nontrivial chain-rule contact
projector. The explicit component fields in that projector are rewritten as endpoint
and mixed-area operations. Gauge-unfixing then replaces
the Wess--Zumino contact expression with the manifestly supersymmetric kernel; their difference is a total derivative in the integrated contour point and
vanishes by the restricted-supergauge Ward identity.  The marked point remains
arbitrary throughout.

The local chiral and anti-chiral operators are therefore the POOC expressions
\begin{equation}
    \widehat{\mathbb L}_+
    =
    \frac18
    \epsilon_{\alpha\beta}
    \bar\sigma^{a\dot\alpha\beta}
    \mathfrak D^{o\,\alpha}
    \delta^\Sigma_{a\dot\alpha},
    \qquad
    \widehat{\mathbb L}_-
    =
    \frac18
    \epsilon_{\dot\alpha\dot\beta}
    \sigma^{a\alpha\dot\beta}
    \mathfrak D^{o\,\dot\alpha}
    \delta^\Sigma_{a\alpha}.
\end{equation}
Equivalently, because the mixed vector--spinor torsion vanishes, each
acts as a triple path-ordered graded commutator with the sign fixed by the
Itoyama--Takashino conventions. These index-contracted operators are not one-component scalar equations: their action produces Grassmann-even
superfield equations at the marked point, whose finite expansions encode the
component Schwinger--Dyson equations of the entire gauge supermultiplet.

The right-hand side is governed by the differentiated superspace contact
kernels \(K_+\) and \(K_-\), alongside exact color splitting and
joining operations. For \(SU(N)\), the complete finite-\(N\)
hierarchy contains same-loop splitting, different-loop joining, and
the traceless-generator subtraction terms displayed in
\eqref{eq:FiniteNMultiloopSMM}. Planar factorization reduces this
hierarchy to the one-loop equation \eqref{eq:PlanarFullSuperspaceSMM}, but the
dressing theorem imposes no planar restriction.

The path-ordered definition must be distinguished from the genuine
short-distance contact in the quantum hierarchy. The latter is physically
defined at a fixed ultraviolet scale by a smooth superspace bi-kernel and short
gauge transporters joining the two separated contour points, retaining the
complete near-diagonal region. The hierarchy is solved at this
fixed ultraviolet scale; only after its exact solution is inserted into the
worldline, worldsheet, proper-time, or other path integrals defining physical
observables is renormalization imposed. No independent perimeter or cusp term is added to the finite-Stokes solution
space.

Section~\ref{sec:StokesFunctionalsLeibniz} identified the even finite Stokes
functionals.  Their first variation is a contour integral of a finite local
super two-form, and their bounded-variation point derivative vanishes.  The
torsion-subtracted area derivative therefore obeys the graded chain rule.
Consequently, the local Itoyama--Takashino operators satisfy
\begin{equation}
    \widehat{\mathbb L}_\pm f(S)
    =
    f'(S)\widehat{\mathbb L}_\pm S
\end{equation}
for every ordinary function of an even finite Stokes functional.
This gives the product rule: a finite-Stokes zero mode factors through every local left-hand side of the
hierarchy.

The differentiated contact kernels require more than a scalar equality
imposed only after the endpoints coincide.  Section
\ref{sec:ZeroModeDressing} uses the geodesic wire of the finite-width
contact as an internal cut of the selected primary surface.  The two
daughter surfaces are restrictions of one parent map, so locality of the
surface action gives the exact additive identity
\[
    S[C_{st}^{\gamma}]+S[C_{ts}^{\gamma}]=S[C].
\]
The Euler--Lagrange boundary variation makes this common total area a
finite Stokes functional.  Hence the relevant defect obeys
\[
    \Delta_S=0,
    \qquad
    \mathfrak D_A(s)\Delta_S
    =\mathfrak D_A(t)\Delta_S=0,
    \qquad
    P_\alpha^{(s)}Q_{(t)}^\alpha\Delta_S=0,
\]
and the joining defect obeys the same three conditions.  Therefore no
quadratic dressing term and no anomalous linear prefactor can be
generated.

These geometric identities give the exact finite-\(N\) theorem
\ref{thm:ProjectedFiniteNDressing}. If \(W_{n,\Lambda}(C_1,\ldots,C_n)\) solves
the complete ultraviolet-defined Itoyama--Takashino hierarchy and \(S[C]\) is an
even finite-Stokes zero mode whose selected surface branch obeys the exact
geodesic-wire split and join identities, then
\begin{equation}
    \widetilde W_{n,\Lambda}(C_1,\ldots,C_n)
    =
    W_{n,\Lambda}(C_1,\ldots,C_n)
    \exp\!\left[
        \rho\sum_{j=1}^{n}S[C_j]
    \right]
\end{equation}
solves exactly the same hierarchy for every even constant \(\rho\), real or
complex. This includes same-loop splitting, different-loop joining, and all
finite-\(N\) \(SU(N)\) subtraction terms, as an identity at every fixed
ultraviolet scale.

Section~\ref{sec:SuperHodgeDualSurface} constructed the Lorentzian
supersymmetric Hodge-dual surface.  The physical superloop first determines the
invariant boundary curve obtained by integrating its tangent,
\begin{equation}
    d\mathcal C^a=e^a\big|_C,
    \qquad
    \oint_Ce^a=0.
\end{equation}
The master surface problem is then posed for a stationary complexified
Minkowski embedding with this boundary curve as its Dirichlet
boundary. Since the Lorentzian Hodge star satisfies \((*_4)^2=-1\) on two-forms,
the two sectors are complex and are selected by
\begin{equation}
    P_\chi
    =
    \frac12(1+i\chi *_4),
    \qquad
    *_4P_\chi=-i\chi P_\chi,
    \qquad
    \chi=\pm1.
\end{equation}
The associated Hodge-locked auxiliary coordinates merely serve as a
representation of these projectors; the physical data lie in the stationary boundary fluxes and their Stokes
variations.

The stationary parent surface directly supplies finite Stokes area
 derivatives $\Omega_{ab}^{(\chi)}$ in the two Lorentzian Hodge
 sectors. Their supersymmetric dependence is already carried by the
 invariant superspace one-form; no independent chiral field
 strength, superinstanton, or superspace two-form completion is needed.
 On the regular Super--Weierstrass branch, $S_\chi[C]$ is the stationary
 value of the parent functional.  Since the bulk equations hold, its first
 variation is the boundary Stokes term.  Thus it is already a well-defined
 loop functional and no separate integrability assumption is needed. On a
 physical Lorentzian slice, the two branches are chosen as conjugates,
\begin{equation}
    S_-[C]=\bigl(S_+[C]\bigr)^\dagger.
\end{equation}

The local identity is the ordered Jacobi identity.  Torsion subtraction turns it
into the loop-space Bianchi identity.  The Itoyama--Takashino chiral or anti-chiral
 projection, four-dimensional Hodge duality converts this identity
 directly into the corresponding homogeneous loop equation. Therefore the operator algebra gives the required loop equation for the
surface:
\begin{equation}
    \boxed{
    \widehat{\mathbb L}_\pm S_+[C]
    =
    \widehat{\mathbb L}_\pm S_-[C]
    =0.
    }
\end{equation}
For the free-trace BV twistor continuation, the four one-sided germs
are redistributed between the two daughter loops.  The geodesic
cut partitions one parent Super--Weierstrass surface, so the two
daughter actions add exactly to the parent action.  Exact additivity makes the splitting defect vanish, while the
finite-Stokes point identity gives its two endpoint conditions and their
projected derivative.  The joining branch is treated by the analogous
disjoint-union construction.  Hence the SHD functionals satisfy every
condition of the finite-\(N\) theorem by construction.

The two complex branches combine into the real Lorentzian phase
\begin{equation}
    \Phi_{\rm SHD}^{(k)}[C]
    =
    \frac{c}{2}
    \left(
        \Lambda_k^2S_+[C]
        +
        \bar\Lambda_k^2S_-[C]
    \right)
    =
    c\,\Re\!\left[
        \Lambda_k^2S_+[C]
    \right],
\end{equation}
where \(c\), the branch of \(\Lambda_k^2=(\Lambda_k^3)^{2/3}\), and the supersymmetric vacuum
label \(k\) are nonperturbative vacuum data. The exact dressed
solution is therefore
\begin{equation}
    \boxed{
    W_{\Lambda,k}[C]
    =
    W_{\Lambda,0}[C]
    \exp\!\left[
        i\Phi_{\rm SHD}^{(k)}[C]
    \right].
    }
\end{equation}
The hierarchy determines the allowed geometric zero mode, but it does not
determine its coefficient or the selected vacuum branch.

For planar contours this gives a direct physical result without an additional
large-area assumption.  On the explicitly constructed
Virasoro--Weierstrass branch, the Fourier--Hilbert evaluation gives the exact
stationary value
\begin{equation}
    S_+[C_{\rm pl}]=S_-[C_{\rm pl}]
    =2\sqrt2\left(
        A_{\rm pl}[x]
        +\Delta_\eta^{(2)}
        +\Delta_\eta^{(4)}
    \right).
\end{equation}
Thus the scalar value of each Hodge branch is the geometric planar area,
with only the finite nilpotent corrections generated by the boundary
fermions.  For the ordinary static-source rectangle the boundary fermions
are set to zero, so that \(A_{\rm pl}=LT_E\) and
\begin{equation}
    \Phi_{E,\rm SHD}^{(k)}[C_{T_E,L}]
    =\sigma_kLT_E,
    \qquad
    \sigma_k=2\sqrt2\,\kappa_k>0.
\end{equation}
After analytic continuation \(T_E\to iT\),
\begin{equation}
    \Phi_{\rm SHD}^{(k)}[C_{T,L}]
    =-\sigma_kLT,
    \qquad
    W_{\Lambda,k}[C_{T,L}]
    =W_{\Lambda,0}[C_{T,L}]\exp[-i\sigma_kLT].
\end{equation}
The SHD zero mode therefore contributes the exact positive linear static
energy
\begin{equation}
    E_{{\rm SHD},k}(L)=\sigma_kL>0.
\end{equation}
For a nonconstant fermionic boundary, the finite nilpotent factor multiplies
this bosonic area-law contribution without changing its body.  The planar
area law is therefore derived from the exact stationary value of the SHD
zero mode, rather than postulated as a Plateau asymptotic.  What remains
undetermined by the hierarchy is the vacuum coefficient \(\kappa_k\), the
fractional-power branch, and the supersymmetric vacuum label.  The
additive touching-loop continuation and its differentiated sewing are fixed
exactly by the geodesic partition of the parent twistor surface, while
exact compatibility with the finite-\(N\) hierarchy follows from the
zero-mode and dressing theorems.

Section~\ref{sec:SpontaneousStochasticization} established an exact
dynamical corollary.  Because the Euclidean SYM gradient flow is generated
by the action gradient, its homogeneous loop-space image is the exact local
equation-of-motion combination
\begin{equation}
    \mathcal H_{\rm loop}^{(0)}
    =
    \oint_C ds\,
    \left[
        \mathcal T_+(s)\circ\widehat{\mathbb L}_{+,E}(s)
        +
        \mathcal T_-(s)\circ\widehat{\mathbb L}_{-,E}(s)
    \right].
\end{equation}
The analytically continued SHD factor is annihilated by each local
operator at every marked point.  Hence
\begin{equation}
    \mathcal H_{\rm loop}^{(0)}\mathcal Z_{E,k}=0,
    \qquad
    \partial_\tau\mathcal Z_{E,k}=0,
\end{equation}
without an additional factorization assumption.  The SHD state is therefore
an exact fixed point of the zero-noise SYM loop flow.  What is not proved
is its asymptotic stability and dynamical selection.  Assuming that this
fixed point is stable in the physical sector, after gauge zero modes have
been removed and with the long-flow-time and large-volume limits taken before
the external noise is sent to zero, the surviving confining state realizes
\emph{spontaneous stochasticization}.  Thus exact stationarity is a theorem;
stability and selection are the remaining dynamical problem.

At large \(N\), Section~\ref{sec:SHDAsWittenMasterField} interpreted the
same object as a loop-space master-surface functional.  It is neither a fluctuating
colored gauge field nor a single contour-independent surface.  It is the map
\begin{equation}
    C
    \longmapsto
    \mathcal C_C
    \longmapsto
    Y_C
    \longmapsto
    \bigl(S_+[C],S_-[C]\bigr)
    \longmapsto
    \Phi_{\rm SHD}^{(k)}[C].
\end{equation}
This master surface describes the rigid confining vacuum geometry carried by
the exact zero mode.  The undressed factor \(W_{\Lambda,0}[C]\) contains the
remaining gauge dynamics.  Additional surface variables are needed to describe
the complete glueball and gluinoball spectrum.

POOC, the finite-\(N\) dressing theorem, and the Lorentzian SHD surface
fit together into one construction on the BV/supertwistor loop class used in
this paper.  The parent variational principle defines $S_\pm[C]$ as stationary scalar
functionals.  Their finite Stokes variations follow from the boundary term,
so no separate integrability assumption is needed.  At a contact, the free
BV twistor data reproduce two daughter loops with independent endpoint
traces and separate closure conditions; this gives the exact additive surface at the self-contact.  Supplementary
Material, Sec.~S7 proves the differentiated geodesic-wire sewing identity.  The Fourier--Hilbert calculation in
Subsection~\ref{subsec:SHDPlanarSuperarea} likewise gives the planar
superarea and rectangular area law exactly.

What remains undetermined by the loop hierarchy is the physical vacuum
data: the coefficient $\kappa_k$ (equivalently $c$), the branch of
$\Lambda_k^2$, the supersymmetric vacuum label $k$, and the undressed
solution $W_{\Lambda,0}$.  The gradient-flow fixed-point statement requires
no additional operator or factorization hypothesis: exact stationarity follows
from the equation-of-motion form of the SYM gradient flow and the local SHD
zero-mode identities.  The unproved dynamical issue is the asymptotic
stability and selection of this fixed point in the physical zero-noise limit.
Thus the Lorentzian supersymmetric Hodge-dual surface gives an exact
confining finite-Stokes dressing of the complete ultraviolet-defined planar
and finite-$N$ superloop hierarchies and, after Euclidean continuation, an
exact fixed point of the zero-noise SYM loop flow.  Vacuum selection
and fixed-point stability remain the physical inputs; no additional integrability assumption, external theorem about Plateau
surfaces or bridges, planar-area-law assumption, or loop-flow factorization
is required.

\appendix

\section{Relation of POOC to existing loop calculus}
\label{app:POOCExistingLoopCalculus}

The path-ordered operator calculus introduced in this paper has clear
predecessors in the mathematical theory of loops.  There is a substantial earlier theory of
loops, holonomies, iterated integrals, endpoint derivatives, and area
derivatives.  This Appendix explains that connection and what POOC adds for the
Makeenko--Migdal and Itoyama--Takashino equations.  A longer comparison is
given in Supplementary Material, Sec.~S11.

Chen's iterated integrals provide the natural algebraic language for
path-ordered information carried by a curve \cite{Chen1973}.  The generalized
loop calculus of Tavares uses this algebra to construct generalized loops and
to define endpoint, area, and variational derivatives, together with a
natural diffeomorphism action \cite{Tavares1994LoopCalculus}.  The extended
loop group of Di Bartolo, Gambini, and Griego gives nonparametric loop
coordinates transforming linearly under diffeomorphisms and enlarges the
ordinary loop group to a locally Lie-like infinite-dimensional object
\cite{DiBartoloGambiniGriego1993ExtendedLoopGroup,
DiBartoloGambiniGriego1995ExtendedRepresentation}.  Reiris and Spallanzani
developed a differentiable calculus on the path bundle and identified the
loop derivative as the curvature of a universal connection, so that its
Bianchi identity becomes an ordinary curvature Bianchi identity
\cite{ReirisSpallanzani1999LoopCalculus}.  The distinction between loops,
retracing classes, and holonomy classes was further clarified in
Ref.~\cite{Spallanzani2001LoopsHoops}.  Finally, the signature theory of
Hambly and Lyons gives a rigorous global algebra for bounded-variation paths
and identifies their iterated-integral data modulo tree-like equivalence
\cite{HamblyLyons2010Signature}.  These constructions, together with the
physics loop calculus reviewed in Ref.~\cite{GambiniPullin1996_LoopsKnots},
establish much of the global mathematical background of POOC.

The local problem posed by the loop equations is, however, more specific.
The physical argument of a Wilson functional is an oriented BV loop modulo
orientation-preserving reparametrization,
\begin{equation}
    \mathcal L_{\rm BV}
    =
    \mathcal P_{\rm BV}/\operatorname{Diff}^{+}_{\rm BV}(S^1),
    \label{eq:AppendixPhysicalBVLoopSpace}
\end{equation}
where the quotient removes only the numerical parametrization.  It retains
cyclic order, orientation, winding multiplicity, and the distinct preimages
of a self-intersection.  We do not quotient by all tree-like or holonomy
equivalences before the contact operation is defined, because the
Makeenko--Migdal hierarchy uses precisely these ordered preimages when it
splits or joins loops.

The new variable is the invariant tangent density.  If
\(s=\varphi(u)\), then
\begin{equation}
    \widetilde E^A(u)
    =\varphi'(u)E^A(\varphi(u)),
\end{equation}
whereas the functional delta-function supplies the inverse Jacobian.  Hence
\begin{equation}
    \boxed{
    \frac{\delta}{\delta\widetilde E^A(u)}
    =
    \frac{\delta}{\delta E^A(\varphi(u))}
    }
    \label{eq:AppendixWeightZeroDerivative}
\end{equation}
and the POOC tangent derivative has weight zero.  This differs from the
coordinate path derivative, which is a density of weight one.  The diagonal
of its second derivative is therefore a distribution with nonzero reparametrization weight rather
than a finite scalar on \(\mathcal L_{\rm BV}\).  This explains
why the old coordinate loop Laplacian mixed the genuine ultraviolet contact
with a kinematical equal-point singularity.

The second extension is the use of the two one-sided traces of a BV tangent
as ordered operator slots.  POOC forms
\begin{equation}
    [\delta_A,\delta_B\}_o(s)
    =
    \delta_A(s-0)\delta_B(s+0)
    -(-1)^{|A||B|}
     \delta_B(s-0)\delta_A(s+0)
    \label{eq:AppendixOrderedBVCommutator}
\end{equation}
before the ordered points are brought together.  Thus the operation is
finite and scale-free.  In superspace the full commutator also contains the
known flat torsion, which must be removed:
\begin{equation}
    \delta^\Sigma_{AB}
    =
    T_{AB}{}^C\delta_C
    -[\delta_A,\delta_B\}_o.
    \label{eq:AppendixTorsionSubtractedAreaDerivative}
\end{equation}
The operator lift then gives the exact dictionary
\begin{equation}
\begin{aligned}
    \delta_A
    &\longleftrightarrow \nabla_A,
    \\
    \delta^\Sigma_{AB}
    &\longleftrightarrow F_{AB},
    \\
    \mathfrak D_C^{\,o}\delta^\Sigma_{AB}
    &\longleftrightarrow \mathcal D_CF_{AB}.
\end{aligned}
    \label{eq:AppendixPOOCDictionary}
\end{equation}
The three-slot ordered Jacobi identity is consequently the torsionful
superspace Bianchi identity.  This is closely related to the earlier
identification of the loop derivative with curvature, but here it is
implemented directly in the local BV operator algebra required by the
superloop equation.

This extra local structure is what connects the Hodge surface to the gauge
theory.  The parent surface gives a finite Stokes area derivative
\(\Omega^{(\chi)}\) in one Lorentzian Hodge sector.  Hodge duality alone is
a statement about a surface.  To conclude that its exponential is a zero
mode of the loop equation, one must identify its area derivative with a
curvature insertion and its next point derivative with the covariant
Bianchi combination.  POOC supplies precisely these two identifications:
\begin{equation}
    *_4\Omega^{(\chi)}=-i\chi\Omega^{(\chi)}
    \quad+\quad
    \text{ordered Jacobi identity}
    \quad\Longrightarrow\quad
    \widehat{\mathbb L}_\pm S_\chi=0.
    \label{eq:AppendixHodgeToLoopEquation}
\end{equation}
Without this ordered curvature calculus there would remain a gap between
the geometric Hodge surface and the local differential operator of the
loop equation.

Finally, the quantum hierarchy contains a second structure which is not part
of the homogeneous curvature calculus: the physical ultraviolet contact.
POOC leaves that distribution intact and defines it with a finite-width
superspace kernel and short gauge wires.  The same geodesic wire is an internal cut of the selected primary
surface.  Locality of the parent
surface action then gives the exact domain identity
\begin{equation}
\begin{aligned}
    S[C_{st}^{\gamma}]+S[C_{ts}^{\gamma}]
    &=
    \int_{\mathcal D_{st}^{\gamma}}\!\mathcal L
    +
    \int_{\mathcal D_{ts}^{\gamma}}\!\mathcal L
    \\
    &=
    \int_{\mathcal D_C}\!\mathcal L
    =S[C].
\end{aligned}
    \label{eq:AppendixExactGeodesicAreaPartition}
\end{equation}
The common wire has zero two-dimensional measure and is traversed with
opposite orientations in the two daughter boundaries.  Equation
\eqref{eq:AppendixExactGeodesicAreaPartition} proves exact equality of the
three scalar surface values, and hence exact cancellation of their
exponential area factors.  By itself, however, scalar additivity is not yet
the statement required by the differentiated contact kernel.

The missing differential statement follows from the finite-Stokes property
of the \emph{total} parent area.  Subsection
\ref{subsec:BoundaryVariationAndHodgeAreaDerivative} derives the boundary
variation from the Euler--Lagrange equation,
\begin{equation}
    \delta S[C]
    =\oint_C ds\,\xi^A(s)E^B(s)\Omega_{BA}[C;s],
    \label{eq:AppendixParentAreaStokesVariation}
\end{equation}
and Subsection~\ref{subsec:SHDStokesIntegrability} identifies the stationary
SHD area as an even finite Stokes functional.  Proposition
\ref{prop:MainFiniteStokesChainRules} therefore applies at every marked
point of the parent boundary:
\begin{equation}
    \boxed{
    \mathfrak D_A(u)S[C]
    =
    \frac{\delta S[C]}{\delta E^A(u-0)}
    -
    \frac{\delta S[C]}{\delta E^A(u+0)}
    =0,
    \qquad u=s,t.
    }
    \label{eq:AppendixParentAreaPointDerivatives}
\end{equation}
Geometrically, a deformation supported at one boundary point sweeps zero
area in the shrinking-interval limit: the Stokes variation
\eqref{eq:AppendixParentAreaStokesVariation} is an integral of a finite
local density and vanishes with the interval, whereas a nonzero jump of the
conjugate tangent derivative would remain finite.  This is the BV endpoint
argument of Subsection~\ref{subsec:PointDerivativeAndClassicalSMMOperator}
and Supplementary Material, Sec.~S5.

Since Eq.~\eqref{eq:AppendixExactGeodesicAreaPartition} is an exact
functional identity for the endpoint-dependent geodesic cut, the point
derivatives appearing in the contact term act on the sum of the two daughter
areas as on the single parent Stokes functional:
\begin{equation}
\begin{aligned}
    \mathfrak D_A(s)
    \bigl(S[C_{st}^{\gamma}]+S[C_{ts}^{\gamma}]\bigr)
    &=\mathfrak D_A(s)S[C]=0,
    \\
    \mathfrak D_A(t)
    \bigl(S[C_{st}^{\gamma}]+S[C_{ts}^{\gamma}]\bigr)
    &=\mathfrak D_A(t)S[C]=0.
\end{aligned}
    \label{eq:AppendixDaughterSumPointDerivatives}
\end{equation}
Subtracting the parent contribution from
\eqref{eq:AppendixDaughterSumPointDerivatives} gives
\begin{equation}
    \mathfrak D_A(s)\Delta_S
    =\mathfrak D_A(t)\Delta_S
    =0.
    \label{eq:AppendixDefectPointDerivativeConditions}
\end{equation}
Because these are functional identities on the endpoint-dependent
geodesic branch, applying the complementary endpoint operator selected by
the differentiated contact gives
\begin{equation}
    P_\alpha^{(s)}Q_{(t)}^\alpha\Delta_S=0,
    \label{eq:AppendixDefectProjectedCondition}
\end{equation}
with the anti-chiral conjugate.  Together with \(\Delta_S=0\), these are
precisely the three conditions which make the quadratic and mixed terms in
the differentiated exponential vanish.  This finite-Stokes
argument---rather than scalar additivity alone---is the reason why the
geodesically partitioned area dressing passes through the ultraviolet
splitting and joining terms.  It is the additional ingredient needed to
extend a homogeneous zero mode to an exact multiplicative solution of the
complete finite-\(N\) hierarchy.

The new content should therefore be stated without erasing the earlier
mathematics.  Endpoint derivatives, area derivatives, generalized loops,
and curvature interpretations already existed.  POOC is the synthesis and
extension needed here: a reparametrization-invariant local calculus of
weight-zero tangent derivatives on BV loops, with circle ordering,
superspace torsion subtraction, and UV-wire sewing.  These additions make
the calculus suitable for deriving the SHD surface zero mode and proving its
compatibility with the exact loop equations.

\section*{Declaration of generative AI and AI-assisted technologies in the writing process}
GPT 5.6 Pro and GPT 5.6 Sol Codex were used for algebraic derivations,
code generation, and exact symbolic checks.  Gemini 3.1 Pro,
ChatGPT 5.5, and Claude Opus 4.8 were used for copy-editing,
presentation, and referee-style review. The author has
reviewed and edited the manuscript and takes full responsibility for its content.
\section*{Code and data availability}

No empirical data were generated or analyzed in this work. The
\Mathematica{} notebooks used to verify the algebraic identities and
component calculations are provided as ancillary files in the
Supplementary Material accompanying this manuscript.

\bibliographystyle{plainnat}
\bibliography{bibliography}

@article{MMEq79,
  author  = {Makeenko, Yu. M. and Migdal, A. A.},
  title   = {Exact equation for the loop average in multicolor QCD},
  journal = {Physics Letters B},
  volume  = {88},
  number  = {1},
  pages   = {135--137},
  year    = {1979},
  issn    = {0370-2693},
  doi     = {10.1016/0370-2693(79)90131-3},
  url     = {https://doi.org/10.1016/0370-2693(79)90131-3}
}

@article{MM1981NPB,
  author       = {Yuri M. Makeenko and Alexander A. Migdal},
  title        = {Quantum Chromodynamics as Dynamics of Loops},
  journal      = {Nuclear Physics B},
  volume       = {188},
  pages        = {269--316},
  year         = {1981},
  doi          = {10.1016/0550-3213(81)90105-2},
  publisher    = {Elsevier},
}

@article{migdal2025SQYMflow,
title = {Spontaneous quantization of the Yang–Mills gradient flow},
journal = {Nuclear Physics B},
volume = {1020},
pages = {117129},
year = {2025},
issn = {0550-3213},
doi = {https://doi.org/10.1016/j.nuclphysb.2025.117129},
url = {https://www.sciencedirect.com/science/article/pii/S0550321325003384},
author = {Alexander Migdal}
}

@article{Mig83,
  author  = {A. A. Migdal},
  title   = {Loop equations and $1/N$ expansion},
  journal = {Physics Reports},
  volume  = {102},
  number  = {4},
  pages   = {199--290},
  year    = {1983},
  doi     = {10.1016/0370-1573(83)90076-5}
}

@article{Mig98Hidden,
    author       = {Migdal, Alexander A.},
  title        = {Hidden symmetries of large {N} {QCD}},
  journal      = {Prog. Theor. Phys. Suppl.},
  volume       = {131},
  pages        = {269--307},
  year         = {1998},
  doi          = {10.1143/PTPS.131.269},
  eprint       = {hep-th/9610126},
  archivePrefix= {arXiv},
  primaryClass = {hep-th},
  note         = {Preprint arXiv:hep-th/9610126 (October 1996)}
}

@article{migdal2023exact, 
    title={To the Theory of Decaying Turbulence},
    volume={7}, 
    ISSN={2504-3110}, 
    url={http://dx.doi.org/10.3390/fractalfract7100754}, 
    DOI={10.3390/fractalfract7100754}, 
    number={10}, 
    journal={Fractal and Fractional}, 
    publisher={MDPI AG}, 
    author={Migdal, Alexander}, 
    year={2023}, 
    month={Oct}, 
    pages={754},
    eprint={2304.13719},
    archivePrefix={arXiv},
    primaryClass={physics.flu-dyn}
}

@article{migdal2024quantum,
  author       = {Alexander Migdal},
  title        = {Quantum solution of classical turbulence: Decaying energy spectrum},
  journal      = {Physics of Fluids},
  volume       = {36},
  number       = {9},
  pages        = {095161},
  year         = {2024},
  doi          = {10.1063/5.0228660},
  publisher    = {AIP Publishing}
}

@article{Migdal_1986_Stochastic,
doi = {10.1070/PU1986v029n05ABEH003373},
url = {https://dx.doi.org/10.1070/PU1986v029n05ABEH003373},
year = {1986},
month = {may},
publisher = {},
volume = {29},
number = {5},
pages = {389},
author = {A A Migdal},
title = {Stochastic quantization of field theory},
journal = {Soviet Physics Uspekhi}
}

@article{Chen1973,
  author  = {Chen, Kuo-Tsai},
  title   = {Iterated integrals of differential forms and loop space},
  journal = {Annals of Mathematics},
  year    = {1973},
  volume  = {97},
  number  = {2},
  pages   = {217--246}
}

@book{GambiniPullin1996_LoopsKnots,
  author    = {Gambini, Rodolfo and Pullin, Jorge},
  title     = {Loops, Knots, Gauge Theories and Quantum Gravity},
  publisher = {Cambridge University Press},
  series    = {Cambridge Monographs on Mathematical Physics},
  year      = {1996},
  address   = {Cambridge, UK},
  isbn      = {9780521473320}
}

@misc{ReviewPaperAM,
  author       = {Migdal, Alexander},
  title        = {Geometric Solution of Turbulence as Diffusion in Loop Space},
  year         = {2026},
  doi          = {10.1098/rsta.2025.0032},
  howpublished = {Forthcoming in Philosophical Transactions of the Royal Society A, Special Issue "Frontiers of Turbulence and Statistical Physics Meet"},
   note         = { Accepted},
  eprint       = {https://doi.org/10.48550/arXiv.2511.02165},
  archivePrefix = {arXiv},
  primaryClass = {physics.flu-dyn}
}

@article{Migdal2026GeometricQCDII,
title = {Geometric QCD II: The confining twistor string and meson spectrum},
author={Alexander Migdal},
journal = {Nuclear Physics B},
volume = {1026},
pages = {117424},
year = {2026},
issn = {0550-3213},
doi = {https://doi.org/10.1016/j.nuclphysb.2026.117424},
url = {https://www.sciencedirect.com/science/article/pii/S055032132600132X}
}

@article{Migdal2026GeometricQCDIII,
title = {Geometric QCD III: Exact transition amplitudes and the glueball spectrum},
author={Alexander Migdal},
 journal = {arXiv preprint},
  year = {2026},
  eprint = {2605.02373},
  archivePrefix = {arXiv},
  primaryClass = {physics.hep-th},
  year = {2026}
}

@misc{migdal2025geometric,
title = {Geometric QCD I: the Hodge-dual surface and quark confinement},
author = {Alexander Migdal},
journal = {Nuclear Physics B},
volume = {1025},
pages = {117380},
year = {2026},
issn = {0550-3213},
doi = {https://doi.org/10.1016/j.nuclphysb.2026.117380},
url = {https://www.sciencedirect.com/science/article/pii/S055032132600088X}
}

@inproceedings{Witten:1980ez,
    author = "Witten, Edward",
    title = "{The $1/N$ Expansion in Atomic and Particle Physics}",
    booktitle = "{Recent Developments in Gauge Theories. Proceedings, Nato Advanced Study Institute, Cargese, France, August 26 - September 8, 1979}",
    editor = "'t Hooft, G. and others",
    publisher = "Plenum Press",
    address = "New York",
    year = "1980",
    pages = "403--419",
    doi = "10.1007/978-1-4684-7571-5_21"
}

@misc{migdal2026EulerStability,
      title={Euler Ensemble as Decaying Turbulence Attractor: Universality, Stability and Parity Classes}, 
      author={Alexander Migdal},
      year={2026},
      eprint={2607.05745},
      archivePrefix={arXiv},
      primaryClass={nlin.CD},
      url={https://arxiv.org/abs/2607.05745}, 
}

@book{WessBagger,
  author    = {Wess, Julius and Bagger, Jonathan},
  title     = {Supersymmetry and Supergravity},
  edition   = {2},
  publisher = {Princeton University Press},
  address   = {Princeton, NJ},
  year      = {1992},
  isbn      = {9780691025308}
}

@book{GatesGrisaruRocekSiegel,
  author    = {Gates, S. James and Grisaru, Marcus T. and Rocek, Martin and Siegel, Warren},
  title     = {Superspace, or One Thousand and One Lessons in Supersymmetry},
  publisher = {Benjamin/Cummings},
  address   = {Reading, MA},
  year      = {1983},
  eprint    = {hep-th/0108200},
  archivePrefix = {arXiv},
  primaryClass  = {hep-th}
}

@article{ADHM,
  author  = {Atiyah, Michael F. and Drinfeld, Vladimir G. and Hitchin, Nigel J. and Manin, Yuri I.},
  title   = {Construction of Instantons},
  journal = {Physics Letters A},
  volume  = {65},
  number  = {3},
  pages   = {185--187},
  year    = {1978},
  doi     = {10.1016/0375-9601(78)90141-X}
}

@article{Osborn1979,
  author  = {Osborn, Hugh},
  title   = {Semiclassical Functional Integrals for Self-Dual Gauge Fields},
  journal = {Annals of Physics},
  volume  = {135},
  number  = {2},
  pages   = {373--408},
  year    = {1981},
  doi     = {10.1016/0003-4916(81)90123-7}
}

@article{DoreyKhozeMattis2002,
  author        = {Dorey, Nick and Hollowood, Timothy J. and Khoze, Valentin V. and Mattis, Michael P.},
  title         = {The Calculus of Many Instantons},
  journal       = {Physics Reports},
  volume        = {371},
  number        = {4--5},
  pages         = {231--459},
  year          = {2002},
  doi           = {10.1016/S0370-1573(02)00343-1},
  eprint        = {hep-th/0206063},
  archivePrefix = {arXiv},
  primaryClass  = {hep-th}
}

@article{KadohUkita2018,
  author        = {Kadoh, Daisuke and Ukita, Naoya},
  title         = {Supersymmetric Gradient Flow in \(N=1\) SYM},
  journal       = {European Physical Journal C},
  volume        = {79},
  pages         = {879},
  year          = {2019},
  doi           = {10.1140/epjc/s10052-019-7375-1},
  eprint        = {1812.02351},
  archivePrefix = {arXiv},
  primaryClass  = {hep-th}
}

@article{KikuchiOnogi2014,
  author        = {Kikuchi, Kenji and Onogi, Tetsuya},
  title         = {Generalized Gradient Flow Equation and Its Application to Super Yang--Mills Theory},
  journal       = {Journal of High Energy Physics},
  volume        = {2014},
  number        = {11},
  pages         = {094},
  year          = {2014},
  doi           = {10.1007/JHEP11(2014)094},
  eprint        = {1408.2185},
  archivePrefix = {arXiv},
  primaryClass  = {hep-th}
}

@article{DelDebbioPatellaRago2013,
  author        = {Del Debbio, Luigi and Patella, Agostino and Rago, Antonio},
  title         = {Space-Time Symmetries and the Yang--Mills Gradient Flow},
  journal       = {Journal of High Energy Physics},
  volume        = {2013},
  number        = {11},
  pages         = {212},
  year          = {2013},
  doi           = {10.1007/JHEP11(2013)212},
  eprint        = {1306.1173},
  archivePrefix = {arXiv},
  primaryClass  = {hep-th}
}

@article{Luscher2010WilsonFlow,
  author        = {L{\"u}scher, Martin},
  title         = {Properties and Uses of the Wilson Flow in Lattice QCD},
  journal       = {Journal of High Energy Physics},
  volume        = {2010},
  number        = {8},
  pages         = {071},
  year          = {2010},
  doi           = {10.1007/JHEP08(2010)071},
  eprint        = {1006.4518},
  archivePrefix = {arXiv},
  primaryClass  = {hep-lat}
}

@article{Luscher2013Chiral,
  author        = {L{\"u}scher, Martin},
  title         = {Chiral Symmetry and the Yang--Mills Gradient Flow},
  journal       = {Journal of High Energy Physics},
  volume        = {2013},
  number        = {4},
  pages         = {123},
  year          = {2013},
  doi           = {10.1007/JHEP04(2013)123},
  eprint        = {1302.5246},
  archivePrefix = {arXiv},
  primaryClass  = {hep-lat}
}

@article{ParisiWu1981,
  author  = {Parisi, Giorgio and Wu, Yong-Shi},
  title   = {Perturbation Theory Without Gauge Fixing},
  journal = {Scientia Sinica},
  volume  = {24},
  pages   = {483},
  year    = {1981}
}

@article{DamgaardHueffel1987,
  author  = {Damgaard, Poul H. and H{\"u}ffel, Helmuth},
  title   = {Stochastic Quantization},
  journal = {Physics Reports},
  volume  = {152},
  number  = {5--6},
  pages   = {227--398},
  year    = {1987},
  doi     = {10.1016/0370-1573(87)90144-X}
}

@article{Struwe1994YangMillsFlow,
  author  = {Struwe, Michael},
  title   = {The Yang--Mills Flow in Four Dimensions},
  journal = {Calculus of Variations and Partial Differential Equations},
  volume  = {2},
  pages   = {123--150},
  year    = {1994},
  doi     = {10.1007/BF01234317}
}

@article{Rade1992YangMillsHeat,
  author  = {R{\aa}de, Johan},
  title   = {On the Yang--Mills Heat Equation in Two and Three Dimensions},
  journal = {Journal f{\"u}r die reine und angewandte Mathematik},
  volume  = {431},
  pages   = {123--163},
  year    = {1992},
  doi     = {10.1515/crll.1992.431.123}
}

@article{Donaldson1985ASD,
  author  = {Donaldson, Simon K.},
  title   = {Anti Self-Dual Yang--Mills Connections over Complex Algebraic Surfaces and Stable Vector Bundles},
  journal = {Proceedings of the London Mathematical Society},
  volume  = {50},
  number  = {1},
  pages   = {1--26},
  year    = {1985},
  doi     = {10.1112/plms/s3-50.1.1}
}

@article{ItoyamaTakashino1997,
  author        = {Itoyama, Hiroshi and Takashino, Hiroyuki},
  title         = {Schwinger--Dyson Equation for Supersymmetric Yang--Mills Theory:
                   Manifestly Supersymmetric Form},
  journal       = {Progress of Theoretical Physics},
  volume        = {97},
  number        = {6},
  pages         = {963--1001},
  year          = {1997},
  doi           = {10.1143/PTP.97.963},
  eprint        = {hep-th/9703115},
  archivePrefix = {arXiv},
  primaryClass  = {hep-th},
  url           = {https://arxiv.org/abs/hep-th/9703115}
}

@article{ItoyamaTakashino1996,
  author        = {Itoyama, H. and Takashino, H.},
  title         = {Schwinger--Dyson and Large $N_c$ Loop Equation for
                   Supersymmetric Yang--Mills Theory},
  journal       = {Physics Letters B},
  volume        = {381},
  pages         = {163--168},
  year          = {1996},
  doi           = {10.1016/0370-2693(96)00579-5},
  eprint        = {hep-th/9604023},
  archivePrefix = {arXiv},
  primaryClass  = {hep-th}
}

@article{ViswanathanParthasarathy1992SuperGauss,
  author  = {Viswanathan, K. S. and Parthasarathy, R.},
  title   = {Extrinsic Geometry of World-Sheet Supersymmetry through Generalized Super-Gauss Maps},
  journal = {International Journal of Modern Physics A},
  volume  = {7},
  number  = {24},
  pages   = {5995--6011},
  year    = {1992},
  doi     = {10.1142/S0217751X92002714}
}

@article{BandosEtAl1995DoublySupersymmetric,
  author        = {Bandos, I. and Pasti, P. and Sorokin, D. and Tonin, M. and Volkov, D.},
  title         = {Superstrings and Supermembranes in the Doubly Supersymmetric Geometrical Approach},
  journal       = {Nuclear Physics B},
  volume        = {446},
  pages         = {79--118},
  year          = {1995},
  doi           = {10.1016/0550-3213(95)00267-V},
  eprint        = {hep-th/9501113},
  archivePrefix = {arXiv}
}

@article{BertrandGrundlandHariton2015SuperSurfaces,
  author        = {Bertrand, S. and Grundland, A. M. and Hariton, A. J.},
  title         = {Supersymmetric Version of the Equations of Conformally Parametrized Surfaces},
  journal       = {Journal of Physics A: Mathematical and Theoretical},
  volume        = {48},
  number        = {17},
  pages         = {175208},
  year          = {2015},
  doi           = {10.1088/1751-8113/48/17/175208},
  eprint        = {1412.4697},
  archivePrefix = {arXiv},
  primaryClass  = {math-ph}
}

@article{Tavares1994LoopCalculus,
  author        = {Tavares, J. N.},
  title         = {Chen Integrals, Generalized Loops and Loop Calculus},
  journal       = {International Journal of Modern Physics A},
  volume        = {9},
  number        = {26},
  pages         = {4511--4548},
  year          = {1994},
  doi           = {10.1142/S0217751X94001795},
  eprint        = {hep-th/9305173},
  archivePrefix = {arXiv}
}

@article{DiBartoloGambiniGriego1993ExtendedLoopGroup,
  author        = {Di Bartolo, Cayetano and Gambini, Rodolfo and Griego, Jorge},
  title         = {The Extended Loop Group: An Infinite Dimensional Manifold Associated with the Loop Space},
  journal       = {Communications in Mathematical Physics},
  volume        = {158},
  pages         = {217--240},
  year          = {1993},
  doi           = {10.1007/BF02108073},
  eprint        = {gr-qc/9303010},
  archivePrefix = {arXiv}
}

@article{DiBartoloGambiniGriego1995ExtendedRepresentation,
  author        = {Di Bartolo, Cayetano and Gambini, Rodolfo and Griego, Jorge},
  title         = {Extended Loop Representation of Quantum Gravity},
  journal       = {Physical Review D},
  volume        = {51},
  pages         = {502--512},
  year          = {1995},
  doi           = {10.1103/PhysRevD.51.502},
  eprint        = {gr-qc/9406039},
  archivePrefix = {arXiv}
}

@article{ReirisSpallanzani1999LoopCalculus,
  author        = {Reiris, Martin and Spallanzani, Pablo},
  title         = {A Calculus in Differentiable Spaces and Its Application to Loops},
  journal       = {Classical and Quantum Gravity},
  volume        = {16},
  number        = {8},
  pages         = {2697--2708},
  year          = {1999},
  doi           = {10.1088/0264-9381/16/8/309},
  eprint        = {math/9802080},
  archivePrefix = {arXiv},
  primaryClass  = {math.DG}
}

@article{Spallanzani2001LoopsHoops,
  author        = {Spallanzani, Pablo},
  title         = {Groups of Loops and Hoops},
  journal       = {Communications in Mathematical Physics},
  volume        = {216},
  pages         = {243--253},
  year          = {2001},
  doi           = {10.1007/s002200000284},
  eprint        = {math/9908098},
  archivePrefix = {arXiv},
  primaryClass  = {math.DG}
}

@article{HamblyLyons2010Signature,
  author        = {Hambly, Ben and Lyons, Terry},
  title         = {Uniqueness for the Signature of a Path of Bounded Variation and the Reduced Path Group},
  journal       = {Annals of Mathematics},
  volume        = {171},
  number        = {1},
  pages         = {109--167},
  year          = {2010},
  doi           = {10.4007/annals.2010.171.109},
  eprint        = {math/0507536},
  archivePrefix = {arXiv}
}
\end{document}